\documentclass[acmsmall, nonacm]{acmart}

\AtBeginDocument{%
  }

\usepackage{setspace}
\usepackage{graphicx}
\usepackage{xspace}
\xspaceaddexceptions{(}
\usepackage{multirow}
\usepackage{array}
\usepackage{appendix}
\usepackage{algorithm}

\usepackage[utf8]{inputenc} 
\usepackage[T1]{fontenc} 
\usepackage{hyperref}     
\usepackage{cleveref}
\usepackage{url}          
\usepackage{amsfonts}      
\usepackage{nicefrac}      
\usepackage{microtype}   

\usepackage{bm}
\usepackage{dsfont}
\usepackage{xcolor}
\usepackage{enumitem}
\usepackage{soul}
\setul{2pt}{.4pt} 
\usepackage[labelformat=simple]{subcaption}

\allowdisplaybreaks

\theoremstyle{plain}
\newtheorem{theorem}{Theorem}

\newtheorem{lemma}{Lemma}
\newtheorem{proposition}{Proposition}

\theoremstyle{definition}
\newtheorem{definition}{Definition}
\newtheorem{assumption}{Assumption}

\theoremstyle{remark}
\newtheorem{remark}{\textit{Remark}}

\theoremstyle{plain}
\newtheorem{innercustomthm}{Theorem}
\newenvironment{customthm}[1]
  {\renewcommand\theinnercustomthm{#1}\innercustomthm}
  {\endinnercustomthm}

\newif\ifcomment
\commentfalse

\newif\iffinalversion

\finalversiontrue

\newcommand{\norm}[1]{\lVert#1\rVert}
\newcommand{\gen}{G}
\newcommand{\E}{\mathbb{E}}
\newcommand{\indi}[1]{\mathds{1}}
\newcommand{\R}{\mathbb{R}}

\newcommand{\Prob}{\mathbb{P}}
\newcommand{\indibrac}[1]{\mathds{1}_{\{#1\}}}
\newcommand{\abs}[1]{\left\lvert#1\right\rvert}
\newcommand{\Thebrac}[1]{\Theta\left(#1\right)}
\newcommand{\obrac}[1]{o\left(#1\right)}
\newcommand{\Obrac}[1]{O\left(#1\right)}

\newcommand{\sysr}[1]{#1^{(r)}}

\newcommand{\ceil}[1]{\left\lceil#1\right\rceil}
\newcommand{\metapolicyfull}{\textsc{Join-Requesting-Server}}
\newcommand{\metapolicy}{JRS}
\newcommand{\ve}[1]{\bm{#1}}

\newcommand{\Kinner}{\mathcal{K}}
\newcommand{\Kouter}{\mathcal{K}}
\newcommand{\Kka}{E}
\newcommand{\jobspace}{\mathcal{I}} 
\newcommand{\Kmax}{K_{\max}} 
\newcommand{\Arr}{\Lambda}
\newcommand{\arr}{\lambda}
\newcommand{\ser}{\mu}
\newcommand{\vzero}{\ve{0}}
\newcommand{\vei}{\ve{e}}
\newcommand{\budget}{\epsilon}
\newcommand{\low}{L}
\newcommand{\high}{H}
\newcommand{\vek}{{\ve{k}}}
\newcommand{\veK}{{\ve{K}}}

\newcommand{\sysbar}[1]{\overline{#1}}
\newcommand{\reqrate}{\sysbar{\lambda}}

\newcommand{\isp}{\mathcal{P}}
\newcommand{\ssp}{\overline{\mathcal{P}}}
\newcommand{\sslp}{\overline{\mathcal{LP}}}
\newcommand{\policy}{\sigma}
\newcommand{\nact}{N}
\newcommand{\costp}{C}
\newcommand{\nactsmall}{n}

\newcommand{\vea}{\ve{a}}
\newcommand{\sumrate}[2]{\sum_{(#1) \in E(#2)}}
\newcommand{\allL}{{1:L}}
\newcommand{\sumL}{\sum_{\ell=1}^L}
\newcommand{\toklim}{\eta_{\max}}
\newcommand{\rate}[2]{\gamma_{#2,(#1)}}

\newcommand{\ratetjm}[2]{\widetilde{\gamma}_{#2,(#1)}} 
\newcommand{\ratetottjm}[1]{\widetilde{\gamma}_{#1}}

\newcommand{\lamest}{\lambda}
\newcommand{\muest}{\mu}
\newcommand{\lamtrue}{\widetilde{\lambda}}
\newcommand{\mutrue}{\widetilde{\mu}}

\newcommand{\syshat}[1]{\widehat{#1}}
\newcommand{\ratetot}[1]{\gamma_{#1}}
\newcommand{\vj}{\zeta}
\newcommand{\vevj}{\ve{\zeta}}
 
\newcommand{\vetok}{\ve{\eta}}
\newcommand{\toK}{\eta} 
\newcommand{\vetoK}{\ve{\eta}} 
\newcommand{\tmix}{\tau_{\text{mix}}}

\newcommand{\tvj}{v} 
\newcommand{\tvJ}{V} 
\newcommand{\tof}{y} 
\newcommand{\toF}{Y}
\newcommand{\ttk}{z} 
\newcommand{\ttK}{Z}
\newcommand{\dtvjg}{dv} 
\newcommand{\dtof}{dy}

\newcommand{\vestt}{\ve{u}}
\newcommand{\vestT}{\ve{U}}
\newcommand{\sumka}{\sum_{\vek',\vea}}
\newcommand{\sumkae}{\sum_{(\vek',\vea)\in E(\vek)}}
\newcommand{\trfr}{u}
\newcommand{\servmax}{t_{\max}}

\ccsdesc[500]{Networks~Network performance analysis}
\ccsdesc[500]{Mathematics of computing~Markov processes}

\keywords{stochastic bin-packing, large service systems, policy conversion, asymptotic optimality}

\title{Near-Optimal Stochastic Bin-Packing in Large Service Systems with Time-Varying Item Sizes}
\date{}

\author{Yige Hong}
\email{yigeh@andrew.cmu.edu}
\orcid{0000-0001-8534-1063}
\affiliation{%
  \institution{Carnegie Mellon University}
  \city{Pittsburgh}
  \state{Pennsylvania}
  \country{USA}
  \postcode{15213}
}

\author{Qiaomin Xie}
\email{qiaomin.xie@wisc.edu}
\orcid{0000-0003-2834-6866}
\affiliation{%
  \institution{University of Wisconsin-Madison}
  \city{Madison}
  \state{Wisconsin}
  \country{USA}
  \postcode{53706}
}

\author{Weina Wang}
\email{weinaw@cs.cmu.edu}
\orcid{0000-0001-6808-0156}
\affiliation{%
  \institution{Carnegie Mellon University}
   \city{Pittsburgh}
  \state{Pennsylvania}
  \country{USA}
  \postcode{15213}
}

\graphicspath{{figs/}}

\begin{document}

\begin{abstract}
    In modern computing systems, jobs' resource requirements often vary over time.
Accounting for this temporal variability during job scheduling is essential for meeting performance goals.
However, theoretical understanding on how to schedule jobs with time-varying resource requirements is limited.
Motivated by this gap, we propose a \emph{new setting} of the stochastic bin-packing problem in service systems that allows for \emph{time-varying} job resource requirements, also referred to as `item sizes' in traditional bin-packing terms.
In this setting, a job or `item' must be dispatched to a server or `bin' upon arrival.
Its resource requirement may vary over time while in service, following a Markovian assumption.
Once the job's service is complete, it departs from the system.
Our goal is to minimize the expected number of active servers, or `non-empty bins', in steady state.

Under our problem formulation, we develop a job dispatch policy, named \metapolicyfull\ (\metapolicy).
Broadly, \metapolicy\ lets each server independently evaluate its current job configuration and decide whether to accept additional jobs, balancing the competing objectives of maximizing throughput and minimizing the risk of resource capacity overruns.
The \metapolicy\ dispatcher then utilizes these individual evaluations to decide which server to dispatch each arriving job to.
The theoretical performance guarantee of \metapolicy\ is in the asymptotic regime where the job arrival rate scales large linearly with respect to a scaling factor $r$.
We show that \metapolicy\ achieves an additive optimality gap of $O(\sqrt{r})$ in the objective value, where the optimal objective value is $\Theta(r)$.
When specialized to constant job resource requirements, our result improves upon the state-of-the-art $o(r)$ optimality gap.
Our technical approach highlights a novel policy conversion framework that reduces the policy design problem into a single-server problem.

\end{abstract}

\maketitle

\section{Introduction}\label{sec:intro}
\subsection{Background and Motivation}
In modern computing systems, a job often takes the form of a virtual machine (VM) or a container \cite{GoogleCloudVM_23,ApacheMesosContainer_23}.
Such a job comes with a resource requirement, such as a certain number of CPUs and amount of memory, while in service.
Each server in the system offers a limited amount of these resources.
When a job arrives at the system, the job dispatch policy needs to decide which server the job should be assigned to, given the job's resource requirement and servers' current job configurations.
This job scheduling problem can be approached as a \emph{Stochastic Bin-Packing (SBP) problem}, where jobs are viewed as items, job resource requirements as item sizes, and servers as bins.
A traditional SBP setting considers a finite set of jobs that arrive online but do not depart from the system.
The objective is to minimize the number of servers that have jobs on them, or `non-empty bins', subject to the resource capacities of the servers.
SBP, with a rich history in operations research and theoretical computer science \citep{CouWeb_86,CouWeb_90,CsiJohKen_06}, is a field of continuous developments and advancements \cite{GupRad_20,FreBan_19,AyyDabKha_22}.

To incorporate job \emph{departures} into the problem formulation, a setting referred to as \emph{stochastic bin-packing in service systems} has been proposed recently \citep{Sto_13,StoZho_13, StoZho_15, Sto_17, StoZho_21, GhaZhoSri_14}.
In this setting, jobs not only arrive but also depart over time.
More specifically, jobs are assumed to arrive according to Poisson processes, and each job is assumed to stay in the system for an exponentially distributed service time.
The service time of a job remains unknown until the job departs.
Before delving further into SBP in service systems, it is worth mentioning that there is a parallel thread of research on the so-called dynamic bin-packing problem that also handles job departures (see, e.g., \citep{CofGarJoh_83,LiTanCai_14,BucFaiMel_21}, and references therein), but it is primarily from a worst-case analysis perspective.
Additionally, the virtual machine scheduling problem with objectives different from minimizing the number of active servers has also been widely studied (see, e.g., \cite{MagSriYin_12,MagSri_13,MagSriYin_14,XieDonLu_15,PsyGha_18,PsyGha_19, PsyGha_21,PsyGha_22}).

For SBP in service systems, the goal is to design a job dispatch policy $\policy$ that minimizes the expected number of active servers in \emph{steady state}, denoted as $\nact(\policy)$.
The \emph{optimality gap} of a policy~$\policy$ is defined as $\nact(\policy)-\nact(\policy^*)$, where $\policy^*$ is the optimal policy.
Since SBP in service systems aims to model today's large-scale computing systems, the optimality gap of a policy is usually studied in the regime where the total job arrival rate becomes large.
As we scale up the total job arrival rate linearly with a scaling factor, $r$, the optimal value $\nact(\policy^*)$ can be shown to be $\Theta(r)$.
\footnote{We use the standard Bachmann–Landau notation. Consider two functions $a(r)$ and $b(r)$, 
where $b(r)$ is positive for large enough $r$. Then $a = \Obrac{b}$ if $\limsup_{r\to+\infty}\frac{|a(r)|}{b(r)}<\infty$; $a = \obrac{b}$ if $\lim_{r\to+\infty} \frac{a(r)}{b(r)} = 0$; $a = \Thebrac{b}$ if $a = \Obrac{b}$ and $b = \Obrac{a}$.}
Therefore, we say a policy is asymptotically optimal if its optimality gap is $o(r)$.

The optimality gap for SBP in service systems has been characterized in the line of work \cite{Sto_13,StoZho_13,StoZho_15,Sto_17,StoZho_21}.
In particular, \citet{Sto_13}, \citet{StoZho_13} propose greedy policies that are asymptotically optimal, but the scheduler that executes these policies needs to know detailed state information, which is in a high-dimensional space.
\citet{StoZho_15,StoZho_21} later develop policies that use much less state information and achieve $\Theta(r)$ (with an arbitrarily small constant) and $o(r)$ optimality gaps, respectively.

While prior work on SBP in service systems has provided substantial insights into scheduling virtual-machine-type jobs, it primarily focuses on job resource requirements that remain fixed over time.
However, in modern computing systems, jobs' resource requirements often vary over time \cite{ReiTumGan_12,DelKoz_14,LoCheGov_15,TirBarDen_20,RzaFinSwi_20,BasDenRza_21}.
For example, when a job involves providing user-facing services, the instantaneous requirement on CPUs and memory depends on the service demand, which is subject to fluctuation over time \cite{DelKoz_14,LoCheGov_15}.
Time-varying job resource requirements pose significant challenges in optimizing system efficiency, particularly when aiming to minimize the number of active servers, thereby improving server utilization.
It is pertinent to note that low utilization has been recognized as a significant obstacle to the continued scaling of today's computing systems.

Motivated by this gap, in this paper, we propose a \emph{new setting of stochastic bin-packing in service systems} that allows job resource requirements, or `item sizes', to vary over time.

\subsection{Problem Formulation: A Simplified Version}
\begin{figure}
    \centering
    \begin{minipage}{0.47\textwidth}
    \begin{subfigure}[b]{1\textwidth}
        \centering
        \includegraphics[height=3.25cm]{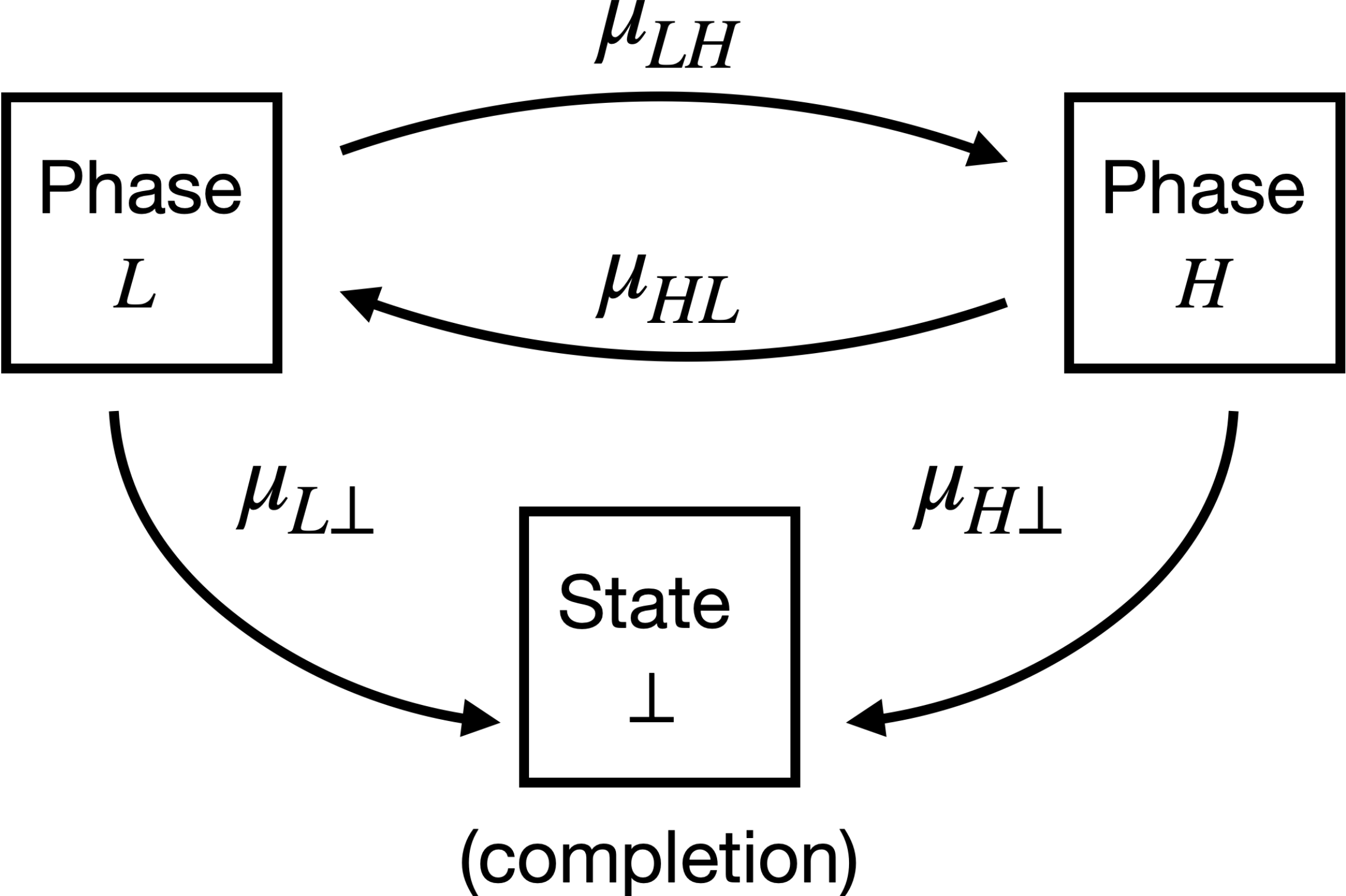}
        \caption{A simplified version of our job model. Each job in service is in either an $\low$ phase or an $\high$ phase, associated with low and high resource requirements, respectively.  When the job is completed, it is said to be in the state $\perp$.  The job transitions between the two phases while in service until it is completed, following a continuous-time Markov chain with rates $\mu_{ii'}$, $i,i'\in\{\low,\high,\perp\}$.
        }
        \label{fig:job-model}
    \end{subfigure}
    \end{minipage}
    \quad
    \begin{minipage}{0.49\textwidth}
    \begin{subfigure}[b]{1\textwidth}
        \centering
        \includegraphics[height=3.55cm, clip, trim = 0 0 0 0]{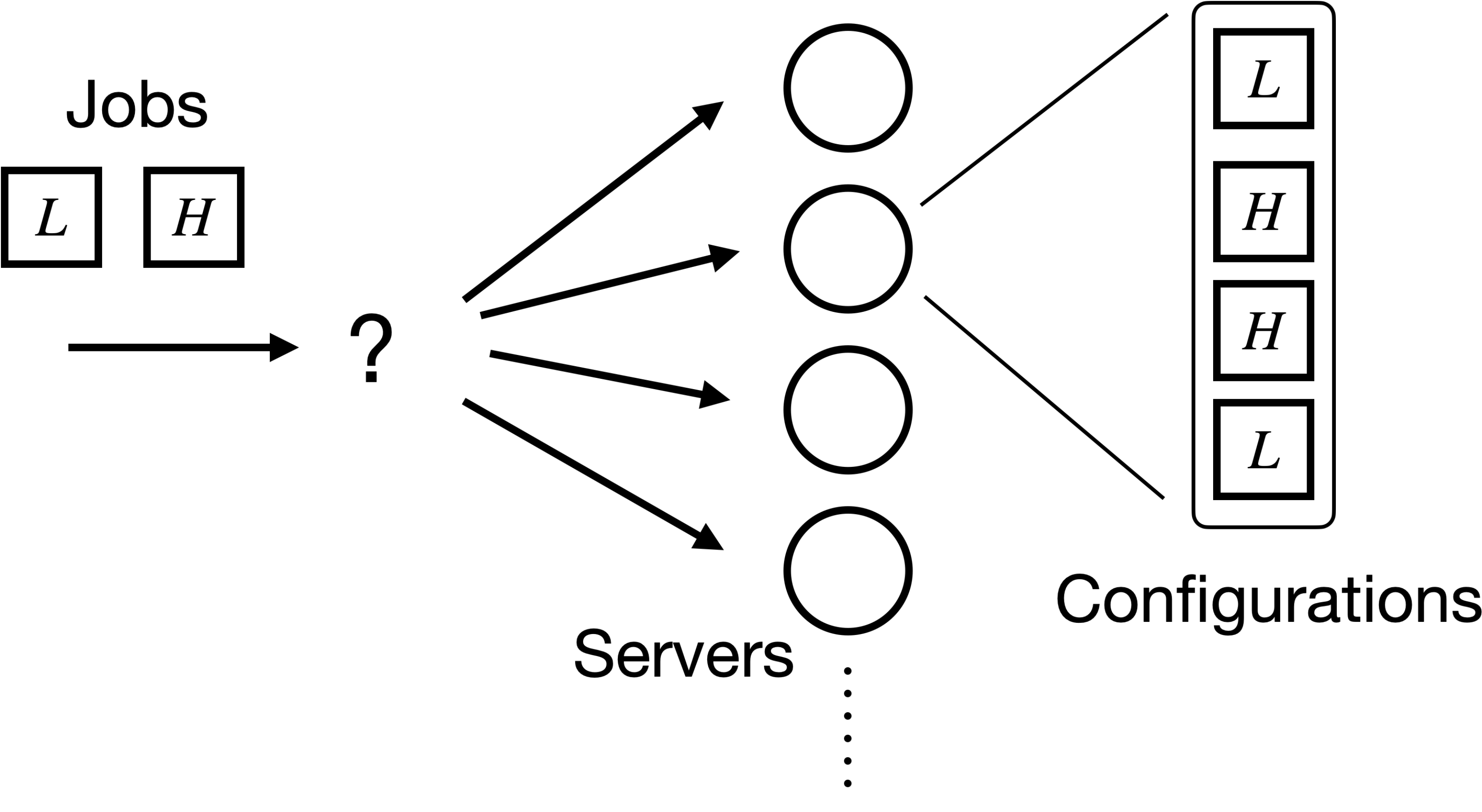}
        \caption{A system model with an infinite number of identical servers. As soon as a job arrives to the system, the job needs to be dispatched to a server to start service immediately.  The configuration of each server is the number of jobs in each phase on the server.
        }
        \label{fig:datacenter-model}
    \end{subfigure}
    \end{minipage}
    \vspace{-1ex}
    \caption{Job model and system model.}
    \vspace{-1ex}
\end{figure}

We first describe our job model that features time-varying resource requirements.
For ease of exposition, here we present a simplified setting where each job in service can be in one of the two \emph{phases}, $\low$ and $\high$, associated with \emph{low} and \emph{high} resource requirements, respectively.
Our full model, presented in Section~\ref{sec:model}, allows \emph{more than two phases and more than one type of resources}.
To model the temporal variation in the resource requirement, we assume that each job transitions between the two phases while in service until it is completed, following a continuous-time Markov chain illustrated in Figure~\ref{fig:job-model}.
We use an absorbing state $\perp$ to denote that the job is completed.
A job can initialize in either phase $\low$ or phase $\high$, and they are referred to as \emph{type} $\low$ and \emph{type} $\high$ jobs, respectively.
Note that the setting where a job's resource requirement does not vary over time is a special case of our job model where the transition rates between phases are~$0$.

We consider a system with an infinite number of identical servers, illustrated in Figure~\ref{fig:datacenter-model}.
We assume jobs arrive according to a Poisson process as existing work on SBP in service systems.
In particular, we assume that the two types of jobs arrive at the system following two independent Poisson processes, with rates $\Arr_{\low}$ and $\Arr_{\high}$, respectively;
i.e., the interarrival times of type $\low$ and type $\high$ jobs are i.i.d.\ following exponential distributions with means $1/\Arr_{\low}$ and $1/\Arr_{\high}$, respectively.
Upon arrival, a job needs to be dispatched to a server according to a \emph{dispatch policy}, and the job enters service immediately.
The goal is to design a policy $\policy$ to minimize the expected number of \emph{active servers} (servers currently serving a positive number of jobs) in steady state, denoted as $\nact(\policy)$.

As job resource requirements vary over time, situations can arise where the total job resource requirement on a server exceeds the server's resource capacity, resulting in resource contention.
Modern computing systems can tolerate temporary overruns of resource capacity, though they often incur performance degradation or other costs \cite{GoogleCloud_23,ApacheMesos_23}.
In our model, we incorporate a rate at which the cost accumulates due to resource contention.
We first represent the state of a server by its \emph{configuration}, a vector $\vek = (k_{\low}, k_{\high})$ where $k_{\low}$ and $k_{\high}$ are the numbers of jobs in phase $\low$ and phase $\high$, respectively.
Then a \emph{cost rate function} $h(\cdot)$ maps a server's configuration to a rate of cost.
For example, the cost rate can be proportional to how much the total resource requirement of the jobs on the server exceeds this server's resource capacity. 
A more general definition of $h(\cdot)$ is given in \Cref{sec:model}. 
We assume that the resource contention does not affect the transition rates in the job model nor prompt jobs to be terminated, suitable for the application scenarios where the contention level is low and manageable.
Let $\costp(\policy)$ denote the average expected cost rate per server.

Now our bin-packing problem can be formulated as follows:
\begin{equation}\label{eq:iss-opt-problem-intro}
\begin{aligned}
& \underset{\policy}{\text{minimize}}
& & \nact(\policy) \\
& \text{subject to}
& & \costp(\policy) \leq \budget,
\end{aligned}
\end{equation}
where $\epsilon > 0$ is a budget for the cost rate of resource contention. 
We are interested in solving this problem in the asymptotic regime where the arrival rates $(\Arr_{\low},\Arr_{\high})$ scale to infinity \citep{StoZho_13, StoZho_15, Sto_17, StoZho_21}, motivated by the ever-increasing computing demand that drives today's computing systems to be large-scale.
Specifically, we assume $(\Arr_{\low}, \Arr_{\high}) = (\arr_{\low}r, \arr_{\high}r)$ for some fixed coefficients $\arr_{\low}$ and $\arr_{\high}$ and a scaling factor $r$, and we study the asymptotic regime where $r$ increases.

\subsection{Main Result} 
Our main result is an \emph{asymptotically optimal} policy, named \metapolicyfull\ (\metapolicy), for this new setting of SBP in service systems with time-varying job resource requirements.
The asymptotic optimality is in the sense that under our proposed policy \metapolicy, the expected number of active servers is at most $\left(1+O\left(r^{-0.5}\right)\right)$ times the optimal objective value of the optimization problem in \eqref{eq:iss-opt-problem-intro}, while the cost rate incurred is at most $\left(1+O\left(r^{-0.5}\right)\right)\cdot \epsilon$ (i.e., exceeding the budget by at most a diminishing fraction).
This asymptotic optimality result translates into an \emph{additive optimality gap} of $O(\sqrt{r})$ in the objective value (expected number of active servers), since the optimal objective value can be shown to be $\Theta(r)$.
This main result is formally presented in Theorem~\ref{theo:asymp-opt}.

Our model can be specialized to the traditional setting of SBP in service systems where jobs' resource requirements remain fixed over time.
For this specialization, we replace the constraint $\costp(\policy) \leq \budget$ in the problem formulation \eqref{eq:iss-opt-problem-intro} with a capacity constraint, which requires the total resource requirement by jobs on a server to be within the server's resource capacity.
Our proposed policy \metapolicy\ can then be adapted into one that has an $O(\sqrt{r})$ optimality gap in the objective value, which improves upon the state-of-the-art $o(r)$ optimality gap.
A discussion on the implementation complexity of \metapolicy\ and how it compares with existing policies for the traditional setting of SBP in service systems is provided in Section~\ref{sec:sat:JRRS-discussion}. 
To be clear, this setting is not a strict special case of the formulation in \eqref{eq:iss-opt-problem-intro} because $\epsilon > 0$ is required there, but our approach and proof carry over.

From a technical approach perspective, our contribution is a novel approach that decomposes the policy design into two steps: defining a single-server sub-problem, and then converting the solution of the sub-problem into a policy in the original problem. This decomposition not only reduces the complexity of policy design but also makes the analysis tractable.
We provide an overview of this approach in \Cref{sec:intro:approach}.

\subsection{Approach Overview}\label{sec:intro:approach}
To motivate our approach, we ask two questions: 
\begin{center}
    \emph{How should we design a good dispatch policy for this system?} \\
    \emph{How can we prove that a dispatch policy is asymptotically optimal? }
\end{center} 
Before presenting our answers to these two questions, we first comment on why they are challenging to answer. 
On the one hand, solving this problem directly via dynamic programming is intractable due to the unbounded state space resulting from the infinitely many servers. 
Even if we restrict ourselves to the servers that are active, the state space is still prohibitively large. 
On the other hand, we can consider designing a heuristic policy. However, unlike traditional SBP problems where we can simply seek to pack the servers as compactly as possible, here the question of \emph{how many jobs should be put on a server} is complicated.
The complexity comes from the time-varying resource requirements of jobs, which affect future resource contention. 

\subsubsection*{\textbf{A policy-conversion framework.}}
We answer the two questions at the same time with a novel \emph{policy-conversion} framework. The framework has two steps:
\begin{enumerate}
    \item Define the \textit{single-server problem}, which is a easy-to-solve low-dimensional subproblem; 
    \item Convert the optimal policy of the single-server problem into a policy in the original problem. 
\end{enumerate}
This framework allows us to break down the complicated policy design problem into two components.
In defining the single-server problem in Step 1, our goal is to quantify the throughput of each server under the resource contention constraint;
in the policy conversion in Step 2, our goal is to dispatch jobs optimally based on each single server's characteristics.
As we will show, a careful construction of the single-server problem and the conversion procedure naturally leads to a policy for the original system and a proof of its asymptotic optimality. 
Below we give a quick overview of how we carry out these two steps and the motivation for the design choices. 

\paragraph{Single-server problem.}
To define the single-server problem, consider the following setting: suppose that our goal is to maximize the throughput of \emph{one specific server} while keeping its expected cost rate of resource contention below $\epsilon$, then how should we send jobs to this server? 
Observe that even though we want to send jobs to the server as frequently as possible, 
the frequency is fundamentally limited by how fast the server is able to serve jobs and how many jobs can be packed on the server. 
This motivates us to consider the single-server system illustrated in Figure~\ref{fig:single-server-system}. The system has one server and an infinite supply of jobs of all types, so the server can start the service of any number of new jobs of any type at any time. We say the server \emph{requests} a job from the infinite supply whenever it starts serving a new job.
We assume the same job model and cost model as in the original infinite-server system. 
The single-server problem aims to find a job-requesting policy that maximizes the throughput (the number of each type of job that can be served) along the direction of the arrival rate vector $(\Arr_{\low}, \Arr_{\high}) = (\arr_{\low}r, \arr_{\high}r)$, while maintaining the steady-state expected cost rate of resource contention below $\epsilon$.

\begin{figure}
    \centering
    \includegraphics[height=2.7cm]{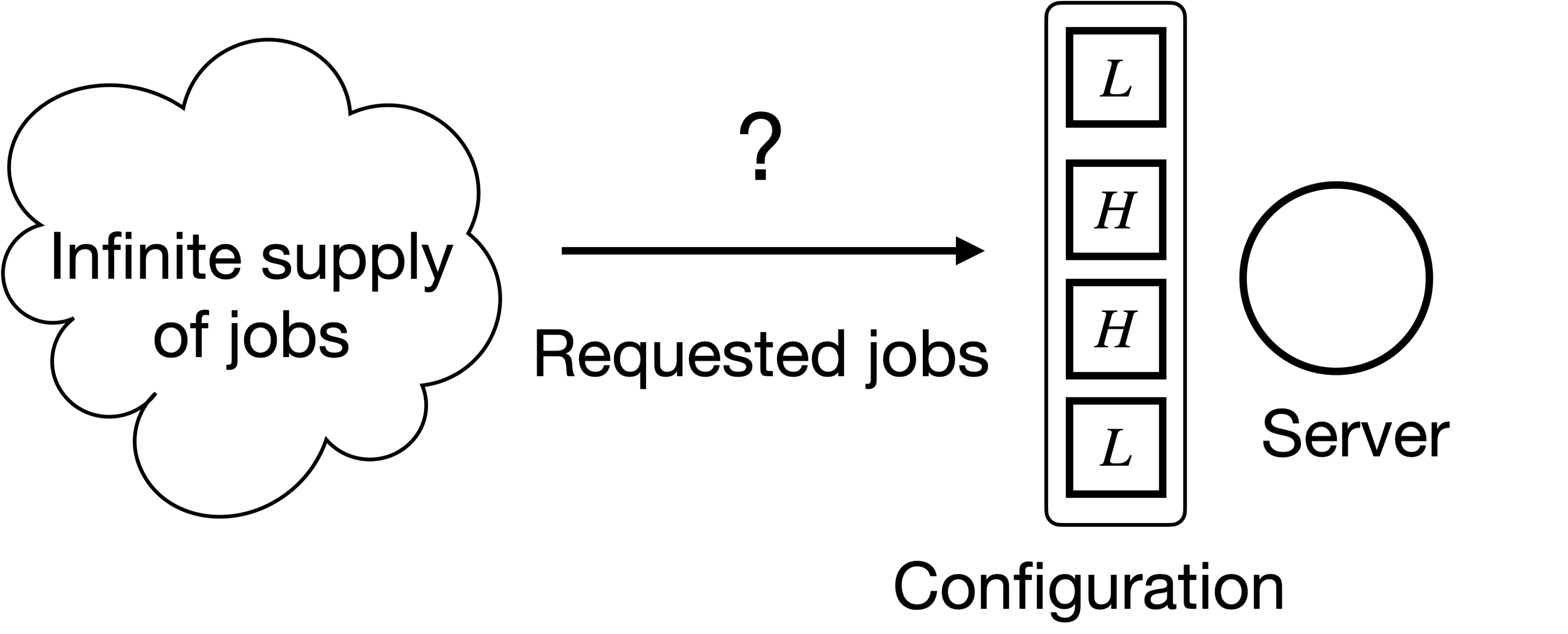}
    \vspace{-1ex}
    \caption{A single-server system with an infinite supply of jobs. A single-server policy decides when to request jobs and how many jobs of each type to request. }
    \label{fig:single-server-system}
    \vspace{-1ex}
\end{figure}

\paragraph{How is the single-server problem related to the original problem?} 
Let $\sysbar{\nact}^*$ be the number such that the total throughput of $\sysbar{\nact}^*$ single-server systems under the optimal job-requesting policy is equal to $(\arr_{\low} r, \arr_{\high}r)$ (assuming $\sysbar{\nact}^*$ is an integer for simplicity). Consider the following policy in the original system: let each of the first $\sysbar{\nact}^*$ servers in the original system adopt the optimal job-requesting policy and send requests to the dispatcher based on its current configuration. 
If the requested jobs were to arrive as soon as the dispatcher received the requests, the dispatcher would be able to fulfill the requests immediately.
In this case, the first $\sysbar{\nact}^*$ servers in the original system would have the same dynamics as $\sysbar{\nact}^*$ i.i.d.\ single-server systems, achieving the largest possible throughput and satisfying the constraint on resources contention. So the original system would have achieved the optimal number of active servers.

However, in the actual model, the dispatcher cannot immediately fulfill a job request because jobs arrive stochastically over time.
Nevertheless, the dispatcher can still find a suitable way to match each job arrival with the requests. 
To see this, note that the time points when the dispatcher receives type $i$ requests result from the superposition of $\sysbar{\nact}^*=\Theta(r)$ independent point processes, each with the average rate $\lambda_i r / \sysbar{\nact}^*$. 
As $r\to\infty$, the instantaneous rates of requesting type $i$ jobs concentrate around the arrival rate $\lambda_i r$ for each $i$. 
As a result, most job requests can be fulfilled within a diminishing delay, so most servers in the original system can closely track the optimal single-server dynamics.

\paragraph{A meta-policy, \metapolicyfull\ (\metapolicy), and its asymptotic optimality.}
Based on the single-server problem and the idea of tracking the optimal single-server dynamics, we propose a meta-policy, \metapolicyfull\ (\metapolicy), which converts a single-server policy $\sysbar{\policy}$ to a dispatch policy in the original infinite-server system. 
We say that \metapolicy\ takes $\sysbar{\policy}$ as a subroutine.
The full definition of \metapolicy\ is given in \Cref{sec:sat:jrs-def}, along with discussions on various practical considerations in its implementation. 

We show that the asymptotic performance of \metapolicy\ is related to its subroutine in the sense described in Theorem~\ref{theo:sat:conversion-general}, which we refer to as the \emph{conversion theorem}. 
In particular, \metapolicy\ with the optimal single-server policy (which we refer to as \textsc{Single-OPT}) as the subroutine is asymptotically optimal for the original infinite-server problem as $r\to\infty$.

In order to track the optimal single-server dynamics, \metapolicy\ uses a more sophisticated mechanism to control the long-term consequences of missing job requests or fulfilling requests with delays. The mechanism involves the auxiliary variables of \emph{tokens} and \emph{virtual jobs}, which regularize the process of generating requests and matching arrivals with requests. These auxiliary variables play a crucial role in the proof of Theorem~\ref{theo:sat:conversion-general} in \Cref{sec:sat:conv}, where a novel Stein's method argument is carried out. 
We discuss the role of tokens and virtual jobs and their necessity at the end of \Cref{sec:sat:jrs-def} and in \Cref{sec:discuss-token-virtual-job}.

Finally, we comment that our policy conversion framework can be applicable to other systems with similar structures. 
Specifically, we can try to define a suitable single-server problem, solve for its optimal policy, and convert the optimal single-server policy to the original problem. 
A similar conversion theorem should hold as long as the servers are weakly coupled in some sense. 
See \Cref{sec:conclusion} for a discussion of such systems.

\paragraph{Relation to the mean-field approach.} 
We remark that the mean-field approach often studies the empirical distribution of configurations on all servers \citep{Sto_13,StoZho_13, StoZho_15, Sto_17, StoZho_21, GhaZhoSri_14}, which can be viewed as a probability distribution of a single server's configuration.
However, the mean-field approach is typically used to \emph{analyze} this empirical distribution under a given policy for the original system.
In contrast, our approach solves the single-server problem to \emph{design} a single-server policy, and then converts it to a policy in the original system with performance guarantee.

\subsection{\textbf{Paper organization}}
In Section~\ref{sec:model}, we present the general problem formulation, which generalizes the simplified version in this section.
In Section~\ref{sec:result-approach}, we give a more detailed overview of our main result and approach, with a short proof of Theorem~\ref{theo:asymp-opt} (main result) based on Theorems~\ref{theo:sat:lower-bound}--\ref{theo:single-opt} at the end of the section.
Section~\ref{sec:sat:jrs-def} provides a detailed description of our meta-policy, \metapolicyfull\ (\metapolicy), along with discussions on practical considerations in its implementation. In Section~\ref{sec:sat:conv}, we prove the performance guarantee of \metapolicy\ (Theorem~\ref{theo:sat:conversion-general}) under an irreducibility assumption. The proof for the general case and other proofs are deferred to the appendices. 
We conclude the paper and discuss some future directions in Section~\ref{sec:conclusion}.

\section{Problem Formulation}\label{sec:model}
\paragraph*{\textbf{Job Model.}}
As described in Section~\ref{sec:intro}, we consider a job model where each job in service can be in one of multiple \emph{phases}, each phase associated with a different resource requirement.
Here the resource requirement can be a multi-dimensional vector, with each coordinate specifying the requirement of one type of resource.
To model the temporal variation in the resource requirement, we assume that each job transitions between phases while in service until it is completed.
The phase transition process is described by a continuous-time Markov chain on the state space $\jobspace\cup\{\perp\}$, where $\jobspace$ is the set of phases and $\perp$ is the absorbing state that denotes the completion of the job. 
We call a transition between two phases in $\jobspace$ an internal transition, and let $\mu_{ii'}$ denote the transition rate from phase $i$ to phase $i'$;
the departure of a job then corresponds to a transition from a phase $i\in\jobspace$ to $\perp$, whose transition rate is denoted as $\mu_{i\perp}$.
The phase transitions of different jobs are assumed to be independent of each other.

We classify a job as a type $i$ job if it starts from phase $i\in\jobspace$ when entering service.
Jobs of each type $i$ arrive to the system according to an independent Poisson process with rate  $\Arr_i$.

\paragraph*{\textbf{Server Model.}} 
We consider an infinite-server system with identical servers.
As soon as a job arrives to the system, the job needs to be dispatched to a server to start service immediately.
Note that this is always feasible because there are an infinite number of servers in the system.
We assume that jobs cannot be preempted or migrated.
To describe the state of a server, we define the \textit{configuration} of a server as an $\abs{\jobspace}$-dimensional vector $\vek = (k_i)_{i\in\jobspace}$, whose $i$-th entry $k_i$ is the number of jobs in phase $i$ on the server.
Each server has a limit on the total number of jobs in service at the same time.
This limit is denoted as $\Kmax$ and referred to as the \textit{service limit}.
Then the set of feasible server configurations is $\Kinner \triangleq \{\vek\colon\sum_{i\in\jobspace} k_i \leq \Kmax\}$.

\paragraph*{\textbf{System Dynamics.}}
The system state can be represented by the concatenation of the configurations of all servers. Specifically, we index the servers by positive integer numbers and denote the configuration of server $\ell$ at time $t$ as $\veK^\ell(t)$.
Then the state of the entire system can be represented by the infinite vector $(\veK^\ell(t))_{\ell\in\mathbb{Z}_+}$.

Suppose that the system is in state $(\vek^\ell)_{\ell\in\mathbb{Z}_+}$.
Let $\vei_i$ be an $\abs{\jobspace}$-dimensional vector whose $i$-th entry is $1$ and all other entries are $0$.
Then the following state transitions can happen:
\begin{itemize}
    \item $\vek^\ell\rightarrow \vek^\ell+\vei_i$,  $\vek^{\ell'}\rightarrow \vek^{\ell'}$ $\forall\ell'\neq\ell$: a type $i$ job arrives and is dispatched to server~$\ell$;
    \item $\vek^\ell\rightarrow \vek^\ell+\vei_{i'}-\vei_i$, $\vek^{\ell'}\rightarrow \vek^{\ell'}$ $\forall\ell'\neq\ell$: a job on server $\ell$ transits from phase $i$ to phase $i'$;
    \item $\vek^\ell\rightarrow \vek^\ell-\vei_i$,  $\vek^{\ell'}\rightarrow \vek^{\ell'}$ $\forall\ell'\neq\ell$: a job on server $\ell$ departs the system from phase $i$.
\end{itemize}
The specifics of the system dynamics depend on the employed \emph{dispatch policy} that decides which server to dispatch to when a job arrives.

\paragraph*{\textbf{Active Servers.}}
We are interested in the number of \emph{active servers}, i.e., servers currently serving a positive number of jobs.
Note that given the arrival rates of jobs, the smaller the number of active servers, the better the system is utilized.
Let $X_{\vek}(t)$ be the number of servers in configuration $\vek$ at time $t$, i.e.,
$X_{\vek}(t)=\sum_{\ell=1}^{\infty}\indibrac{\veK^\ell(t)=\vek}.$
Then the number of active servers can be written as
$\sum_{\vek\neq\vzero} X_{\vek}(t),$
where $\vzero\in\mathbb{R}^{|\jobspace|}$ is the zero vector.

\paragraph*{\textbf{Cost of Resource Contention.}}
Recall that the cost rate function $h(\cdot)$ maps a server's configuration to a rate of cost.
We assume that $h(\cdot)$ is any function that is $\Gamma$-Lipschitz continuous with respect to the $L^1$ distance for some constant $\Gamma > 0$ and satisfies $h(\vzero)=0$.

\paragraph*{\textbf{Performance Goal.}}
Our high-level goal is to design dispatch policies that minimize the number of active users while keeping the cost rate of resource contention within a certain budget. 
Specifically, we consider policies that are allowed to be randomized and non-Markovian (i.e., the policies can make history-dependent decisions). 
We further focus on policies that induce a unique stationary distribution on the configuration process $\{(\veK^\ell(t))_{\ell\in\mathbb{Z}_+}\}$,
assuming that the configuration process is embedded in a Markov chain that has a unique stationary distribution.
We are interested in such policies because the resulting time averages of quantities related to the configurations are equal to the corresponding expectations under the unique stationary distribution regardless of the initial state. 
Let $\policy$ be a policy of interest, $(\veK^\ell)_{\ell\in\mathbb{Z}_+}$ be a random element that follows the stationary distribution of the system state induced by $\policy$, and $X_{\vek}$ be the corresponding number of servers in configuration $\vek$ in steady state under $\policy$.
Then the expected number of active servers is given by
\begin{equation*}
\nact(\policy) \triangleq \sum_{\vek\neq\vzero}\E[X_{\vek}].
\end{equation*}
We define the expected cost rate per expected active server as
\[
    \costp(\policy)\triangleq \frac{\sum_{\vek\neq\vzero}h(\vek)\E[X_{\vek}]}{\sum_{\vek\neq\vzero}\E[X_{\vek}]}.
\]
Note that if $\costp(\policy)\leq\epsilon$, we have $\sum_{\vek\neq\vzero}h(\vek)\E[X_{\vek}]\leq \budget \sum_{\vek\neq\vzero}\E[X_{\vek}]$. 
Now our goal can be formulated as the following optimization problem, referred to as problem $\isp((\Arr_i)_{i\in\jobspace}, \budget)$:
\begin{equation}\label{eq:iss-opt-problem} 
\begin{aligned}
& \underset{\policy}{\text{minimize}}
& & \nact(\policy) \\
& \text{subject to}
& & \costp(\policy) \leq \budget,
\end{aligned}
\end{equation}
where $\epsilon$ is a budget for the cost rate of resource contention.

\paragraph*{\textbf{Asymptotic Optimality.}}
We focus on the asymptotic regime where for all $i\in\jobspace$, the arrival rate is given by $\Arr_i = \arr_i r$ for some constant coefficient $\arr_i$ and a positive \emph{scaling factor} $r\to+\infty$.
To define asymptotic optimality, we first define the following notion of approximation to the optimization problem $\isp((\Arr_i)_{i\in\jobspace}, \budget)$ in \eqref{eq:iss-opt-problem}:
a policy $\policy$ is said to be $(\alpha, \beta)$-optimal if $\nact(\policy) \leq \alpha \cdot \nact^*((\Arr_i)_{i\in\jobspace}, \budget)$ and $\costp(\policy) \leq \beta \cdot \budget$, where $\nact^*((\Arr_i)_{i\in\jobspace}, \budget)$ is the optimal objective value in \eqref{eq:iss-opt-problem}.
Now consider a family of policies $\sysr{\policy}$ indexed by the scaling factor $r$. We say that the policy $\sysr{\policy}$ is \emph{asymptotically optimal} if it is $\left(\sysr{\alpha}, \sysr{\beta}\right)$-optimal to the optimization problem $\isp((\arr_ir)_{i\in\jobspace}, \budget)$ with $\sysr{\alpha}, \sysr{\beta}\to 1$ as $r \to\infty$.
We will suppress the superscript $(r)$ for simplicity when it is clear from the context.

We note that under any policy $\policy$, $\nact(\policy)=\Theta(r)$.
This can be proven using the renowned Little's Law \citep{Kle_75} in the following way. 
The total job arrival rate is $\Theta(r)$ and the expected time that a job spends in the system is $O(1)$.
So by Little's Law, the expected number of jobs in the system in steady state is $\Theta(r)$.
Since each server can accommodate a constant number of jobs, the expected number of active servers is $\Theta(r)$.
Given this, $\left(\sysr{\alpha}, \sysr{\beta}\right)$-optimality implies an optimality gap of $\sysr{\alpha}\cdot r$ in the objective value.

\section{Main Result and Our Approach}\label{sec:result-approach}
\subsection{Main Result}\label{subsec:result}
Our main result, Theorem~\ref{theo:asymp-opt}, is the asymptotic optimality of our proposed policy \metapolicyfull\ (\metapolicy), with a subroutine we call \textsc{Single-OPT}, as briefly discussed in Section~\ref{sec:intro}. 
This asymptotic optimality result implies an $O(\sqrt{r})$ optimality gap in the expected number of active servers.
We defer the detailed descriptions of \metapolicy\ and \textsc{Single-OPT} to Section~\ref{sec:sat:jrs-def} and Appendix~\ref{sec:solve-single-lp}.
Theorem~\ref{theo:asymp-opt} follows immediately from Theorems~\ref{theo:sat:lower-bound}--\ref{theo:single-opt} to be introduced in Section~\ref{subsec:approach}; a short proof is included at the end of Section~\ref{subsec:approach} for clarity.

\begin{theorem}[Asymptotic Optimality]\label{theo:asymp-opt}
    Consider a stochastic bin-packing problem in service systems with time-varying job resource requirements.
    Let the arrival rates be $(\arr_i r)_{i\in\jobspace}$ and the cost rate budget be $\budget > 0$.
    Then the policy \metapolicyfull\ (\metapolicy) with the subroutine \textsc{Single-OPT} is $\left(1+O\left(r^{-0.5}\right), 1+O\left(r^{-0.5}\right)\right)$-optimal.
    That is, the expected number of active servers under \metapolicy\ with \textsc{Single-OPT} is at most $\left(1+O\left(r^{-0.5}\right)\right)$ times the optimal value of the problem $\mathcal{P}((\arr_i r)_{i\in\jobspace},\epsilon)$, while the cost rate incurred is at most $\left(1+O\left(r^{-0.5}\right)\right)\cdot \epsilon$.
\end{theorem}

\paragraph*{\textbf{Specialization to Non-Time-Varying Resource Requirements.}}
As mentioned in Section~\ref{sec:intro}, we can specialize this result to the setting where the resource requirement of a job does not vary over time.
To do that, we remove the cost constraint in $\mathcal{P}((\arr_i r)_{i\in\jobspace},\epsilon)$, and redefine the set of feasible server configurations, $\Kinner$, to also incorporate hard capacity constraints for each type of resources.
The rest of the analysis is almost identical to that of the analysis for time-varying resource requirements; we omit the details due to the space limit.
This specialization results in a policy that is $\left(1+O\left(r^{-0.5}\right)\right)$-optimal in the expected number of active servers.

\subsection{Our Approach}\label{subsec:approach}
In a nutshell, our approach is to reduce the original optimization problem in an infinite-server system to an optimization problem in a single-server system, which is defined below.

\paragraph*{\textbf{A Single-Server System.}}
Consider a single-server system serving jobs with time-varying resource requirements. 
The system has \textit{an infinite supply of jobs} of all types. 
As a result, the server can \textit{request} any number of new jobs of any type at any time. Once a job is requested, it immediately enters service. 

We represent the server configuration at time $t$ using a vector $\sysbar{\veK}(t) = (\sysbar{K}_i(t))_{i\in\jobspace}$, whose $i$-th entry denotes the number of jobs in phase $i$. We assume that the single-server system has the same service limit $\Kmax$ and cost rate function $h(\cdot)$ as a server in the original infinite-server system. Therefore, the server configuration $\sysbar{\veK}(t)$ is also in the set $\Kinner = \{\vek\colon\sum_{i\in\jobspace} k_i \leq \Kmax\}$, and the cost rate at time $t$ is $h(\sysbar{\veK}(t))$.

A single-server policy $\sysbar{\policy}$ determines when and how many jobs of each type to request. 
We allow the single-server policy to be randomized and assume it is Markovian, i.e., it makes decisions only based on the current configuration. Note that allowing non-Markovian policies will not change the optimal value of the single-server problem that we will consider (see \Cref{sec:solve-single-lp}). 
Let $\pi\triangleq(\pi(\vek))_{\vek\in\Kinner}$ be a stationary distribution of the server configuration under the policy $\sysbar{\policy}$, and let $\sysbar{\veK}(\infty)$ be a random variable with the distribution $\pi$. 
When we consider a policy $\sysbar{\policy}$ and its stationary distribution $\pi$, we assume that the system is initialized from $\pi$. 
The policy $\sysbar{\policy}$ together with $\pi$ defines the \textit{request rate} of type $i$ jobs $\reqrate_i$, which is the expected number of type $i$ jobs requested per unit time in steady state. Note that $\reqrate_i$ is the throughput of type $i$ jobs since the system has a finite state space. 

We consider the following single-server problem, denoted as $\ssp((\arr_ir)_{i\in\jobspace},\budget)$:
\begin{equation}\label{eq:sat:ssss-opt-problem}
    \begin{aligned}
        & \underset{\sysbar{\nact}, \mspace{10mu}  \sysbar{\policy}, \mspace{10mu} \pi }{\text{minimize}}
        & & \sysbar{\nact} \\
        & \text{subject to}
        &  &\E\left[h\big(\sysbar{\veK}(\infty)\big)\middle| \sysbar{\veK}(\infty) \neq \vzero\right] \leq \budget,\\
        & &  & \sysbar{\nact}\, \reqrate_i = \arr_i r, \quad \forall i\in\jobspace.
    \end{aligned}
\end{equation}
The single-server problem can be interpreted as follows.
We can think of $\sysbar{\nact}$ as the number of copies of the single-server system under $\sysbar{\policy}$ needed to support the arrival rates $(\arr_ir)_{i\in\jobspace}$ in the infinite-server system.
To minimize $\sysbar{\nact}$, it is equivalent to maximizing the throughput $(\reqrate_i)_{i\in\jobspace}$ in each single-server system, while maintaining their proportions as $(\arr_ir)_{i\in\jobspace}$.

We remark that for the problem $\ssp((\arr_ir)_{i\in\jobspace},\budget)$, we only need to consider policies that do not depend on the scaling factor $r$. 
To see this, we can replace the decision variable $\sysbar{\nact}$ with $\sysbar{\nactsmall} \triangleq \sysbar{\nact} / r$ and the optimization problem can be equivalently formulated as follows, which does not involve $r$:
\begin{equation}\label{eq:sat:ssss-opt-problem-normalized}
    \begin{aligned}
        & \underset{\sysbar{\nactsmall}, \mspace{10mu} \sysbar{\policy}, \mspace{10mu} \pi }{\text{minimize}}
        & & \sysbar{\nactsmall} \\
        & \text{subject to}
        &  & \E\left[h\big(\sysbar{\veK}(\infty)\big)\middle| \sysbar{\veK}(\infty) \neq \vzero\right] \leq \budget,\\
        & &  & \sysbar{\nactsmall} \, \reqrate_i = \arr_i, \quad \forall i\in\jobspace.
    \end{aligned}
\end{equation}

\paragraph*{\textbf{Lower Bound.}}
The single-server problem gives a lower bound to the original problem in \eqref{eq:iss-opt-problem} as stated in the following theorem. The proof is given in Appendix~\ref{sec:sat:lower-bound}.
\begin{theorem}[Lower Bound]\label{theo:sat:lower-bound}
    Consider a stochastic bin-packing problem in service systems with time-varying job resource requirements.  Let the arrival rates be $(\arr_i r)_{i\in\jobspace}$ and the cost rate budget be $\budget>0$. Let $\nact^*$ be the optimal value of the original infinite-server problem in \eqref{eq:iss-opt-problem}, and let $\sysbar{\nact}^*$ be the optimal value of the single-server problem $\ssp((\arr_i r)_{i\in\jobspace}, \budget)$, then $\nact^* \geq \sysbar{\nact}^*$.
\end{theorem}

\paragraph*{\textbf{Converting From the Single-Server System to the Infinite-Server System.}}
Having established a lower bound on the infinite-server problem $\isp((\arr_i r)_{i\in\jobspace}, \budget)$ in terms of the optimal value of the single-server problem $\ssp((\arr_i r)_{i\in\jobspace}, \budget)$, next we focus on finding an asymptotically optimal policy. 
We will characterize the performance guarantee of a class of policies and then show that the best policy within the class is asymptotically optimal. Specifically, we consider a meta-policy called \metapolicyfull\ (\metapolicy), which converts a Markovian single-server policy $\sysbar{\policy}$ into an infinite-server policy. We call the policy resulting from the conversion  \textit{a \metapolicy\ policy with a subroutine $\sysbar{\policy}$}. Through analyzing the meta-policy \metapolicy, we show that the performance of each \metapolicy\ policy can be characterized by the performance of its subroutine, as stated in Theorem~\ref{theo:sat:conversion-general} below.
The proof of Theorem~\ref{theo:sat:conversion-general} under an irreducibility assumption is given in Section~\ref{sec:sat:conv}, and the proof for the full version is given in Appendix~\ref{sec:proof-conv-general}.

\begin{theorem}[Conversion Theorem]\label{theo:sat:conversion-general}
Consider a stochastic bin-packing problem in service systems with time-varying job resource requirements.
Let the arrival rates be $(\arr_i r)_{i\in\jobspace}$ and the cost rate budget be $\budget>0$. Let $(\sysbar{\nact}, \sysbar{\policy}, \pi(\vek))$ be a solution feasible to the single-server problem $\ssp((\arr_i r)_{i\in\jobspace}, \budget)$. In addition, we assume that the policy $\sysbar{\policy}$ is Markovian. Let the infinite-server policy $\policy$ be \metapolicy\ with a subroutine $\sysbar{\policy}$. Then under $\policy$, we have
\begin{align}
    \abs{\sum_{\vek \neq \vzero} \E\left[X_{\vek}\right] - \ceil{\sysbar{\nact}} \cdot\Prob\left(\sysbar{\veK}\neq \vzero\right)}  &= \Obrac{\sqrt{r}}, \label{eq:sat:active-servers}\\
    \abs{\sum_{\vek\neq\vzero} h(\vek) \E\left[X_{\vek}\right] - \ceil{\sysbar{\nact}} \cdot \E\left[h(\sysbar{\veK})\right]} &= \Obrac{\sqrt{r}}.\label{eq:sat:cost}
\end{align}
As a result, 
\begin{align}
    \nact(\policy) &\leq \left(1 + \Obrac{r^{-0.5}}\right) \cdot \sysbar{\nact}, \label{eq:alpha-optimal}\\
    \costp(\policy) &\leq \left(1 + \Obrac{r^{-0.5}}\right)\cdot \budget \label{eq:beta-optimal}.
\end{align}
\end{theorem}

\paragraph*{\textbf{Optimal Single-Server Policy.}}
Theorem~\ref{theo:sat:conversion-general} together with the lower bound in Theorem~\ref{theo:sat:lower-bound} reduces the infinite-server problem $\isp((\arr_i r)_{i\in\jobspace}, \budget)$ in \eqref{eq:iss-opt-problem} to the single-server problem $\ssp((\arr_i r)_{i\in\jobspace}, \budget)$ in \eqref{eq:sat:ssss-opt-problem}. 
We can obtain the optimal single-server policy, \textsc{Single-OPT}, by solving a linear program, as stated in the theorem below. 
\begin{theorem}[Optimality of Single-OPT, Informal]\label{theo:single-opt}
    There exists a linear program $\sslp((\arr_i)_{i\in\jobspace}, \budget)$ that is equivalent to the single-server problem $\ssp((\arr_i r)_{i\in\jobspace}, \budget)$. In particular, we can construct an optimal Markovian policy for $\ssp((\arr_i r)_{i\in\jobspace}, \budget)$ from the optimal solution of  $\sslp((\arr_i)_{i\in\jobspace}, \budget)$.
\end{theorem}
Proof of \Cref{theo:single-opt} and details on the construction of the optimal policy are given in Appendix~\ref{sec:solve-single-lp}.

\begin{proof}[\textbf{Proof of the Theorem~\ref{theo:asymp-opt} Based on Theorems~\ref{theo:sat:lower-bound}--\ref{theo:single-opt}.}]
    Because \textsc{Single-OPT} (along with an optimal stationary distribution) optimally solves $\ssp((\arr_i r)_{i\in\jobspace}, \budget)$, it achieves the optimal value $\sysbar{\nact}^*$. Let $\policy$ be \metapolicy\ with a subroutine \textsc{Single-OPT}, then according to Theorem~\ref{theo:sat:conversion-general}, we have $\nact(\policy) \leq \left(1 + \Obrac{r^{-0.5}}\right) \cdot \sysbar{\nact}^*$ and $\costp(\policy) \leq \left(1 + \Obrac{r^{-0.5}}\right)\cdot \budget$.
    By Theorem~\ref{theo:sat:lower-bound}, we also have 
    $
        \nact^* \geq \sysbar{\nact}^*.
    $
    So we conclude that \metapolicy\ with a subroutine \textsc{Single-OPT} is $\left(1+O(r^{-0.5}), 1+O(r^{-0.5})\right)$-optimal.
\end{proof}

\section{Proposed meta-policy: \metapolicyfull\ (\metapolicy)}\label{sec:sat:jrs-def}

In this section, we describe our meta-policy, \metapolicyfull\ (\metapolicy), in full detail. For ease of presentation, we focus on the case where the subroutine policy $\sysbar{\policy}$ for \metapolicy\ is \textit{$\vek^0$-irreducible}, i.e., under $\sysbar{\policy}$, there exists a configuration $\vek^0$ such that the single-server system can return to $\vek^0$ from any other configurations (which is equivalent to assuming that the configuration of the single-server system under policy $\sysbar{\policy}$ forms a \emph{unichain}). The algorithm for the general case is given in \Cref{sec:proof-conv-general}.

\subsection{How the Single-Server Policy Requests Jobs} \label{sec:how-single-server-request-jobs}
Before going into the definition of \metapolicy, we first take a closer look at how the Markovian \emph{single-server} policy $\sysbar{\policy}$ requests jobs, to avoid potential ambiguity caused by the fact that a single-server policy can request jobs at \textit{any time}. 
Let $a_i$ denote the number of type $i$ jobs requested, and let $\vea \triangleq (a_i)_{i\in\jobspace}$.
We say $\vea$ is \textit{feasible} if the total number of jobs on the server does not exceed $\Kmax$ after adding the jobs. 
The policy $\sysbar{\policy}$ performs one of the following two types of requests based on the current configuration. 
\begin{itemize}[leftmargin=1em]
    \item \textbf{Reactive requests.} A \emph{reactive request }is triggered by \emph{either an internal transition or a departure}. 
    The changes in the configuration when a reactive request is made can be represented by the diagram 
    \[
        \vek \to \vek' \to \vek' + \vea,
    \]
    where $\vek\to\vek'$ is due to the internal transition or departure, and $\vek' \to \vek' + \vea$ happens since the policy immediately requests $\vea$ jobs. 
    The policy $\sysbar{\policy}$ specifies a probability distribution over all feasible $\vea$ when it decides to perform reactive requests for the configuration $\vek'$. 
    \item\textbf{Proactive requests.} A \emph{proactive request} happens on its own, and it happens at a \emph{finite rate} depending on the current configuration of the server. 
    The change of the configuration when a proactive request happens can be represented by the diagram 
    \[
        \vek \to \vek + \vea. 
    \]
    More specifically, suppose the policy $\sysbar{\policy}$ decides to perform proactive requests for a configuration~$\vek$.  Then for each feasible $\vea$, the policy $\sysbar{\policy}$ specifies a rate and runs a timer with an exponentially distributed duration with the specified rate. 
    When the timer ticks, $\vea$ jobs are requested. When the configuration changes, all the timers are canceled and restarted with new rates based on the new configuration. 
\end{itemize}

\subsection{Description of \metapolicyfull\ (\metapolicy)}\label{sec:sat:JRRS-irreducible}
The inputs of \metapolicy\ include: (i) the single-server policy $\sysbar{\policy}$, (ii) the objective value of $\sysbar{\policy}$ in the single-server problem \eqref{eq:sat:ssss-opt-problem}, denoted as $\sysbar{\nact}$, and (iii) the transition rates in the job model.

We first divide the infinite server pool into two sets based on the server index $\ell$. Let $L = \lceil\sysbar{\nact} \rceil$. We call servers with index $\ell \leq L$ \textit{normal servers}; we call servers with index $\ell > L$ \textit{backup servers}. The normal servers are responsible for serving most of the jobs, while the backup servers are activated only to handle overflow jobs (jobs that are not dispatched to normal servers).

The \metapolicy\ is specified in two steps. 
\begin{itemize}[leftmargin=1em]
    \item \textbf{Step 1 (Job Requesting on a Normal Server):} We let each normal server request jobs using its subroutine, the single-server policy  $\sysbar{\policy}$.
    The input to the policy $\sysbar{\policy}$ is what we refer to as the \emph{observed configuration} of the server, which will be further explained below.
    When $\sysbar{\policy}$ requests $\vea = (a_i)_{i\in\jobspace}$ jobs, $a_i$ type $i$ \textit{tokens} are generated for each $i\in\jobspace$ to store the job requests. 
    The server \emph{pauses} the job requesting process if it already has any type of tokens, and resumes when all the tokens that it generated are removed. 
    \item \textbf{Step 2 (Arrival Dispatching):} 
    \begin{itemize}[leftmargin=1.5em]
        \item \textbf{Real jobs.} When a type $i$ job arrives, the dispatcher chooses a type $i$ token uniformly at random, removes the token, and assigns the job to the corresponding server. When there are no type $i$ tokens, the dispatcher sends the job to an idle backup server. 
        \item \textbf{Virtual jobs.} When the total number of type $i$ tokens throughout the system exceeds the limit $\toklim = \lceil \sqrt{L} \rceil$ (called the \textit{token limit}), a type $i$ \textit{virtual arrival} is triggered, which causes the dispatcher to choose a type $i$ token uniformly at random, remove the token, and assign a \textit{virtual job} to the corresponding server. A virtual job has the same transition dynamics as a real job but does not consume physical resources. 
    \end{itemize} 
\end{itemize}
The \emph{observed configuration} of a normal server in Step~1 is the configuration resulting from real jobs and virtual jobs combined.
That is, it is a vector whose $i$-th entry represents the total number of real and virtual jobs in phase $i$ on this server.
The observed configuration changes when there is a new real or virtual job arrival assigned to the server, or when a real or virtual job on the server has a phase transition or departs.
We update the input to the policy $\sysbar{\policy}$ when the observed configuration changes.
Whenever the observed configuration changes, the policy $\sysbar{\policy}$ cancels the exponential timers in progress;
but a reactive request from the policy $\sysbar{\policy}$ can only be triggered when a real or virtual job on the server has a phase transition or departs.

\subsubsection*{\textbf{Intuition behind \metapolicy}}
To provide a better understanding of the main design ideas of \metapolicy, here we give an intuitive description of how it works. 
Broadly, servers generate job requests and store unfulfilled requests as tokens; the dispatcher assigns jobs to servers according to the tokens to fulfill job requests.
This is the mechanism for matching job arrivals with requests, which is referenced at the end of \Cref{sec:intro:approach}.  
However, rather than matching all tokens with job arrivals, \metapolicy\ opts to convert some of the tokens into virtual jobs to keep the total number of tokens within an upper limit $\toklim$. 
By capping the number of tokens, \metapolicy\ ensures that the job requests generated by each server get fulfilled quickly (either by a real job or a virtual job), and thus the observed configurations of servers maintain proximity to i.i.d.\ copies of the single-server systems. 

The choice of the token limit $\toklim=\Theta(\sqrt{r})$ balances two key considerations.
On the one hand, a smaller $\toklim$ brings the observed configurations closer to i.i.d.\ copies of single-server systems.
On the other hand, if $\toklim$ is overly small, the rate of generating virtual jobs becomes high and the probability for a job arrival to see no tokens is also high.
As a result, the observed configurations, which include both real and virtual jobs, deviate from the real-job configurations.
A more in-depth discussion on the role of tokens and virtual jobs and whether they are fundamental is in \Cref{sec:discuss-token-virtual-job}.

\subsection{Practical considerations in implementing \metapolicyfull\ (\metapolicy)}\label{sec:sat:JRRS-discussion}
\subsubsection*{\textbf{Computational complexity of \metapolicy}}
The computational complexity of \metapolicy\ consists of two components: the \emph{offline} component that computes a single-server policy $\sysbar{\policy}$ and its objective value $\sysbar{\nact}$, and the \emph{online} component that carries out the two steps of \metapolicy.

The offline component reduces to solving the linear program given in \eqref{eq:sat:ssss-lp} in Appendix~\ref{sec:sat:lp-relaxation}, whose number of optimization variables is linear in the number of feasible configurations times the number of phases, i.e., $|\Kinner|\times|\jobspace|$, on a single server.
Admittedly, $|\Kinner|$ can be large when a single server has a large quantity of resources and there are many job phases.
However, we opt for the view that a single server is not excessively large and the system's scale is primarily captured by the scaling factor $r$.
Therefore, it is advantageous that the computational complexity of this offline component is independent of $r$.

In the online component, the bulk of the computation is in job requesting and virtual job simulation, which can be executed distributedly on the normal servers.
Specifically, each normal server monitors its observed configuration and generates tokens according to the single-server policy $\sysbar{\policy}$; additionally, when a virtual job is assigned to the server, the server simulates the dynamics of the virtual job, i.e., generates random variables corresponding to phase transitions and job departure.
Backup servers do not need to perform any computation beyond serving jobs.

The scheduler, which stores all the tokens, has two responsibilities in the online component: (i)~the scheduler matches each newly arrived job to a token of the same type, chosen uniformly at random, or sends the job to a backup server when there are no tokens of the same type; (ii)~the scheduler monitors the number of tokens of each type and assigns virtual jobs when the number of tokens exceeds the limit $\toklim$.

It is informative to compare the computational complexity of \metapolicy\ with existing algorithms designed for the traditional setting of stochastic bin-packing in service systems, where the resource requirements are non-time-varying \cite{Sto_13,StoZho_13,GhaZhoSri_14,StoZho_15,Sto_17,StoZho_21}.
At a high level, these existing algorithms function as follows: upon the arrival of a job, the scheduler checks the current configurations of all servers and assigns the job to a server whose configuration optimizes certain predefined criteria.
Among these, the GRAND algorithm \cite{StoZho_15,Sto_17,StoZho_21} stands out for its simplicity and asymptotic optimality.
Under GRAND, the scheduler only needs to identify configurations that can accommodate the incoming job and then randomly assigns the job to one of these feasible servers, along with some idle servers.
Compared to \metapolicy, GRAND does not have an offline planning component, and individual servers do not perform computation beyond serving jobs.
The scheduler's role in GRAND is slightly more complex than in \metapolicy.
Consequently, when considering using \metapolicy\ in settings where job resource requirements are non-time-varying, practitioners should weigh whether the additional computational complexity is warranted.

\subsubsection*{\textbf{Model parameter estimation.}}
A limitation of \metapolicy\ is its dependency on known model parameters, including job arrival rates and phrase transition rates.
Such dependency is not present in existing algorithms designed for the setting with non-time-varying resource requirements.
In real-world applications, the model parameters can be estimated from workload traces such as \cite{TirBarDen_20,Wil_19}.
Estimation errors can impact the system's performance, an issue that merits further in-depth investigation in future work.
Here, we provide a preliminary result on the performance degradation due to parameter estimation errors.
Roughly speaking, suppose that the estimation error in the job arrival rate coefficients $\lambda_i$'s and the phase transition rates $\mu_{ii'}$'s and $\mu_{i\perp}$'s are bounded by $\delta\geq 0$ (along with an insensitivity assumption on the single-server problem).
Then if we use \metapolicy\ where the single-server policy is obtained by solving for the optimal single-server policy under the estimated parameters, the resulting \metapolicy\ is $\left(1+ B \delta + O\left(r^{-0.5}\right), 1+ B \delta+O\left(r^{-0.5}\right)\right)$-optimal for any $\delta \leq \delta_{\max}$, where $B$ and $\delta_{\max}$ are positive constants independent of $r$. 
The exact statement is given in Proposition~\ref{theo:misspec} in Appendix~\ref{sec:proof-prop-estimation}, along with a proof.

\subsubsection*{\textbf{Connection to practical algorithms.}} 
Recent progress has been made in addressing the issue of low utilization due to time-varying job resource requirements, notably within Google's datacenters, as discussed in \cite{BasDenRza_21}.
The approach in \cite{BasDenRza_21} makes predictions on the future resource requirements of jobs, which lead to a further prediction on the future peak resource requirement on a server if a newly arrived job were to be sent to that server (assuming no future job arrivals).
This prediction categorizes each server as either feasible or infeasible for the new job, and this binary outcome is subsequently used by a separate scheduler for job assignment.

Our proposed \metapolicy\ policy can be viewed as giving more detailed predictions on whether it is suitable for a server to take on new jobs, represented by the tokens.
The predictions are optimized by taking into account future job arrivals and the stochastic dynamics of jobs.

\section{Proof of Theorem~\ref{theo:sat:conversion-general} (Conversion Theorem) Assuming Irreducibility} \label{sec:sat:conv} 
In this section, we prove Theorem~\ref{theo:sat:conversion-general} to establish the performance guarantee of \metapolicy. 
For ease of presentation, we focus on the case where the subroutine policy $\sysbar{\policy}$ is $\vek^0$-irreducible. 
The proof for the general case is in \Cref{sec:proof-conv-general}. 

This section is organized as follows.
We first provide some preliminaries in \Cref{sec:conv-thm:preliminary}. Then we outline the steps and lemmas needed for the proof in \Cref{sec:conv-thm:steps-lemmas}. In \Cref{sec:conv-thm:proof-thm}, we prove \Cref{theo:sat:conversion-general} based on the lemmas. In \Cref{appx:proof-sat-w1-distance}, we prove one of the lemmas, \Cref{lem:sat:w1-distance}, where we devise a novel approach to employ Stein's method. Finally, in \Cref{sec:discuss-token-virtual-job}, we discuss the role of tokens and virtual jobs and their necessity from a proof perspective. 
The proofs of the rest of the lemmas presented in this section are given in \Cref{appx:conversion}. 

\subsection{\textbf{Preliminaries}}\label{sec:conv-thm:preliminary}
Consider an infinite-server system under the \metapolicy\ policy. For each normal server $\ell$, we describe its status at time $t$ using the following variables: \textit{configuration of real jobs} $\veK^\ell(t)$ (referred to simply as configuration in previous sections), \textit{tokens} $\vetoK^\ell(t)$, \textit{configuration of virtual jobs} $\vevj^\ell(t)$, and \textit{observed configuration} $\syshat{\veK}^\ell(t) \triangleq \veK^\ell(t) + \vevj^\ell(t)$. We use the superscript ``$\allL$'' to refer to a certain descriptor of all normal servers, for example, $\syshat{\veK}^\allL(t) \triangleq \big(\syshat{\veK}^\ell(t)\big)_ {\ell=1,2,\dots L}$. The system under \metapolicy\ is a Markov chain with a unique Markovian representation $((\veK^\ell(t))_{\ell=1, 2, \dots}, \vevj^\allL(t), \vetoK^\allL(t))$. The following lemma shows that the system has a unique stationary distribution (the proof is provided in \Cref{appx:proof-sat-pos-recur}). 

\begin{lemma}[Unique Stationary Distribution]\label{lem:sat:pos-recur}
    Consider an infinite-server system under the \metapolicy\ policy with $\sysbar{\policy}$ as its subroutine, where $\sysbar{\policy}$ is a single-server policy that is Markovian and $\vek^0$-irreducible. Then the state of the system $((\veK^\ell(t))_{\ell=1, 2, \dots}, \vevj^\allL(t), \vetoK^\allL(t))$ has a unique stationary distribution. 
\end{lemma}

Let $\sysbar{\veK}^\allL(t) \triangleq \big(\sysbar{\veK}^\ell(t)\big)_{\ell=1, 2, \dots, L}$ be the configuration vector of $L$ i.i.d.\ copies of the single-server system under $\sysbar{\policy}$. As discussed in Section~\ref{sec:sat:JRRS-irreducible}, we will show that $\veK^\allL(\infty)$ can be approximated by
$\sysbar{\veK}^\allL(\infty)$. 
In the remainder of this section, we omit the steady-state symbol $(\infty)$ for simplicity.

To rigorously discuss the approximation of the steady-state random variables, we define some metrics. Recall that $\Kinner \triangleq \{\vek\colon\sum_{i\in\jobspace} k_i \leq \Kmax\}$ is the set of feasible single-server configurations. Let $\Kinner^L\triangleq \{\vek^{1:L}\colon \vek^\ell \in \Kinner,\forall \ell \}$ be the set of feasible configurations for all normal servers. 
We use $\norm{\cdot}$ to denote the $L^1$ norm in both space $\Kinner$ and space $\Kinner^L$:
\begin{align*}
   \norm{\vek - \vek'} &= \textstyle \sum_{i\in\jobspace} \abs{k_i - k_i'}, \quad \text{for } \vek, \vek'\in\Kinner,\\
    \norm{\vek^\allL - \vek'^\allL} &= \textstyle \sumL \norm{\vek^\ell - \vek'^\ell}, \quad \text{for } \vek^\allL, \vek'^\allL \in \Kinner^L.
\end{align*}
For any two random variables $\ve{U}^{a}, \ve{U}^{b} \in \Kinner^L$, 
their closeness will be measured in terms of Wasserstein distance as follows:
\begin{equation*}
    d(\ve{U}^{a}, \ve{U}^{b}) \triangleq \sup_{f\in \text{Lip}(1)} \left\{ \E\left[f\left(\ve{U}^{a}\right)\right] - \E[f(\ve{U}^{b})] \right\},
\end{equation*}
where the supremum is taken over the all Lipschitz-$1$ functions from $\Kinner^L$ to $\mathbb{R}$.

\subsection{\textbf{Steps and Lemmas Needed for the Proof of Theorem~\ref{theo:sat:conversion-general} Assuming Irreducibility}}\label{sec:conv-thm:steps-lemmas} 
Our goal is to show that the steady-state distribution of the normal servers' real-job configurations $\veK^\allL$ is close to the steady-state distribution of i.i.d.\ copies of the single-server systems $\sysbar{\veK}^\allL$ in Wasserstein distance, and that the backup servers are almost empty as the arrival rate gets large.
More formally, we aim to show that $d(\veK^\allL, \sysbar{\veK}^\allL) = O(\sqrt{r})$ and $\sum_{\ell=L+1}^\infty \sum_{i\in\jobspace} K_i^\ell=O(\sqrt{r})$ as $r\to\infty$. 
These two bounds provide the performance guarantee claimed in \Cref{theo:sat:conversion-general}. 

Instead of directly looking into the distribution of real-job configuration $\veK^\allL$, we show that the distribution of each of the three sums, $\veK^\allL + \vevj^\allL + \vetoK^\allL$, $\veK^\allL + \vevj^\allL$, and $\veK^\allL$, can be approximated by the distribution of $\sysbar{\veK}^\allL$ in Wasserstein distance. The approximation result for each sum helps us derive the approximation result for the sum with one fewer term. The result that the backup servers are almost empty also follows from these approximations. 
This sequence of approximations is illustrated in the figure below, where recall that $\syshat{\veK}^\ell(t) \triangleq \veK^\ell(t) + \vevj^\ell(t)$ is the observed configuration.

\begin{figure}[h]
    \vspace{-2ex}
    \centering
    \includegraphics[height=1.2cm]{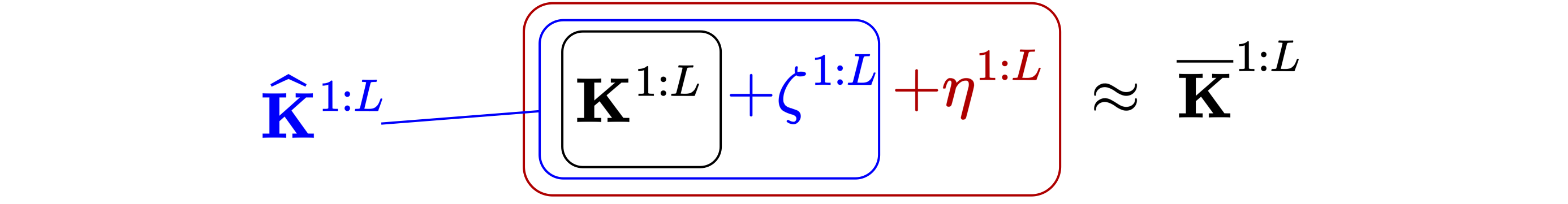}
    \caption*{}
    \label{fig:proof-illustration}
    \vspace{-8ex}
\end{figure}

A crucial observation that leads to this stepwise proof is that the process $(\syshat{\veK}^\allL(t), \vetoK^\allL(t))$ forms a Markov chain on its own. This is because real jobs and virtual jobs have the same transition dynamics and are indistinguishable by the subroutine when requesting jobs. 
Moreover, by the construction of \metapolicy, the Markov chain $(\syshat{\veK}^\allL(t), \vetoK^\allL(t))$ governs the dynamics of the virtual-job configurations $\vevj^\allL(t)$ and the configurations on backup servers. 

Our proof consists of two steps. In {\bf Step 1}, we focus on the process $(\syshat{\veK}^\allL(t), \vetoK^\allL(t))$. We show that $d(\syshat{\veK}^\allL + \vetoK^\allL, \sysbar{\veK}^\allL) = O(\sqrt{r})$, which immediately implies $d(\syshat{\veK}^\allL, \sysbar{\veK}^\allL) = O(\sqrt{r})$ because we have limited the total number of tokens to $O(\sqrt{r})$. 
In {\bf Step 2}, we use the approximation result for $\syshat{\veK}^\allL$ in {\bf Step 1} to show that the total number of virtual jobs, $\sum_{i\in\jobspace} \sumL \vj_i^\ell$, and the total number of jobs on backup servers are both $O(\sqrt{r})$. Recall that $\veK^\allL = \syshat{\veK}^\allL - \vevj^\allL$, so we get $d(\veK^\allL, \sysbar{\veK}^\allL)=O(\sqrt{r})$.

Next, we state the specific lemmas. 

\subparagraph*{\textbf{Step 1.}} Lemma~\ref{lem:sat:w1-distance} below bounds the Wasserstein distance between $\syshat{\veK}^\allL$ and $\sysbar{\veK}^\allL$.

\begin{lemma}\label{lem:sat:w1-distance} Under the conditions of Theorem \ref{theo:sat:conversion-general} and $\sysbar{\policy}$ being $\vek^0$-irreducible, we have
    \begin{equation*}
        d\left(\syshat{\veK}^\allL, \sysbar{\veK}^\allL \right) = \Obrac{\sqrt{r}}.
    \end{equation*}
\end{lemma}
The key challenge for proving \Cref{lem:sat:w1-distance} is that the job dispatching decisions are based on the configurations of \emph{all normal servers}, which creates dependencies among the transitions of different servers. 
The key idea that helps us overcome this challenge is to consider the sum $\syshat{\veK}^\allL+ \vetoK^\allL$, which remains unchanged under job arrivals regardless of dispatching decisions. Observe that $\syshat{\veK}^\allL+ \vetoK^\allL$ has  \emph{decoupled} dynamics across servers because it is only changed by internal phase transitions, departures, and requests of new jobs, which happen independently on each server. This helps us prove $d(\syshat{\veK}^\allL+ \vetoK^\allL, \sysbar{\veK}^\allL ) = \Obrac{\sqrt{r}}$, which implies \Cref{lem:sat:w1-distance}, as argued earlier in the section.

Formally, the proof of Lemma~\ref{lem:sat:w1-distance} makes use of Stein's method (see, e.g., \citep{BraDai_17, BraDaiFen_17, Bra_22}) to compare  $\syshat{\veK}^\allL+ \vetoK^\allL$ with $\sysbar{\veK}^\allL$. 
Stein's method usually consists of three steps: generator comparison, Stein factor bound, and moment bound. In our case, due to the finiteness of the state space $\Kinner$, we only need to do the generator comparison and the Stein factor bound. In the generator comparison step, we show that the instantaneous transition rates of $\syshat{\veK}^\allL+ \vetoK^\allL$ match with those of $\sysbar{\veK}^\allL$; in the Stein factor bound step, we show that small difference in the transition rates of $\syshat{\veK}^\allL+ \vetoK^\allL$ and $\sysbar{\veK}^\allL$ does not cause much increase in the overall distance of the distributions. The detailed proof is in \Cref{appx:proof-sat-w1-distance}.

\subparagraph{\textbf{Step 2.}} We establish Lemma~\ref{lem:sat:virtual-jobs} and Lemma~\ref{lem:sat:overflow-jobs} below, which bound the steady-state expected number of virtual jobs and the jobs on backup servers. 
\begin{lemma}\label{lem:sat:virtual-jobs} Under the conditions of Theorem \ref{theo:sat:conversion-general} and $\sysbar{\policy}$ being $\vek^0$-irreducible, for each $i\in\jobspace$, the steady-state expected number of virtual jobs of type $i$ is of the order $O(\sqrt{r})$, i.e.,
    \begin{equation*}
        \E\left[\textstyle \sumL \vj_i^\ell\right] = \Obrac{\sqrt{r}}.
    \end{equation*}
\end{lemma}

\begin{lemma}\label{lem:sat:overflow-jobs} Under the conditions of Theorem \ref{theo:sat:conversion-general} and $\sysbar{\policy}$ being $\vek^0$-irreducible, for each $i\in\jobspace$, the steady-state expected number of type $i$ jobs on backup servers is of the order $O(\sqrt{r})$, i.e.,
    \begin{equation*}
     \E\left[\textstyle \sum_{\ell=L+1}^\infty K_i^\ell\right]= \Obrac{\sqrt{r}}. 
    \end{equation*}
\end{lemma}

The key idea for proving \Cref{lem:sat:virtual-jobs} and \Cref{lem:sat:overflow-jobs} is that by the characterization of $\syshat{\veK}^\allL$ in \Cref{lem:sat:w1-distance} and the fact that the job requests are made based on $\syshat{\veK}^\allL$, we can show that the rate of requesting jobs is approximately equal to the arrival rate for each job type. Therefore, the number of tokens rarely reaches $0$ or $\toklim$. This implies the rarity of virtual jobs and jobs on backup servers. 
The proofs are provided in \Cref{appx:proof-sat-virtual-overflow-jobs}. 

\subsection{\textbf{Proof of Theorem~\ref{theo:sat:conversion-general} Assuming Irreducibility Based on Lemmas~\ref{lem:sat:pos-recur}--\ref{lem:sat:overflow-jobs}.}}\label{sec:conv-thm:proof-thm}
\begin{proof}
First we show that Lemmas~\ref{lem:sat:w1-distance} and \ref{lem:sat:virtual-jobs} imply the closeness between $\veK^\allL$ and $\sysbar{\veK}^\allL$. By Lemma~\ref{lem:sat:w1-distance}, for any $f\in\text{Lip}(1)$, we have
$
\E\big[f\big(\sysbar{\veK}^\allL\big)\big] - \E\big[f\big(\syshat{\veK}^\allL\big)\big] = \Obrac{\sqrt{r}}.
$ 
By Lemma~\ref{lem:sat:virtual-jobs}, $\E\left[\sumL \vj_i^\ell\right] = \Obrac{\sqrt{r}}$. Recall that $\syshat{\veK}^\ell = \veK^\ell + \vevj^\ell$, so $\E\big[f(\syshat{\veK}^\allL)\big] - \E\big[f(\veK^\allL)\big]=\Obrac{\sqrt{r}}$. Therefore, 
\begin{equation}\label{eq:sat:w1-distance-k-and-kbar}
    \E\left[f(\sysbar{\veK}^\allL)\right] - \E\left[f(\veK^\allL)\right] = \Obrac{\sqrt{r}}.
\end{equation}

Now we prove {$\eqref{eq:sat:active-servers}$}, the bound on the expected number of the active servers, by taking a suitable $f$ in \eqref{eq:sat:w1-distance-k-and-kbar}. Observe that 
\begin{align}
    \sum_{\vek\neq \vzero}\E\Big[X_{\vek}\Big] - L\cdot \Prob(\sysbar{\veK} \neq 0)
    &= \E\Big[\sumL \indibrac{\veK^\ell \neq \vzero}\Big] - \E\Big[\sumL \indibrac{\sysbar{\veK}^\ell \neq \vzero}\Big]  + \E\Big[\sum_{\ell=L+1}^\infty \indibrac{\veK^\ell\neq\vzero}\Big], \label{eq:sat:active-server-intermediate} 
\end{align}
where the last term on RHS is $\Obrac{\sqrt{r}}$ by \Cref{lem:sat:overflow-jobs}. To show that the difference between the first two terms on the RHS of \eqref{eq:sat:active-server-intermediate} are also $\Obrac{\sqrt{r}}$, consider $f_1(\vek^\allL) \triangleq \sumL \indibrac{\vek^\ell \neq \vzero}$. Because 
\begin{align*}
   \abs{f_1(\vek^\allL) - f_1(\vek'^\allL) }= \abs{\sumL \left(\indibrac{\vek^\ell \neq \vzero} - \indibrac{\vek'^\ell \neq \vzero} \right)}
    \leq \sumL \indibrac{\vek^\ell \neq \vek'^\ell}
    \leq \norm{\vek^\allL - \vek'^\allL},
\end{align*} 
for any $\vek^\allL, \vek'^\allL \in \Kinner^L$, we have $f_1\in\text{Lip}(1)$. 
By \eqref{eq:sat:w1-distance-k-and-kbar}, $\E\big[f_1(\veK^\ell)\big] - \E\big[\sumL f_1(\sysbar{\veK}^\ell)\big] = \Obrac{\sqrt{r}}$. Therefore, 
$\sum_{\vek\neq \vzero}\E\big[ X_{\vek}\big] - L\cdot \Prob\big(\sysbar{\veK} \neq 0\big) = \Obrac{\sqrt{r}}.$ Recall that $L = \ceil{\sysbar{\nact}}$, so we get $\eqref{eq:sat:active-servers}$. 

Similarly, to prove {\eqref{eq:sat:cost}}, we observe that
\begin{align}
    \sum_{\vek\neq\vzero} h(\vek) \E[X_{\vek}] - L \cdot \E\Big[h(\sysbar{\veK})\Big]
    &= \E\Big[\sumL h(\veK^\ell)\Big] - \E\Big[\sumL h(\sysbar{\veK}^\ell)\Big] + \E\Big[\sum_{\ell=L+1}^\infty h(\veK^\ell)\Big].   \label{eq:sat:resource-contention-intermediate} 
\end{align}
The last term of \eqref{eq:sat:resource-contention-intermediate} can be bounded as 
$
    \E\Big[\sum_{\ell=L+1}^\infty h(\veK^\ell)\Big] \leq \E\Big[\sum_{\ell=L+1}^\infty \indibrac{K_i^\ell\neq\vzero}\Big]\cdot \max_{\vek\in\Kouter} h(\vek)
$, which is $\Obrac{\sqrt{r}}$ by \Cref{lem:sat:overflow-jobs} and the fact that $\Kinner$ is a finite set. 
To show that the difference between the first two terms on the RHS of \eqref{eq:sat:resource-contention-intermediate} is also $\Obrac{\sqrt{r}}$, consider $f_2(\vek^\allL) = \frac{1}{\Gamma}\sumL h(\vek^\ell)$,
where $\Gamma$ is the Lipschitz constant of $h(\cdot)$. 
Because
\[
    \abs{f_2(\vek^\allL) - f_2(\vek'^\allL)} = \frac{1}{\Gamma} \abs{ \sumL (h(\vek^\ell) - h(\vek'^\ell)) }\leq \sumL \norm{\vek^\ell - \vek'^\ell} = \norm{\vek^\allL - \vek'^\allL},
\]
for any $\vek^\allL, \vek'^\allL \in \Kinner^L$, we have $f_2\in\text{Lip(1)}$. By \eqref{eq:sat:w1-distance-k-and-kbar},  $\E\big[f_2(\veK^\ell)\big] - \E\big[\sumL f_2(\sysbar{\veK}^\ell)\big] = \Obrac{\sqrt{r}}$. Therefore, $\sum_{\vek\neq\vzero} h(\vek) \E[X_{\vek}] - L \cdot \E\left[h(\sysbar{\veK})\right] = \Obrac{\sqrt{r}}$. Recall that $L = \left\lceil\sysbar{\nact}\right\rceil$, so we get $\eqref{eq:sat:cost}$.

To show \eqref{eq:alpha-optimal} and \eqref{eq:beta-optimal}, noting that $\ceil{\sysbar{\nact}} = \Thebrac{r}$, we have
    \begin{gather*}
        \nact(\policy) = \textstyle\sum_{\vek\neq\vzero} \E[X_{\vek}] \leq \ceil{\sysbar{\nact}} + \Obrac{\sqrt{r}} = \left(1+\Obrac{r^{-0.5}}\right) \cdot \sysbar{\nact}, \\
        \costp(\policy) = \frac{\sum_{\vek\neq\vzero} h(\vek)  \E[X_{\vek}]}{\sum_{\vek\neq\vzero} \E[X_{\vek}]}
        = \frac{\sum_{\vek\neq\vzero} h(\vek) \pi(\vek) + \Obrac{r^{-0.5}}}{1- \pi(\vzero) + \Obrac{r^{-0.5}}} 
        \leq \left(1+\Obrac{r^{-0.5}}\right) \cdot \budget,
    \end{gather*}
where in the last inequality we have used the fact that $\epsilon > 0$. 
This completes the proof.

\end{proof}

\subsection{More Details on the System and Proof of Lemma~\ref{lem:sat:w1-distance}}\label{appx:proof-sat-w1-distance}
To bound the distance between $\sysbar{\veK}^\allL$ and $\syshat{\veK}^\allL$, observe that because $f\in\text{Lip}(1)$ and $\sumL \sum_{i\in\jobspace} \toK_i^\ell = \Obrac{\sqrt{r}}$, it suffices to bound  bound the Wasserstein distance between $\sysbar{\veK}^\allL$ and $\syshat{\veK}^\allL+\vetoK^\allL$, i.e.,
\begin{equation}\label{eq:sat:lem1-subgoal}
    \sup_{f\in \text{Lip}(1)}\left\{ \E\left[f\big(\sysbar{\veK}^\allL\big)\right] - \E\left[f\big(\syshat{\veK}^\allL+\vetoK^\allL\big)\right]\right\} = \Obrac{\sqrt{r}},
\end{equation}
where $f\big(\syshat{\veK}^\allL+\vetoK^\allL\big)$ is a valid expression because $\syshat{\veK}^\allL+\vetoK^\allL \in \Kinner$ as discussed in \Cref{remark:domain-of-g} below.

\subsubsection{\textbf{More Details on System Dynamics and Generator}}
To prepare for the proof, we first look into the dynamics of the two systems under study. In particular, we  write out the \textit{generators} of $\sysbar{\veK}^\allL(t)$ and $(\syshat{\veK}^\allL(t), \vetoK^\allL(t))$, which are used in the Stein's method arguments. 

We first examine the dynamics of the single-server system under Markovian policy $\sysbar{\policy}$. 
Four types of events change a single-server system's configuration: internal transitions, departures, reactive requests, and proactive requests (see \Cref{sec:how-single-server-request-jobs}). 
The change of configuration due to any event can be represented by the diagram
\[
\vek \overset{}{\rightarrow} \vek' \overset{}{\rightarrow} \vek' + \vea, 
\]
where the arrow $\vek\to\vek'$ denotes an internal transition or a departure from configuration $\vek$ to $\vek'$ if $\vek\neq \vek'$; the arrow $\vek'\to\vek'+\vea$ denotes a reactive request that adds $\vea$ jobs to the system if $\vek\neq \vek'$, and denotes a proactive request if $\vek=\vek'$. 
We call the above change of configuration a \textit{transition}, and denote its rate as $\rate{\vek', \vea}{\vek}$. Let  $E(\vek)$ denote the set of possible $(\vek', \vea)$ pairs in a transition.

We define the total transition rate at configuration $\vek$ as
$
        \ratetot{\vek} \triangleq \sumrate{\vek',\vea}{\vek} \rate{\vek',\vea}{\vek},
$
and define the maximal transition rate $\ratetot{\max} = \max_{\vek\in\Kinner} \ratetot{\vek}$. Since $\Kinner$ is a finite set, we have $\ratetot{\max}<\infty$. Also, observe that the request rate of type $i$ jobs is given by
\begin{equation}
    \reqrate_{i} \triangleq \sum_{\vek} \sumrate{\vek', \vea}{\vek} \rate{\vek', \vea}{\vek} a_i  \cdot \pi(\vek),
\end{equation}
where $\pi$ denotes the stationary distribution of single-server configuration under policy $\sysbar{\policy}.$

Next, we focus on the dynamics of $L$ i.i.d. copies of single-server systems. Consider the generator $\sysbar{\gen}$ of the corresponding Markov chain $\{\sysbar{\veK}^\allL(t)\}$, which is a linear operator on functions $g\colon \Kouter^L \to \R$  defined as:
\begin{equation}
    \sysbar{\gen}g(\vek^\allL) \triangleq  \frac{d}{dt} \E\left[g\left( \sysbar{\veK}^\allL(t)\right) \middle| \sysbar{\veK}^\allL(0)=\vek^\allL \right] \Big|_{t=0},
\end{equation}
and we call the resulting function $\sysbar{\gen}g(\cdot)$ the \textit{drift} of $g(\cdot)$. Based on the transition rates defined above, we have
\begin{equation}
    \sysbar{\gen}g(\vek^\allL) = \sumL \sumrate{\vek', \vea}{\vek^\ell} \rate{\vek', \vea}{\vek^\ell} \left(g(\cdot, \vek'+\vea, \cdot) - g(\cdot, \vek^\ell, \cdot) \right)\label{eq:sat:gen-bar},
\end{equation}
where $g(\cdot, \vek'+\vea, \cdot) - g(\cdot, \vek^\ell, \cdot)$ is a shorthand for $g(\vek^1, \dots, \vek^{\ell-1}, \vek'+\vea, \vek^{\ell+1}, \dots, \vek^L) - g(\vek^\allL)$, i.e., we use $\cdot$ to omit the entries that agree with $\vek^\allL$.

Similarly, for the infinite-server system, consider the generator $\syshat{\gen}$ of $(\syshat{\veK}^\allL(t), \vetoK^\allL(t))$ defined as
\begin{equation}
    \syshat{\gen} \psi(\vek^\allL, \vetok^\allL)  \triangleq \frac{d}{dt} \E\left[\psi\left(\syshat{\veK}^\allL(t),\vetoK^\allL(t)\right) \middle|  \syshat{\veK}^\allL(0)=\vek^\allL, \veK^\allL(0)=\vetok^\allL \right] \Big|_{t=0},
\end{equation}
for any function $\psi\colon (\Kouter \times \Kouter)^L \to \R$. 
The drift of $\psi$ under $\syshat{\gen}$ turns out to have a similar decoupled form as $\sysbar{\gen}g$: observe that for each $\ell$, the transition of $(\syshat{\veK}^\ell(t), \vetoK^\ell(t))$ from $(\vek, \vetok)$ to $(\vek', \vetok+ \vea \indibrac{\vetok=\vzero})$ occurs at the rate $\rate{\vek',\vea}{\vek}$ for each $(\vek', \vea)\in E(\vek)$, and any real or virtual job arrivals do not change the sum $\syshat{\veK}^\ell(t) + \vetoK^\ell(t)$.
Consider any function $g\colon \Kinner^L \to \R$ and the function $\psi(\vek^\allL, \vetok^\allL) = g(\vek^\allL+\vetok^\allL)$. 
\begin{align}
    \syshat{\gen} \psi(\vek^\allL, \vetok^\allL)
    &= \sumL \sumrate{\vek', \vea}{\vek^\ell} \rate{\vek', \vea}{\vek^\ell} \left(g(\cdot, \vek'+\vea, \cdot) - g(\cdot, \vek^\ell, \cdot) \right) \indibrac{\vetok^\ell = \vzero} \nonumber\\
    &\mspace{23mu}+\sumL \sumrate{\vek', \vea}{\vek^\ell} \rate{\vek', \vea}{\vek^\ell} \left(g(\cdot, \vek'+\vetok^\ell, \cdot) - g(\cdot, \vek^\ell+\vetok^\ell, \cdot) \right)\indibrac{\vetok^\ell\neq \vzero}
    \label{eq:sat:gen-hat}.
\end{align}
In this context $g(\cdot, \vek'+\vetok^\ell, \cdot) - g(\cdot, \vek^\ell+\vetok^\ell, \cdot)$ is a shorthand for $g(\vek^1+\vetok^1, \dots, \vek^{\ell-1}+\vetok^{\ell-1}, \vek'+\vetok^\ell, \vek^{\ell+1}+\vetok^{\ell+1}, \dots, \vek^L+\vetok^L) - g(\vek^\allL+\vetok^\allL)$. In other words, we use $\cdot$ to omit the entries of $g$'s input that agree with the corresponding entries of $\vek^\allL+\vetok^\allL$.

\begin{remark}\label{remark:domain-of-g}
    In \eqref{eq:sat:gen-hat}, although $g$ is only defined on the domain $\Kinner^L$, it is valid to write $\vek^\allL + \vetok^\allL$ as its input because we always have $\syshat{\veK}^\ell + \vetoK^\ell \in \Kinner$, i.e., the total number of real jobs, virtual jobs, and tokens on a normal server never exceeds $\Kmax$. To see why this is true, the single-server policy $\sysbar{\policy}$ requests jobs only when there are no tokens on the server, and it will not request more than $\Kmax - n$ jobs if there are already $n$ real and virtual jobs on the server. 
\end{remark}

\subsubsection{\textbf{Proof of Lemma~\ref{lem:sat:w1-distance}.}}
\begin{proof}
\textbf{Generator Comparison.}
    For any $f\in\text{Lip}(1)$, consider the Poisson equation (see, e.g., \citep{Bra_22}) that solves for $g_f \colon \Kouter^L \to \R$:
    \begin{equation}\label{eq:sat:lem1-poisson-eq}
        \E\left[f\big(\sysbar{\veK}^\allL\big)\right] - f\big(\vek^\allL\big) = \sysbar{\gen} g_f\big(\vek^\allL\big).
    \end{equation}
    We let $\vek^\allL = \syshat{\veK}^\allL + \vetoK^\allL$ in \eqref{eq:sat:lem1-poisson-eq} and take the expectation. This results in
    \begin{equation}\label{eq:sat:lem1-gencomp-1}
        \E\left[f\big(\sysbar{\veK}^\allL\big)\right] - \E\left[f\big(\syshat{\veK}^\allL + \vetoK^\allL\big)\right] = \E\left[\sysbar{\gen} g_f\big(\syshat{\veK}^\allL + \vetoK^\allL\big)\right].
    \end{equation}
    On the other hand, because $(\syshat{\veK}^\allL(t), \vetoK^\allL(t))$ is a finite-state Markov chain, 
    we have 
    \begin{equation}\label{eq:sat:lem1-gencomp-2}
        \E\left[\syshat{\gen}\psi_f\big(\syshat{\veK}^\allL, \vetoK^\allL\big)\right] = 0,
    \end{equation}
    where $\psi_f$ is given by 
    $
        \psi_f\big(\syshat{\veK}^\allL, \vetoK^\allL\big) = g_f(\syshat{\veK}^\allL + \vetoK^\allL).
    $
    Subtracting \eqref{eq:sat:lem1-gencomp-2} from \eqref{eq:sat:lem1-gencomp-1}, we get
    \begin{equation}\label{eq:sat:lem1-gencomp}
        \E\left[f\big(\sysbar{\veK}^\allL\big)\right] - \E\left[f\big(\syshat{\veK}^\allL + \vetoK^\allL\big)\right] = \E\left[\sysbar{\gen} g_f\big(\syshat{\veK}^\allL + \vetoK^\allL \big) -\syshat{\gen} \psi_f\big(\syshat{\veK}^\allL, \vetoK^\allL\big)\right].
    \end{equation}
    We want to show that $\sysbar{\gen}$ and $\syshat{\gen}$ are close so that we can bound the RHS of \eqref{eq:sat:lem1-gencomp}.
    
    Now we plug the formula of the generators in \eqref{eq:sat:gen-bar} and \eqref{eq:sat:gen-hat} into the RHS of \eqref{eq:sat:lem1-gencomp} and get 
    \begin{align}
         &\mspace{18mu} \abs{\sysbar{\gen} g_f\big(\syshat{\veK}^\allL + \vetoK^\allL \big) -\syshat{\gen} \psi_f\big(\syshat{\veK}^\allL, \vetoK^\allL\big) } \nonumber \\
         &\stackrel{\text{(i)}}{=} \Bigg\lvert \sumL \sumrate{\vek', \vea}{\syshat{\veK}^\ell+\vetoK^\ell} \rate{\vek', \vea}{\syshat{\veK}^\ell+\vetoK^\ell} \left(g_f\big(\cdot, \vek'+\vea, \cdot\big) - g_f\big(\cdot, \syshat{\veK}^\ell+\vetoK^\ell, \cdot\big) \right)\cdot \indibrac{\vetoK^\ell\neq\vzero} \nonumber \\
         &\mspace{23mu}-\sumL \sumrate{\vek', \vea}{\syshat{\veK}^\ell} \rate{\vek', \vea}{\syshat{\veK}^\ell} \left(g_f\big(\cdot, \vek'+\vetoK^\ell, \cdot\big) - g_f\big(\cdot, \syshat{\veK}^\ell+\vetoK^\ell, \cdot\big) \right) \cdot \indibrac{\vetoK^\ell\neq\vzero} \Bigg\rvert \nonumber\\
         &\stackrel{\text{(ii)}}{\leq} \sumL \ratetot{\syshat{\veK}^\ell+\vetoK^\ell} \cdot \sup_{(\vek',\vea)\in E(\syshat{\veK}^\ell+\vetoK^\ell)} \abs{g_f\big(\cdot, \vek'+\vea, \cdot\big) - g_f\big(\cdot, \syshat{\veK}^\ell+\vetoK^\ell, \cdot\big) } \cdot \indibrac{\vetoK^\ell\neq\vzero} \nonumber \\
         &\mspace{23mu}+ \sumL \ratetot{\syshat{\veK}^\ell} \cdot \sup_{(\vek',\vea)\in E(\syshat{\veK}^\ell)} \abs{g_f\big(\cdot, \vek'+\vetoK^\ell, \cdot\big) - g_f\big(\cdot, \syshat{\veK}^\ell+\vetoK^\ell, \cdot\big) } \cdot \indibrac{\vetoK^\ell\neq\vzero} \nonumber \\
         &\stackrel{\text{(iii)}}{\leq} 2 \ratetot{\max} \cdot \sumL \sup_{\vek'\in\Kouter} \abs{g_f\big(\cdot, \vek', \cdot\big) - g_f\big(\cdot, \syshat{\veK}^\ell+\vetoK^\ell, \cdot \big)}  \cdot \indibrac{\vetoK^\ell\neq\vzero}  \nonumber \\
         &\leq 2 \ratetot{\max} \cdot \sup_{\vek, \vek'\in\Kouter} \abs{g_f\big(\cdot, \vek', \cdot\big) - g_f\big(\cdot, \vek, \cdot \big)}  \cdot \sumL \indibrac{\vetoK^\ell\neq\vzero}, \label{eq:sat:gencomp-long-calculation-intermediate}
    \end{align}
    where in $g_f\big(\cdot, \vek', \cdot\big) - g_f\big(\cdot, \vek, \cdot \big)$ we have omitted the entries that agree with $\vek^\allL$. The equality (i) is true because each of the $\ell$-th terms in $\sysbar{\gen}$ and $\syshat{\gen}$ is equal if $\vetoK^\ell=\vzero$. 
    For the inequalities (ii) and (iii), recall that $\ratetot{\vek}$ is the total transition rate given by $\ratetot{\vek} \triangleq \sumrate{\vek',\vea}{\vek} \rate{\vek',\vea}{\vek}$,
    and $\ratetot{\max} = \max_{\vek\in\Kinner} \ratetot{\vek}$. Observe that    
    \begin{equation*}
        \sumL \indibrac{\vetoK^\ell\neq\vzero} \leq \sumL \sum_{i\in\jobspace} \toK_i^\ell \leq \abs{\jobspace} \cdot \toklim = \Obrac{\sqrt{r}}.
    \end{equation*}
    Therefore \eqref{eq:sat:gencomp-long-calculation-intermediate} can be further bounded by
    \begin{align*}
        \abs{\sysbar{\gen} g_f\big(\syshat{\veK}^\allL + \vetoK^\allL \big) -\syshat{\gen} \psi_f\big(\syshat{\veK}^\allL, \vetoK^\allL\big) }
        &\leq 2 \ratetot{\max} \cdot \sup_{\vek, \vek'\in\Kouter} \abs{g_f\big(\cdot, \vek', \cdot\big) - g_f\big(\cdot, \vek, \cdot \big)}  \cdot \sumL \sum_{i\in\jobspace} \toK_i^\ell \\
        &\leq 2 \ratetot{\max} \cdot \sup_{\vek, \vek'\in\Kouter} \abs{g_f\big(\cdot, \vek', \cdot\big) - g_f\big(\cdot, \vek, \cdot \big)} \cdot \Obrac{\sqrt{r}}.
    \end{align*}
    To prove \eqref{eq:sat:lem1-subgoal}, it remains to show that
    \begin{equation}\label{eq:sat:grad-moment-bound}
       \sup_{\vek, \vek'\in\Kouter} \abs{g_f(\cdot, \vek', \cdot) - g_f(\cdot, \vek, \cdot )} = \Obrac{1}.
    \end{equation}

\textbf{Stein Factor Bound.}
    To prove \eqref{eq:sat:grad-moment-bound}, observe that the following $g_f(\cdot)$ is a solution to the Poisson equation \eqref{eq:sat:lem1-poisson-eq}:
    \begin{equation}\label{eq:sat:lem1-poisson-sol}
        g_f\big(\vek^\allL\big) = \int_0^\infty \E\left[ \left(f\big(\sysbar{\veK}^\allL(t)\big)- \E\left[f\big(\sysbar{\veK}^\allL\big)\right]\right)\middle| \sysbar{\veK}^\allL(0) = \vek^\allL \right] dt.
    \end{equation}
    This allows us to bound the difference of $g_f$ using coupling. Specifically, we define the coupling of two systems, each consisting of $L$ i.i.d. copies of the single-server system under $\sysbar{\policy}$. The two systems are initialized with configurations $(\cdot,  \vek', \cdot)$ and $(\cdot, \vek, \cdot)$ that only differ at the $\ell$-th server, where we omit the entries that agree with $\vek^\allL$. Let $\big( \sysbar{\veK}^{\allL,1}(t), \sysbar{\veK}^{\allL,2}(t)\big)$ be the joint configuration of the two systems, which is actually $2L$ i.i.d. copies of the single-server system. As a result, we can specify the couplings $(\sysbar{\veK}^{\ell', 1}(t), \sysbar{\veK}^{\ell', 2}(t))$ for different $\ell'$ separately. For $\ell' \neq \ell$, the corresponding server in the two systems have the same initial configurations, so we can always keep their configurations identical. For the $\ell$-th servers, we let them evolve independently following their own dynamics until a stopping time $\tmix$ when their configurations become the same. After that, we can use coupling to keep their configurations identical. Under this coupling, it is not hard to see that
    \begin{align}
        \abs{g_f(\cdot, \vek', \cdot) - g_f(\cdot, \vek, \cdot)} &=\abs{\int_0^\infty \E\left[ f\big(\sysbar{\veK}^{\allL,1}(t)\big)- f\big(\sysbar{\veK}^{\allL,2}(t)\big) \right]dt}\nonumber \\
        &\leq  \E\left[ \int_0^\infty \abs{f\big(\sysbar{\veK}^{\allL,1}(t)\big)- f\big(\sysbar{\veK}^{\allL,2}(t)\big)} dt \right]\nonumber \\
        &\leq \E\left[ \int_0^\infty  
        \sum_{\ell'=1}^L \big\lVert \sysbar{\veK}^{\ell',1}(t)-  \sysbar{\veK}^{\ell',2}(t) \big\rVert dt \right] \nonumber\\
        &= \E\left[\int_0^{\tmix}  \big\lVert \sysbar{\veK}^{\ell,1}(t)- \sysbar{\veK}^{\ell,2}(t) \big\rVert dt\right] \label{eq:sat:gradient-bound-intermediate},
    \end{align}
    where in the second inequality we have used the fact that $f$ is $1$-Lipschitz continuous under the $L^1$ norm of the space $\Kinner^L$. 
    For each pair of $\vek, \vek'$, observe that because $\sysbar{\policy}$ is a $\vek^0$-irreducible policy, $\E[\tmix]$ is finite; 
    and because $\Kinner$ is a finite set, $\big\lVert \sysbar{\veK}^{\ell,1}(t)- \sysbar{\veK}^{\ell,2}(t) \big\rVert$ is uniformly bounded. {All these finite quantities depend on a single-server system under a policy $\sysbar{\policy}$ that is independent of $r$.} As a result, the last expression in \eqref{eq:sat:gradient-bound-intermediate} is of constant order. Moreover, because there are finite pairs of $(\vek, \vek')$, the supremum
    $\sup_{\vek, \vek'} \E\left[\int_0^{\tmix}  \norm{ \sysbar{\veK}^{\ell,1}(t)- \sysbar{\veK}^{\ell,2}(t)} dt\right]$ is also of constant order, independent of $r$. This proves the Stein factor bound in  \eqref{eq:sat:grad-moment-bound}. Together with the generator comparison, we have proved 
    \begin{equation}\tag{\ref{eq:sat:lem1-subgoal}}
         \sup_{f\in \text{Lip}(1)} \E\left[f\big(\sysbar{\veK}^\allL\big)\right] - \E\left[f\big(\syshat{\veK}^\allL+\vetoK^\allL\big)\right] = \Obrac{\sqrt{r}}.
    \end{equation}
    
    Because $\sumL \sum_{i\in\jobspace} \toK_i^\ell \leq \abs{\jobspace} \toklim = \Obrac{\sqrt{r}}$, for any $f\in \text{Lip}(1)$, we have
    \begin{equation}
        \left\lvert\E\left[f\big(\syshat{\veK}^\allL\big)\right] - \E\left[f\big(\syshat{\veK}^\allL+\vetoK^\allL\big)\right]\right\rvert \leq \E\left[\sumL \sum_{i\in\jobspace} \toK_i^\ell(t)\right] = \Obrac{\sqrt{r}}.
    \end{equation}
    Plugging the above equation to \eqref{eq:sat:lem1-subgoal}, we get 
    $
        \sup_{f\in \text{Lip}(1)} \E\left[f\big(\sysbar{\veK}^\allL\big)\right] - \E\left[f\big(\syshat{\veK}^\allL\big)\right] = \Obrac{\sqrt{r}}.
    $
    This proves Lemma~\ref{lem:sat:w1-distance}. 
\end{proof}

\subsection{Role of Tokens and Virtual Jobs}\label{sec:discuss-token-virtual-job}
This section aims to shed light on the role of tokens and virtual jobs in the proposed policy, \metapolicy.
We first outline how we devise the token-and-virtual-job mechanism from the perspective of generators. 
To begin with, consider the scenario where dispatch decisions are solely based on real-job configurations $(\veK^\ell)_{\ell=1, 2, \dots}$.
In this case, the transitions of servers' configurations would be correlated in general due to job arrivals, which are dispatched based on the joint configuration of all servers. 
To break this correlation, let us introduce tokens but not set an upper limit yet on the number of tokens (thus no virtual jobs). 
Observe that $\veK^\allL(t)+\vetoK^\allL(t)$ remains unchanged by job arrivals, so tokens remove the correlation brought about by job arrivals. 
However, because tokens lack internal phase transitions or departures, the transition dynamics of $\veK^\allL(t)+\vetoK^\allL(t)$ will diverge from $\sysbar{\veK}^\allL(t)$ when $\vetoK(t)$ is large. In other words, although tokens help decouple the transitions on servers, they cannot keep the transitions of $\veK^\allL(t)+\vetoK^\allL(t)$ close to $\sysbar{\veK}^\allL(t)$. 
To solve this issue, we finally introduce the mechanism of converting tokens into virtual jobs when the number of tokens is high, where virtual jobs can make internal phase transitions or departures just like real jobs.
Now, the sum $\veK^\allL(t) + \vevj^\allL(t)+\vetok^\allL(t)$ remains unchanged by job arrivals nor the creation of virtual jobs, and the internal phase transitions and job departures are similar to those of $\sysbar{\veK}^\allL(t)$. 
More formally, the generators of $\veK^\allL(t) + \vevj^\allL(t)+\vetok^\allL(t)$ and $\sysbar{\veK}^\allL(t)$ are close to each other -- their additive difference can be upper bounded by a quantity proportional to the expected number of tokens, as shown in \eqref{eq:sat:gencomp-long-calculation-intermediate}. Therefore, by regulating the number of tokens, we can control the difference between the generators of $\veK^\allL(t) + \vevj^\allL(t)+\vetok^\allL(t)$ and $\sysbar{\veK}^\allL(t)$. 

Another key design component of \metapolicy\ is that the subroutine requests jobs based on the observed configurations.
This has been used in the proof of \Cref{lem:sat:w1-distance} to show that the observed configurations $\veK^\allL + \vevj^\allL$ are close to $\sysbar{\veK}^\allL$, which consist of $L$ i.i.d.\ single-server systems.
Recall that each single-server system in $\sysbar{\veK}^\allL$ is designed to have a throughput of $\lambda_ir/L$ for each job type $i\in\jobspace$.
Therefore, the proximity between $\veK^\allL + \vevj^\allL$ and $\sysbar{\veK}^\allL$ ensures that the job request rate mirrors the arrival rate for each job type, regardless of the real-job configurations. The fact that these two rates are approximately equal is important for proving \Cref{lem:sat:virtual-jobs} and \Cref{lem:sat:overflow-jobs}.
It guarantees that both the rate of generating virtual jobs (when there are too many tokens) and the rate of dispatching jobs to backup servers (when there are no tokens) are appropriate.

A natural follow-up question is whether the usage of tokens and virtual jobs is fundamental or an artifact of our analysis technique.
For example, it is unclear whether removing the upper limit on the number of tokens would still yield an asymptotically optimal policy. 
This is an interesting question that we do not have a complete answer to.
The token-and-virtual-job mechanism emerges as a natural choice under our analysis framework.
Nevertheless, it is worth noting that our analysis primarily treats each server's configuration as a generic Markov chain, without utilizing many properties specific to the stochastic bin-packing setting. An exception to this is the proofs in \Cref{appx:proof-sat-virtual-overflow-jobs}, where we use the model that each job leaves the system within a constant expected time.
It would be interesting to explore more properties of the problem to better understand policy designs without auxiliary state variables like tokens and virtual jobs.

\section{Conclusion}\label{sec:conclusion}
In this paper, we study a new setting of stochastic bin-packing in service systems that features time-varying item sizes.
Since our formulation is motivated by the problem of virtual-machine scheduling in computing systems, we use the terminology of jobs and servers, where jobs are viewed as items, whose sizes are their resource requirements, and servers as bins.
The time-varying item sizes capture the emerging trend in practice that jobs' resource requirements vary over time.
Our goal is to design a job dispatch policy to minimize the expected number of active servers in steady state, subject to a constraint on resource contentions. 
Our main result is the design of a policy that achieves an optimality gap of $O(\sqrt{r})$, where $r$ is the scaling factor of the arrival rate. When specialized to the setting where jobs' resource requirements remain fixed over time, this result improves upon the state-of-the-art $o(r)$ optimality gap. 
Our technical approach highlights a novel policy conversion framework, \metapolicyfull, that reduces the policy design problem to that in a single-server system. 

There are several potential directions that may be worth further exploration. One direction is to strengthen the optimality result within the current setting. Specifically, it is interesting to investigate: (i)~whether it is possible to achieve an optimality gap smaller than $\Theta(\sqrt{r})$; and (ii)~whether there exist asymptotically optimal policies whose average cost rate of resource contention satisfies the budget strictly instead of asymptotically. 

We are also interested in extending our technique to the optimal control of other systems with similar structures. Intuitively, this technique could be applied to systems with many components that evolve mostly independently but are weakly coupled by certain constraints. Viewing each component as a server, we can define a suitable single-server problem and then design a policy for the original system to track the dynamics of the optimal single-server solution. Below we list several variations of our model that can potentially be analyzed using the proposed technique.
\begin{itemize}[leftmargin=1.1em]
    \item A model where jobs running on each server will be put into a local queue when there are resource contentions. The goal thus becomes finding the optimal trade-offs between the number of active servers and the waiting time of the jobs. 
    \item A model that allows each server to have a Markovian state that affects the dynamics of the jobs running on the server. 
    \item A model that allows jobs to migrate to different servers at the cost of migration delays.
    \item A closed-system model where jobs re-enter the system after completion. 
\end{itemize}

A third possible direction is to tackle the problem when the arrival rates and the parameters in the job model are unknown, as mentioned in Section~\ref{sec:sat:JRRS-discussion}. A possible approach is to develop an approximate version of the \metapolicy\ framework, where the optimal single-server policy and the simulator for the virtual jobs are both learned from data. It is desirable to design such an approximate framework whose  performance degrades gracefully as the approximation error increases.

\begin{acks}
    Y.\ Hong and W.\ Wang are supported in part by NSF grants CNS-200773 and ECCS-2145713. 
    Q.\ Xie is supported in part by NSF grant CNS-1955997. 
\end{acks}

\bibliographystyle{ACM-Reference-Format}
\bibliography{refs-yige-v220908}


\begin{thebibliography}{39}


\ifx \showCODEN    \undefined \def \showCODEN     #1{\unskip}     \fi
\ifx \showDOI      \undefined \def \showDOI       #1{#1}\fi
\ifx \showISBNx    \undefined \def \showISBNx     #1{\unskip}     \fi
\ifx \showISBNxiii \undefined \def \showISBNxiii  #1{\unskip}     \fi
\ifx \showISSN     \undefined \def \showISSN      #1{\unskip}     \fi
\ifx \showLCCN     \undefined \def \showLCCN      #1{\unskip}     \fi
\ifx \shownote     \undefined \def \shownote      #1{#1}          \fi
\ifx \showarticletitle \undefined \def \showarticletitle #1{#1}   \fi
\ifx \showURL      \undefined \def \showURL       {\relax}        \fi
\providecommand\bibfield[2]{#2}
\providecommand\bibinfo[2]{#2}
\providecommand\natexlab[1]{#1}
\providecommand\showeprint[2][]{arXiv:#2}

\bibitem[Ayyadevara et~al\mbox{.}(2022)]%
        {AyyDabKha_22}
\bibfield{author}{\bibinfo{person}{Nikhil Ayyadevara}, \bibinfo{person}{Rajni
  Dabas}, \bibinfo{person}{Arindam Khan}, {and} \bibinfo{person}{K.~V.~N.
  Sreenivas}.} \bibinfo{year}{2022}\natexlab{}.
\newblock \showarticletitle{{Near-Optimal Algorithms for Stochastic Online Bin
  Packing}}. In \bibinfo{booktitle}{\emph{Proc. Int. Conf. Automata, Languages
  and Programming (ICALP)}}, Vol.~\bibinfo{volume}{229}.
  \bibinfo{pages}{12:1--12:20}.
\newblock


\bibitem[Bashir et~al\mbox{.}(2021)]%
        {BasDenRza_21}
\bibfield{author}{\bibinfo{person}{Noman Bashir}, \bibinfo{person}{Nan Deng},
  \bibinfo{person}{Krzysztof Rzadca}, \bibinfo{person}{David Irwin},
  \bibinfo{person}{Sree Kodak}, {and} \bibinfo{person}{Rohit Jnagal}.}
  \bibinfo{year}{2021}\natexlab{}.
\newblock \showarticletitle{Take It to the Limit: Peak Prediction-Driven
  Resource Overcommitment in Datacenters}. In \bibinfo{booktitle}{\emph{Proc.
  European Conf. Computer Systems (EuroSys)}}. \bibinfo{address}{Online Event,
  United Kingdom}, \bibinfo{pages}{556--573}.
\newblock


\bibitem[Braverman(2022)]%
        {Bra_22}
\bibfield{author}{\bibinfo{person}{Anton Braverman}.}
  \bibinfo{year}{2022}\natexlab{}.
\newblock \showarticletitle{The Prelimit Generator Comparison Approach of
  Stein's Method}.
\newblock \bibinfo{journal}{\emph{Stoch. Syst.}} \bibinfo{volume}{12},
  \bibinfo{number}{2} (\bibinfo{year}{2022}), \bibinfo{pages}{181--204}.
\newblock


\bibitem[Braverman and Dai(2017)]%
        {BraDai_17}
\bibfield{author}{\bibinfo{person}{Anton Braverman} {and}
  \bibinfo{person}{J.~G. Dai}.} \bibinfo{year}{2017}\natexlab{}.
\newblock \showarticletitle{Stein's method for steady-state diffusion
  approximations of $M/\mathit{Ph}/n+M$ systems}.
\newblock \bibinfo{journal}{\emph{Ann. Appl. Probab.}}  \bibinfo{volume}{27}
  (\bibinfo{date}{Feb.} \bibinfo{year}{2017}), \bibinfo{pages}{550--581}.
\newblock


\bibitem[Braverman et~al\mbox{.}(2017)]%
        {BraDaiFen_17}
\bibfield{author}{\bibinfo{person}{Anton Braverman}, \bibinfo{person}{J.~G.
  Dai}, {and} \bibinfo{person}{Jiekun Feng}.} \bibinfo{year}{2017}\natexlab{}.
\newblock \showarticletitle{Stein's method for steady-state diffusion
  approximations: an introduction through the {Erlang-A} and {Erlang-C}
  models}.
\newblock \bibinfo{journal}{\emph{Stoch. Syst.}} \bibinfo{volume}{6},
  \bibinfo{number}{2} (\bibinfo{year}{2017}), \bibinfo{pages}{301--366}.
\newblock


\bibitem[Buchbinder et~al\mbox{.}(2021)]%
        {BucFaiMel_21}
\bibfield{author}{\bibinfo{person}{Niv Buchbinder}, \bibinfo{person}{Yaron
  Fairstein}, \bibinfo{person}{Konstantina Mellou}, \bibinfo{person}{Ishai
  Menache}, {and} \bibinfo{person}{Joseph~(Seffi) Naor}.}
  \bibinfo{year}{2021}\natexlab{}.
\newblock \showarticletitle{Online Virtual Machine Allocation with Lifetime and
  Load Predictions}.
\newblock \bibinfo{journal}{\emph{ACM SIGMETRICS Perform. Eval. Rev.}}
  \bibinfo{volume}{49}, \bibinfo{number}{1} (\bibinfo{date}{May}
  \bibinfo{year}{2021}), \bibinfo{pages}{9--10}.
\newblock


\bibitem[Cloud(2023a)]%
        {GoogleCloud_23}
\bibfield{author}{\bibinfo{person}{Google Cloud}.}
  \bibinfo{year}{2023}\natexlab{a}.
\newblock \bibinfo{title}{Overcommitting {CPU}s on sole-tenant {VM}s}.
\newblock
  \bibinfo{howpublished}{\url{https://cloud.google.com/compute/docs/nodes/overcommitting-cpus-sole-tenant-vms}}.
\newblock


\bibitem[Cloud(2023b)]%
        {GoogleCloudVM_23}
\bibfield{author}{\bibinfo{person}{Google Cloud}.}
  \bibinfo{year}{2023}\natexlab{b}.
\newblock \bibinfo{title}{Virtual machine instances}.
\newblock
  \bibinfo{howpublished}{\url{https://cloud.google.com/compute/docs/instances}}.
\newblock


\bibitem[Coffman et~al\mbox{.}(1983)]%
        {CofGarJoh_83}
\bibfield{author}{\bibinfo{person}{E.~G. Coffman, Jr.}, \bibinfo{person}{M.~R.
  Garey}, {and} \bibinfo{person}{D.~S. Johnson}.}
  \bibinfo{year}{1983}\natexlab{}.
\newblock \showarticletitle{Dynamic Bin Packing}.
\newblock \bibinfo{journal}{\emph{SIAM J. Comput.}} \bibinfo{volume}{12},
  \bibinfo{number}{2} (\bibinfo{year}{1983}), \bibinfo{pages}{227--258}.
\newblock


\bibitem[Courcobetis and Weber(1990)]%
        {CouWeb_90}
\bibfield{author}{\bibinfo{person}{Coastas Courcobetis} {and}
  \bibinfo{person}{Richard Weber}.} \bibinfo{year}{1990}\natexlab{}.
\newblock \showarticletitle{Stability of On-Line Bin Packing with Random
  Arrivals and Long-Run-Average Constraints}.
\newblock \bibinfo{journal}{\emph{Probab. Eng. Inf. Sci.}} \bibinfo{volume}{4},
  \bibinfo{number}{4} (\bibinfo{year}{1990}), \bibinfo{pages}{447–460}.
\newblock


\bibitem[Courcoubetis and Weber(1986)]%
        {CouWeb_86}
\bibfield{author}{\bibinfo{person}{C. Courcoubetis} {and}
  \bibinfo{person}{R.~R. Weber}.} \bibinfo{year}{1986}\natexlab{}.
\newblock \showarticletitle{Necessary and Sufficient Conditions for Stability
  of a Bin-Packing System}.
\newblock \bibinfo{journal}{\emph{J. Appl. Probab.}} \bibinfo{volume}{23},
  \bibinfo{number}{4} (\bibinfo{year}{1986}), \bibinfo{pages}{989--999}.
\newblock


\bibitem[Csirik et~al\mbox{.}(2006)]%
        {CsiJohKen_06}
\bibfield{author}{\bibinfo{person}{Janos Csirik}, \bibinfo{person}{David~S.
  Johnson}, \bibinfo{person}{Claire Kenyon}, \bibinfo{person}{James~B. Orlin},
  \bibinfo{person}{Peter~W. Shor}, {and} \bibinfo{person}{Richard~R. Weber}.}
  \bibinfo{year}{2006}\natexlab{}.
\newblock \showarticletitle{On the Sum-of-Squares Algorithm for Bin Packing}.
\newblock \bibinfo{journal}{\emph{J. ACM}} \bibinfo{volume}{53},
  \bibinfo{number}{1} (\bibinfo{date}{Jan.} \bibinfo{year}{2006}),
  \bibinfo{pages}{1–65}.
\newblock


\bibitem[Delimitrou and Kozyrakis(2014)]%
        {DelKoz_14}
\bibfield{author}{\bibinfo{person}{Christina Delimitrou} {and}
  \bibinfo{person}{Christos Kozyrakis}.} \bibinfo{year}{2014}\natexlab{}.
\newblock \showarticletitle{Quasar: Resource-Efficient and {QoS}-Aware Cluster
  Management}. In \bibinfo{booktitle}{\emph{Proc. Int. Conf. Architectural
  Support for Programming Languages and Operating Systems (ASPLOS)}}.
  \bibinfo{address}{Salt Lake City, UT}, \bibinfo{pages}{127--144}.
\newblock


\bibitem[Foundation(2023a)]%
        {ApacheMesosContainer_23}
\bibfield{author}{\bibinfo{person}{Apache~Software Foundation}.}
  \bibinfo{year}{2023}\natexlab{a}.
\newblock \bibinfo{title}{Apache Mesos: Containerizers}.
\newblock
  \bibinfo{howpublished}{\url{https://mesos.apache.org/documentation/latest/containerizers/}}.
\newblock


\bibitem[Foundation(2023b)]%
        {ApacheMesos_23}
\bibfield{author}{\bibinfo{person}{Apache~Software Foundation}.}
  \bibinfo{year}{2023}\natexlab{b}.
\newblock \bibinfo{title}{Apache Mesos: Oversubscription}.
\newblock
  \bibinfo{howpublished}{\url{https://mesos.apache.org/documentation/latest/oversubscription/}}.
\newblock


\bibitem[Freund and Banerjee(2019)]%
        {FreBan_19}
\bibfield{author}{\bibinfo{person}{Daniel Freund} {and}
  \bibinfo{person}{Siddhartha Banerjee}.} \bibinfo{year}{2019}\natexlab{}.
\newblock \showarticletitle{Good prophets know when the end is near}.
\newblock \bibinfo{journal}{\emph{Available at SSRN:
  https://ssrn.com/abstract=3479189}} (\bibinfo{date}{Nov.}
  \bibinfo{year}{2019}).
\newblock


\bibitem[Ghaderi et~al\mbox{.}(2014)]%
        {GhaZhoSri_14}
\bibfield{author}{\bibinfo{person}{Javad Ghaderi}, \bibinfo{person}{Yuan
  Zhong}, {and} \bibinfo{person}{R. Srikant}.} \bibinfo{year}{2014}\natexlab{}.
\newblock \showarticletitle{Asymptotic Optimality of BestFit for Stochastic Bin
  Packing}.
\newblock \bibinfo{journal}{\emph{SIGMETRICS Perform. Eval. Rev.}}
  \bibinfo{volume}{42}, \bibinfo{number}{2} (\bibinfo{date}{Sept.}
  \bibinfo{year}{2014}), \bibinfo{pages}{64--66}.
\newblock


\bibitem[Gupta and Radovanovi\'{c}(2020)]%
        {GupRad_20}
\bibfield{author}{\bibinfo{person}{Varun Gupta} {and} \bibinfo{person}{Ana
  Radovanovi\'{c}}.} \bibinfo{year}{2020}\natexlab{}.
\newblock \showarticletitle{Interior-Point-Based Online Stochastic Bin
  Packing}.
\newblock \bibinfo{journal}{\emph{Oper. Res.}} \bibinfo{volume}{68},
  \bibinfo{number}{5} (\bibinfo{year}{2020}), \bibinfo{pages}{1474--1492}.
\newblock


\bibitem[Kleinrock(1975)]%
        {Kle_75}
\bibfield{author}{\bibinfo{person}{Leonard Kleinrock}.}
  \bibinfo{year}{1975}\natexlab{}.
\newblock \bibinfo{booktitle}{\emph{Queueing Systems}}.
\newblock \bibinfo{publisher}{John Wiley \& Son}.
\newblock


\bibitem[Li et~al\mbox{.}(2014)]%
        {LiTanCai_14}
\bibfield{author}{\bibinfo{person}{Yusen Li}, \bibinfo{person}{Xueyan Tang},
  {and} \bibinfo{person}{Wentong Cai}.} \bibinfo{year}{2014}\natexlab{}.
\newblock \showarticletitle{On Dynamic Bin Packing for Resource Allocation in
  the Cloud}. In \bibinfo{booktitle}{\emph{Proc. Ann. ACM Symp. Parallelism in
  Algorithms and Architectures (SPAA)}}. \bibinfo{address}{Prague, Czech
  Republic}, \bibinfo{pages}{2–11}.
\newblock


\bibitem[Lo et~al\mbox{.}(2015)]%
        {LoCheGov_15}
\bibfield{author}{\bibinfo{person}{David Lo}, \bibinfo{person}{Liqun Cheng},
  \bibinfo{person}{Rama Govindaraju}, \bibinfo{person}{Parthasarathy
  Ranganathan}, {and} \bibinfo{person}{Christos Kozyrakis}.}
  \bibinfo{year}{2015}\natexlab{}.
\newblock \showarticletitle{Heracles: Improving resource efficiency at scale}.
  In \bibinfo{booktitle}{\emph{Proc. ACM/IEEE Ann. Int. Symp. Computer
  Architecture (ISCA)}}. \bibinfo{address}{Portland, OR},
  \bibinfo{pages}{450--462}.
\newblock


\bibitem[Maguluri and Srikant(2013)]%
        {MagSri_13}
\bibfield{author}{\bibinfo{person}{Siva~Theja Maguluri} {and}
  \bibinfo{person}{R. Srikant}.} \bibinfo{year}{2013}\natexlab{}.
\newblock \showarticletitle{Scheduling jobs with unknown duration in clouds}.
  In \bibinfo{booktitle}{\emph{Proc. IEEE Int. Conf. Computer Communications
  (INFOCOM)}}. \bibinfo{address}{Turin, Italy}, \bibinfo{pages}{1887--1895}.
\newblock


\bibitem[Maguluri et~al\mbox{.}(2012)]%
        {MagSriYin_12}
\bibfield{author}{\bibinfo{person}{Siva~Theja Maguluri}, \bibinfo{person}{R
  Srikant}, {and} \bibinfo{person}{Lei Ying}.} \bibinfo{year}{2012}\natexlab{}.
\newblock \showarticletitle{Stochastic Models of Load Balancing and Scheduling
  in Cloud Computing Clusters}. In \bibinfo{booktitle}{\emph{Proc. IEEE Int.
  Conf. Computer Communications (INFOCOM)}}. \bibinfo{address}{Orlando, FL},
  \bibinfo{pages}{702--710}.
\newblock


\bibitem[Maguluri et~al\mbox{.}(2014)]%
        {MagSriYin_14}
\bibfield{author}{\bibinfo{person}{Siva~Theja Maguluri}, \bibinfo{person}{R.
  Srikant}, {and} \bibinfo{person}{Lei Ying}.} \bibinfo{year}{2014}\natexlab{}.
\newblock \showarticletitle{Heavy traffic optimal resource allocation
  algorithms for cloud computing clusters}.
\newblock \bibinfo{journal}{\emph{Perform. Eval.}}  \bibinfo{volume}{81}
  (\bibinfo{year}{2014}), \bibinfo{pages}{20--39}.
\newblock


\bibitem[Meyn(2007)]%
        {Mey_07}
\bibfield{author}{\bibinfo{person}{Sean Meyn}.}
  \bibinfo{year}{2007}\natexlab{}.
\newblock \bibinfo{booktitle}{\emph{Control Techniques for Complex Networks}
  (\bibinfo{edition}{1st} ed.)}.
\newblock \bibinfo{publisher}{Cambridge University Press},
  \bibinfo{address}{USA}.
\newblock


\bibitem[Psychas and Ghaderi(2018)]%
        {PsyGha_18}
\bibfield{author}{\bibinfo{person}{Konstantinos Psychas} {and}
  \bibinfo{person}{Javad Ghaderi}.} \bibinfo{year}{2018}\natexlab{}.
\newblock \showarticletitle{On Non-Preemptive {VM} Scheduling in the Cloud}. In
  \bibinfo{booktitle}{\emph{Proc. ACM SIGMETRICS Int. Conf. Measurement and
  Modeling of Computer Systems}}. \bibinfo{address}{Irvine, CA},
  \bibinfo{pages}{67–69}.
\newblock


\bibitem[Psychas and Ghaderi(2019)]%
        {PsyGha_19}
\bibfield{author}{\bibinfo{person}{Konstantinos Psychas} {and}
  \bibinfo{person}{Javad Ghaderi}.} \bibinfo{year}{2019}\natexlab{}.
\newblock \showarticletitle{Scheduling Jobs with Random Resource Requirements
  in Computing Clusters}. In \bibinfo{booktitle}{\emph{Proc. IEEE Int. Conf.
  Computer Communications (INFOCOM)}}. \bibinfo{pages}{2269--2277}.
\newblock


\bibitem[Psychas and Ghaderi(2021)]%
        {PsyGha_21}
\bibfield{author}{\bibinfo{person}{Konstantinos Psychas} {and}
  \bibinfo{person}{Javad Ghaderi}.} \bibinfo{year}{2021}\natexlab{}.
\newblock \showarticletitle{High-Throughput Bin Packing: Scheduling Jobs With
  Random Resource Demands in Clusters}.
\newblock \bibinfo{journal}{\emph{IEEE/ACM Trans. Netw.}} \bibinfo{volume}{29},
  \bibinfo{number}{1} (\bibinfo{year}{2021}), \bibinfo{pages}{220--233}.
\newblock


\bibitem[Psychas and Ghaderi(2022)]%
        {PsyGha_22}
\bibfield{author}{\bibinfo{person}{Konstantinos Psychas} {and}
  \bibinfo{person}{Javad Ghaderi}.} \bibinfo{year}{2022}\natexlab{}.
\newblock \showarticletitle{A Theory of Auto-Scaling for Resource Reservation
  in Cloud Services}.
\newblock \bibinfo{journal}{\emph{Stoch. Syst.}} \bibinfo{volume}{12},
  \bibinfo{number}{3} (\bibinfo{year}{2022}), \bibinfo{pages}{227--252}.
\newblock


\bibitem[Reiss et~al\mbox{.}(2012)]%
        {ReiTumGan_12}
\bibfield{author}{\bibinfo{person}{Charles Reiss}, \bibinfo{person}{Alexey
  Tumanov}, \bibinfo{person}{Gregory~R. Ganger}, \bibinfo{person}{Randy~H.
  Katz}, {and} \bibinfo{person}{Michael~A. Kozuch}.}
  \bibinfo{year}{2012}\natexlab{}.
\newblock \showarticletitle{Heterogeneity and Dynamicity of Clouds at Scale:
  Google Trace Analysis}. In \bibinfo{booktitle}{\emph{Proc. ACM Symp. Cloud
  Computing (SoCC)}}. \bibinfo{address}{San Jose, CA}, Article
  \bibinfo{articleno}{7}, \bibinfo{numpages}{13}~pages.
\newblock


\bibitem[Rzadca et~al\mbox{.}(2020)]%
        {RzaFinSwi_20}
\bibfield{author}{\bibinfo{person}{Krzysztof Rzadca}, \bibinfo{person}{Pawel
  Findeisen}, \bibinfo{person}{Jacek Swiderski}, \bibinfo{person}{Przemyslaw
  Zych}, \bibinfo{person}{Przemyslaw Broniek}, \bibinfo{person}{Jarek
  Kusmierek}, \bibinfo{person}{Pawel Nowak}, \bibinfo{person}{Beata Strack},
  \bibinfo{person}{Piotr Witusowski}, \bibinfo{person}{Steven Hand}, {and}
  \bibinfo{person}{John Wilkes}.} \bibinfo{year}{2020}\natexlab{}.
\newblock \showarticletitle{Autopilot: Workload Autoscaling at {Google}}. In
  \bibinfo{booktitle}{\emph{Proc. European Conf. Computer Systems (EuroSys)}}.
  \bibinfo{address}{Heraklion, Greece}, Article \bibinfo{articleno}{16},
  \bibinfo{numpages}{16}~pages.
\newblock


\bibitem[Stolyar(2013)]%
        {Sto_13}
\bibfield{author}{\bibinfo{person}{Alexander~L. Stolyar}.}
  \bibinfo{year}{2013}\natexlab{}.
\newblock \showarticletitle{An Infinite Server System with General Packing
  Constraints}.
\newblock \bibinfo{journal}{\emph{Oper. Res.}} \bibinfo{volume}{61},
  \bibinfo{number}{5} (\bibinfo{year}{2013}), \bibinfo{pages}{1200--1217}.
\newblock


\bibitem[Stolyar(2017)]%
        {Sto_17}
\bibfield{author}{\bibinfo{person}{Alexander~L. Stolyar}.}
  \bibinfo{year}{2017}\natexlab{}.
\newblock \showarticletitle{Large-scale heterogeneous service systems with
  general packing constraints}.
\newblock \bibinfo{journal}{\emph{Adv. Appl. Probab.}}  \bibinfo{volume}{49}
  (\bibinfo{date}{March} \bibinfo{year}{2017}), \bibinfo{pages}{61--83}.
\newblock
Issue 1.


\bibitem[Stolyar and Zhong(2013)]%
        {StoZho_13}
\bibfield{author}{\bibinfo{person}{Alexander~L. Stolyar} {and}
  \bibinfo{person}{Yuan Zhong}.} \bibinfo{year}{2013}\natexlab{}.
\newblock \showarticletitle{A large-scale service system with packing
  constraints: Minimizing the number of occupied servers}.
\newblock \bibinfo{journal}{\emph{ACM SIGMETRICS Perform. Eval. Rev.}}
  \bibinfo{volume}{41}, \bibinfo{number}{1} (\bibinfo{date}{June}
  \bibinfo{year}{2013}), \bibinfo{pages}{41--52}.
\newblock


\bibitem[Stolyar and Zhong(2015)]%
        {StoZho_15}
\bibfield{author}{\bibinfo{person}{Alexander~L. Stolyar} {and}
  \bibinfo{person}{Yuan Zhong}.} \bibinfo{year}{2015}\natexlab{}.
\newblock \showarticletitle{Asymptotic optimality of a greedy randomized
  algorithm in a large-scale service system with general packing constraints}.
\newblock \bibinfo{journal}{\emph{Queueing Syst.}}  \bibinfo{volume}{79}
  (\bibinfo{date}{June} \bibinfo{year}{2015}), \bibinfo{pages}{117--143}.
\newblock
Issue 2.


\bibitem[Stolyar and Zhong(2021)]%
        {StoZho_21}
\bibfield{author}{\bibinfo{person}{Alexander~L. Stolyar} {and}
  \bibinfo{person}{Yuan Zhong}.} \bibinfo{year}{2021}\natexlab{}.
\newblock \showarticletitle{A Service System with Packing Constraints: Greedy
  Randomized Algorithm Achieving Sublinear in Scale Optimality Gap}.
\newblock \bibinfo{journal}{\emph{Stoch. Syst.}}  \bibinfo{volume}{11}
  (\bibinfo{date}{June} \bibinfo{year}{2021}), \bibinfo{pages}{83--111}.
\newblock
Issue 2.


\bibitem[Tirmazi et~al\mbox{.}(2020)]%
        {TirBarDen_20}
\bibfield{author}{\bibinfo{person}{Muhammad Tirmazi}, \bibinfo{person}{Adam
  Barker}, \bibinfo{person}{Nan Deng}, \bibinfo{person}{Md~E. Haque},
  \bibinfo{person}{Zhijing~Gene Qin}, \bibinfo{person}{Steven Hand},
  \bibinfo{person}{Mor Harchol-Balter}, {and} \bibinfo{person}{John Wilkes}.}
  \bibinfo{year}{2020}\natexlab{}.
\newblock \showarticletitle{Borg: The next Generation}. In
  \bibinfo{booktitle}{\emph{Proc. European Conf. Computer Systems (EuroSys)}}.
  \bibinfo{address}{Heraklion, Greece}, Article \bibinfo{articleno}{30},
  \bibinfo{numpages}{14}~pages.
\newblock


\bibitem[Wilkes(2019)]%
        {Wil_19}
\bibfield{author}{\bibinfo{person}{John Wilkes}.}
  \bibinfo{year}{2019}\natexlab{}.
\newblock \bibinfo{title}{{Google} cluster-usage traces v3}.
\newblock \bibinfo{howpublished}{\url{http://github.com/google/cluster-data}}.
\newblock


\bibitem[Xie et~al\mbox{.}(2015)]%
        {XieDonLu_15}
\bibfield{author}{\bibinfo{person}{Qiaomin Xie}, \bibinfo{person}{Xiaobo Dong},
  \bibinfo{person}{Yi Lu}, {and} \bibinfo{person}{R. Srikant}.}
  \bibinfo{year}{2015}\natexlab{}.
\newblock \showarticletitle{Power of d Choices for Large-Scale Bin Packing: A
  Loss Model}. In \bibinfo{booktitle}{\emph{Proc. ACM SIGMETRICS Int. Conf.
  Measurement and Modeling of Computer Systems}}. \bibinfo{address}{Portland,
  OR}, \bibinfo{pages}{321--334}.
\newblock


\end{thebibliography}

\appendix
\section{Proof of Theorem~\ref{theo:sat:lower-bound} (Lower Bound)}\label{sec:sat:lower-bound}
\begin{proof} It is sufficient to show that given an infinite-server policy $\policy$ for $\isp((\arr_i r)_{i\in\jobspace}, \budget)$ in \eqref{eq:iss-opt-problem}, we have $\nact(\policy) \geq \sysbar{\nact}^*$. To this end, we will construct a single-server policy $\sysbar{\policy}$ such that the resulting system configuration $\sysbar{\veK}(\infty)$ in steady state satisfies:
    \begin{align}
        &\E\left[h\big(\sysbar{\veK}(\infty)\big)\middle| \sysbar{\veK}(\infty) \neq \vzero\right] = \costp(\policy) \leq \budget, \label{eq:lb-goal-budget} \\
        &\reqrate_i = \frac{\arr_i r}{\nact(\policy)}, \quad \forall i\in\jobspace. \label{eq:lb-goal-reqrate}
    \end{align} 
    Let $\pi$ be the distribution of $\sysbar{\veK}(\infty)$, then $(\nact(\policy), \sysbar{\policy}, \pi)$ is a feasible solution to the problem $\ssp((\arr_i r)_{i\in\jobspace}, \budget)$ in \eqref{eq:sat:ssss-opt-problem}. As a result, we have $\nact(\policy) \geq \sysbar{\nact}^*$. Note that although $\sysbar{\policy}$ is actually non-Markovian, i.e., it makes decisions based on not only the current configuration but also the history, as we will show in \Cref{sec:solve-single-lp}, $\sysbar{\nact}^*$ is still a lower bound to the objective value that $\sysbar{\policy}$ can achieve in $\ssp((\arr_i r)_{i\in\jobspace}, \budget)$.

    The construction of the single-server policy $\sysbar{\policy}$ involves simulating an infinite-server system under $\policy$ from the empty configuration. At time $0$, the policy $\sysbar{\policy}$ randomly chooses the $\ell$-th server in the infinite-server system with probability $p^\ell$, for $\ell = 1, 2, \cdots$. It then requests jobs for the single-server system according to a policy $\sysbar{\policy}^\ell$. The key to our policy $\sysbar{\policy}^\ell$ is to make the single-server system emulate the job assignment at the $\ell$-th server of the simulated infinite-server system, but without incurring idleness. We first construct the policy $\sysbar{\policy}^\ell$, and then specify the probabilities $p^\ell$. 
    
    Let us start by introducing some useful notation. Let $\sysbar{\veK}^\ell(t)$ be the single-server system configuration under $\sysbar{\policy}^\ell$ at time $t$ and $\veK^\ell(t)$ be the configuration of the $\ell$-th simulated server in the infinite-server system under $\policy$. We define a stochastic process $\{s^\ell(t),t\geq0\}$ as follows: $s^{\ell}(t)=\max_{\tau} \big\{\tau: \int_0^{\tau} \indibrac{\veK^\ell(x)\neq\vzero}dx=t \big\}$. The ``$\max$'' is well-defined {because the integral is continuous in $\tau$.} 
    Intuitively, $s^\ell(t)$ gives the maximum time when the accumulative busy time of the $\ell$-th server is $t$. 
    Note that $\{s^\ell(t),t\geq0\}$ is only discontinuous when $\veK^{\ell}(\tau)$ reaches $0$, thus it is right-differentiable with derivative equal to $1$ at any point. 
  
    We construct $\sysbar{\policy}^\ell$ and the  simulation of the infinite-server system under $\sigma$ in a way such that: 
    \begin{equation}\label{eq:lb-coupling}
        \sysbar{\veK}^\ell(t) = \veK^\ell(s^\ell(t)) \quad \forall t.
    \end{equation}
    That is, we want that the single-server system has the same dynamic of the simulated $\ell$-th server except skipping the idle period. To this end, we couple the two systems as follows: 
    \begin{enumerate}
        \item When the $\ell$-th simulated server $\veK^\ell(s^\ell(t))$ receives a type $i$ job, we let the single-server system $\sysbar{\veK}^\ell(t)$ request a type $i$ job at time $t$. For each such job, its phase transition process in the $\ell$-th simulated server is the same as that in the single-server system. That is, when we observe any internal transition or departure event in $\sysbar{\veK}^\ell(t)$, we produce a same event on the $\ell$-th simulated server $\veK^\ell(s^\ell(t))$.
        \item The simulations of the rest of the infinite-server system under policy $\policy$ are driven by independently generated random seeds. 
    \end{enumerate}
    It is not hard to see that the simulated infinite server-system has the same stochastic behavior as an uncoupled system under $\policy.$
    Moreover, as we couple all the events that happen in $\sysbar{\veK}^\ell(t)$ and $\veK^\ell(s^\ell(t))$, together with the facts that $\sysbar{\veK}^\ell(t)$ and $\veK^\ell(s^\ell(t))$ are piecewise constant and $\sysbar{\veK}^\ell(0-) = \veK^\ell(s^\ell(0-)) = \vzero$, we get \eqref{eq:lb-coupling}.

    Next we claim that \eqref{eq:lb-coupling} implies the following relationship between the steady-state cost of the single-server system under $\sysbar{\policy}^\ell$ and the steady-state cost of the $\ell$-th simulated server in the infinite-server system under $\policy$:
    \begin{equation}\label{eq:lb-budget-policy-ell}
        \E\left[h\left(\sysbar{\veK}^\ell(\infty)\right)\right] = \frac{ \E\left[h(\veK^\ell(\infty))\right]} {\Prob\left(\veK^\ell(\infty) \neq \vzero \right)}. 
    \end{equation}
    This is because 
    for all $\vek \neq \vzero$, we have
    \begin{align}
        \frac{\Prob\left(\veK^\ell(\infty) = \vek \right)}{\Prob\left(\veK^\ell(\infty) \neq \vzero \right)}  &\stackrel{(a)}{=}   \lim_{S\to\infty}\frac{\int_0^S \indibrac{\veK^\ell(s)=\vek}ds}{\int_0^S \indibrac{\veK^\ell(s)\neq\vzero}ds} \stackrel{(b)}{=} \lim_{T\to\infty}\frac{\int_0^T \indibrac{\veK^\ell(s^\ell(t))=\vek}dt}{\int_0^T \indibrac{\veK^\ell(s^\ell(t))\neq\vzero}dt} \nonumber\\
        &\stackrel{(c)}{=} \lim_{T\to\infty}\frac{1}{T} \int_0^T \indibrac{\veK^\ell(s^\ell(t))=\vek}dt \stackrel{(d)}{=} \lim_{T\to\infty}\frac{1}{T} \int_0^T \indibrac{\sysbar{\veK}^\ell(t)=\vek}dt\nonumber\\
        &\stackrel{(e)}{=} \Prob\left(\sysbar{\veK}^\ell(\infty) = \vek \right), \nonumber
    \end{align}
    where $(a)$ and $(e)$ hold because long-run averages converge to steady-state expectations; $(b)$ is due to the fact that $
    \int_0^T \indibrac{\veK^\ell(s^\ell(t))=\vek}dt = \int_0^{s^\ell(T)} \indibrac{\veK^\ell(s)=\vek}ds,
    $
    for any $\vek\neq \vzero$; $(c)$ is due to the fact that
   $\indibrac{\veK^\ell(s^\ell(t))\neq\vzero} = 1;$
    and $(d)$ follows from \eqref{eq:lb-coupling}. 
    
    Let $\reqrate_i^\ell$ be the long-run request rates of type $i$ jobs in the single-server system under $\sysbar{\policy}^\ell$, and  $\arr_i^\ell$ be the throughput of type $i$ jobs of $\ell$-th simulated server under $\policy$.   By the construction of $\sysbar{\policy}^\ell$, the single-server system requests jobs based on the arrival events of the $\ell$-th simulated server, we have
 $       \reqrate_i^\ell = \frac{\arr_i^\ell}{\Prob\left(\veK^\ell(\infty) \neq \vzero \right)}, \forall i\in\jobspace.$
    
    With the constructed policies $\{\sysbar{\policy}^\ell, \ell=1,2,\dots\}$, we are ready to define the policy $\sysbar{\policy}$. We let $\sysbar{\policy}$ choose an index $\ell$ with probability $p^\ell$ at time $0$, and then follow $\sysbar{\policy}^\ell$. We set the probability $p^\ell$ as 
    \begin{align} \label{eq:lb_prob_ell}
        p^\ell = \frac{\Prob\left(\veK^\ell(\infty) \neq \vzero \right)}{\sum_{\ell'=1}^\infty \Prob\left(\veK^{\ell'}(\infty) \neq \vzero \right)} = \frac{\Prob\left(\veK^\ell(\infty) \neq \vzero \right)}{\sum_{\vek \neq \vzero} \E[X_{\vek}(\infty)]}, \quad \forall \ell=1,2,\dots
   \end{align}
    where the second inequality uses the fact that $\sum_{\ell'=1}^\infty \Prob\left(\veK^{\ell'}(\infty) \neq \vzero \right) = \sum_{\vek \neq \vzero} \E[X_{\vek}(\infty)]$. 
Then under $\sysbar{\policy}$, we have
    \begin{equation*}
    \begin{aligned}
        \E\left[h\big(\sysbar{\veK}(\infty)\big)\right] &= \sum_{\ell=1}^{\infty} p^\ell \E\left[h\left(\sysbar{\veK}^\ell(\infty)\right)\right] \stackrel{(a)}{=} \sum_{\ell=1}^{\infty} \frac{\Prob\left(\veK^\ell(\infty) \neq \vzero \right)}{\sum_{\vek \neq \vzero} \E[X_{\vek}(\infty)]} \cdot \frac{ \E\left[h(\veK^\ell(\infty))\right]} {\Prob\left(\veK^\ell(\infty) \neq \vzero \right)} \\
        &= \frac{\sum_{\ell=1}^\infty \E\left[h(\veK^\ell(\infty))\right]}{\sum_{\vek \neq \vzero} \E[X_{\vek}(\infty)]} = \frac{ \sum_{\vek \neq \vzero} h(\vek) \E[X_{\vek}(\infty)]} {\sum_{\vek \neq \vzero} \E[X_{\vek}(\infty)]} = \costp(\policy),
    \end{aligned}
    \end{equation*}
    where $(a)$ follows from \eqref{eq:lb-budget-policy-ell} and \eqref{eq:lb_prob_ell}.
    Observe that under $\sysbar{\policy}$, $\sysbar{\veK}(\infty) \neq \vzero$ almost surely, we thus have
    \begin{equation*}
        \E\left[h\big(\sysbar{\veK}(\infty)\big)\middle| \sysbar{\veK}(\infty) \neq \vzero\right] = \E\left[h\big(\sysbar{\veK}(\infty)\big)\right] = \costp(\policy),
    \end{equation*}
    which proves \eqref{eq:lb-goal-budget}. Moreover, for each $i\in\jobspace$ the request rate $\reqrate_i$ is given by 
    \begin{align*}
        \reqrate_i &= \sum_{\ell=1}^\infty p^\ell \cdot \reqrate_i^\ell \nonumber 
        \stackrel{}{=}  \sum_{\ell=1}^\infty \frac{\Prob\left(\veK^\ell(\infty) \neq \vzero \right)}{\sum_{\vek \neq \vzero} \E[X_{\vek}(\infty)]} \cdot \frac{\arr_i^\ell}{\Prob\left(\veK^\ell(\infty) \neq \vzero \right)} \nonumber= \frac{\sum_{\ell=1}^\infty \arr_i^\ell}{\sum_{\vek \neq \vzero} \E[X_{\vek}(\infty)]} \stackrel{}{=} \frac{\arr_i r}{\nact(\policy)}.
    \end{align*}
    This proves \eqref{eq:lb-goal-reqrate}. By the argument presented at the beginning of the proof, we get $\nact(\policy) \geq \sysbar{\nact}^*$.
\end{proof}

\section{The Rest of the Proofs Needed for Theorem~\ref{theo:sat:conversion-general} (Conversion Theorem)}\label{appx:conversion}

\subsection{Proof of Lemma~\ref{lem:sat:pos-recur}} \label{appx:proof-sat-pos-recur}
\begin{proof}
We will show that under the \metapolicy\ policy, the Markov chain for the system state (represented as $((\veK^\ell(t))_{\ell=1, 2, \dots}, \vevj^\allL(t), \vetoK^\allL(t))$) has a unique stationary distribution by first arguing that it is $\vek^0$-irreducible (here being $\vek^0$-irreducible means the Markov chain has a state that can be reached by all other states through transitions, ``$\vek^0$'' in ``$\vek^0$-irreducible'' does not refer to any specific states), and then use Foster-Lyapunov theorem to show the positive recurrence \citep[see, e.g.,][]{Mey_07}. Combining $\vek^0$-irreducibility with positive recurrence, we can conclude that the Markov chain under study has a unique stationary distribution.

First, we show that the Markov chain $((\veK^\ell(t))_{\ell=1, 2, \dots}, \vevj^\allL(t), \vetoK^\allL(t))$ is $\vek^0$-irreducible. Specifically, observe that the Markov chain starting from any state $((\vek^\ell)_{\ell=1, 2, \dots}, \vevj^\allL, \vetok^\allL)$ can reach the state $(\vek^\allL, (\vzero)_{\ell=L+1, \dots}, \vzero^\allL, \vzero^\allL)$, after experiencing a sequence of departures and arrivals that clears up all the tokens, virtual jobs and jobs on backup servers. Further, letting $\widetilde{\vek}$ be the configuration reachable by all other configuration in the single-server system under the policy $\sysbar{\policy}$, we argue that starting from any states of the form  $(\vek^\allL, (\vzero)_{\ell=L+1, \dots}, \vzero^\allL, \vzero^\allL)$, the Markov chain can reach the state $(\widetilde{\vek}^\allL, (\vzero)_{\ell=L+1, \dots}, \vzero^\allL, \vzero^\allL)$. Because for any $\ell\leq L$, there is a transition path from $\vek^\ell$ to $\widetilde{\vek}$, consider the sequence of events where each $\veK^\ell(t)$ transitions independently following the path, and the jobs arrive right after $\veK^\ell(t)$ making a request, so that the tokens are checked out before $\veK^\ell(t)$ has a further transition. In this way, each $\veK^\ell(t)$ with $\ell\leq L$ can eventually reach $\widetilde{\vek}$ from $\vek^\ell$. This proves the $\vek^0$-irreducibility of $((\veK^\ell(t))_{\ell=1, 2, \dots}, \vevj^\allL(t), \vetoK^\allL(t))$.

Next, we show that $((\veK^\ell(t))_{\ell=1, 2, \dots}, \vevj^\allL(t), \vetoK^\allL(t))$ satisfies the Foster-Lyapunov criterion, i.e.,
\begin{equation}\label{eq:foster-lyapunov}
    \syshat{\gen} g \leq -1 + b \indibrac{S},
\end{equation}
where $g$ is a non-negative function of the states, $S$ is a finite set, $b$ is a finite number, and $\syshat{\gen}$ is the infinitesimal generator of the continuous-time Markov chain. Let $t_i$ be the expected remaining time in the system when a job is in phase $i$ for each $i\in\jobspace$. According to the job model, we have the recurrence relation
\begin{equation}\label{eq:pos-recur-proof:choose-coefficient}
    \Big(\mu_{i\perp} + \sum_{i'\in\jobspace\colon i'\neq i}  \mu_{ii'}\Big) t_i  = \sum_{i'\in\jobspace\colon i'\neq i}  \mu_{ii'}t_{i'}  \quad \forall i\in\jobspace. 
\end{equation}
We construct a Lyapunov function $g$ as follows:
\begin{equation}
    g((\vek^\ell)_{\ell=1, 2, \dots}, \vevj^\allL, \vetok^\allL) = \sum_{i\in\jobspace} \sum_{\ell=1}^\infty t_i k_i^\ell +  \sum_{i\in\jobspace} \sum_{\ell=1}^\infty t_i \vj_i^\ell. 
\end{equation}
Using the relation \eqref{eq:pos-recur-proof:choose-coefficient}, it can be verified that the drift of $g$ satisfies
\begin{align*}
    &\mspace{23mu}\syshat{\gen} g((\vek^\ell)_{\ell=1, 2, \dots}, \vevj^\allL, \vetok^\allL)\\
    &\leq \sum_{i\in\jobspace}  \Big(\arr_i t_i r - \sum_{\ell=1}^\infty  k_i^\ell\Big)  + \sum_{i\in\jobspace}  \Big(\sumL \sumrate{\vek', \vea}{\vek^\ell}\rate{\vek', \vea}{\vek^\ell} a_i  t_i - \sum_{\ell=1}^\infty \vj_i^\ell\Big) \\
    &\leq \sum_{i\in\jobspace} \Big(\arr_i t_i r + L\cdot \max_{\vek\in\Kinner} \sumrate{\vek', \vea}{\vek}\rate{\vek', \vea}{\vek} a_i  t_i\Big) -  \Big(\sum_{\ell=1}^\infty \sum_{i\in\jobspace} k_i^\ell + \sum_{\ell=1}^\infty \sum_{i\in\jobspace} \vj_i^\ell \Big),
\end{align*}
where the first inequality uses the fact that virtual jobs are generated at a rate no faster than the total rate of job requests. Then the Foster-Lyapunov criterion in \eqref{eq:foster-lyapunov} is satisfied with $b$ and $S$ given by
\[
b = \sum_{i\in\jobspace} \Big(\arr_i t_i r + L\cdot \max_{\vek\in\Kinner} \sumrate{\vek', \vea}{\vek}\rate{\vek', \vea}{\vek} a_i  t_i\Big),
\]
\[
    S = \left\{((\vek^\ell)_{\ell=1, 2, \dots}, \vevj^\allL, \vetok^\allL) \colon g((\vek^\ell)_{\ell=1, 2, \dots}, \vevj^\allL, \vetok^\allL) \leq b+1 \right\}.
\]
By the Foster-Lyapunov theorem, $((\veK^\ell(t))_{\ell=1, 2, \dots}, \vevj^\allL(t), \vetoK^\allL(t))$ is positive recurrent.
\end{proof}

\subsection{Proofs of Lemma~\ref{lem:sat:virtual-jobs} and Lemma~\ref{lem:sat:overflow-jobs}} \label{appx:proof-sat-virtual-overflow-jobs}

In this subsection, we prove Lemma~\ref{lem:sat:virtual-jobs} and Lemma~\ref{lem:sat:overflow-jobs} together. We begin by introducing some notations. We use $\vestT^\ell \triangleq (\syshat{\veK}^\ell, \vetoK^\ell)$ to represent the state of the $\ell$-th server, and use $\vestT^\allL$ to represent the joint state of the first $L$ servers. We also use the lowercase $\vestt^\ell$, $\vestt^\allL$ to represent the realizations of the corresponding random variables. We denote the total number of type $i$ virtual jobs as $\tvJ_i \triangleq \sumL \vj_i^\ell$ for $i\in\jobspace$, and its realization as $\tvj_i$. We denote the total number of type $i$ jobs on backup servers as $\toF_i$ for $\in\jobspace$, and its realizations as $\tof_i$. We also denote the total number of type $i$ tokens throughout the system as $\ttK_i \triangleq \sumL \toK_i^\ell$, and its realization as $\ttk_i$. Our goal can be rewritten as proving $\E[\tvJ_i ]=\Obrac{\sqrt{r}}$ and $\E[\toF_i ] = \Obrac{\sqrt{r}}$ for each $i\in\jobspace$.

We first give an overview of the proof. Observe that in our model, the expected time that a job stay in the system is fixed. As a result, bounding the number of virtual jobs or jobs on backup servers in the system is equivalent to bounding the rate that they are generated, according to Little's Law. By our construction of the policy, the rate of generating those jobs is closely related to the dynamics of the total number of type $i$ tokens $\ttK_i(t)$. 

To describe the dynamics of $\ttK_i(t)$, we first introduce two functions $\dtvjg_i$ and $\dtof_i$:
\begin{align*}
    \dtvjg_i(a_i, \ttk_i) &\triangleq (\ttk_i + a_i - \toklim)^+, \\
    \dtof_i(\ttk_i) &\triangleq (1-\ttk_i)^+.
\end{align*}
The function $\dtvjg_i$ represents the increment in the number of type $i$ virtual jobs due to the event that the total number of type $i$ tokens on the first $L$ servers exceeds the token limit $\toklim$. 
The function $\dtof_i$ corresponds to the increment in the total number of type $i$ jobs on backup servers due to the event that a type $i$ job arrives to the system without seeing a type $i$ token. For a function $g\colon (\Kouter \times \Kouter)^L \to \R$ that only depends on the number of type $i$ tokens $z_i$, its drift can be written as
\begin{align}\label{eq:sat:gen-ttk}
    \syshat{\gen} g( \vestt^\allL) &= \sumL \sumrate{\vek', \vea}{\vek^\ell} \rate{\vek', \vea}{\vek^\ell} \left(g(\ttk_i + a_i - \dtvjg_i(a_i, \ttk_i)) - g(\ttk_i)\right)\indibrac{\vetok^\ell=\vzero} \nonumber \\
    &\mspace{15mu} + \arr_i r \left(g(\ttk_i - 1 + \dtof_i(\ttk_i)) - g(\ttk_i)\right).
\end{align}
We abuse the notation of $g$ here. For ease of exposition, {we will simply write $\dtvjg_i$ and $\dtof_i$ to represent $\dtvjg_i(a_i, \ttk_i)$ and $\dtof_i(\ttk_i)$.}

By construction, the total number of type $i$ tokens $\{\ttK_i(t\})$ is a stochastic process constrained within $[0, \toklim]$. 
Note that $\ttK_i(t)$ increases when some servers request new tokens, and decreases when a real or virtual arrival checks out the token or when some servers have the excessive tokens removed. When $\ttK_i(t)$ is away from the boundaries, the average rate that it increases is approximately equal to $\arr_i r$, and the average rate that $\ttK_i(t)$ decreases is given by
\[
\E\left[\sumL \sumrate{\vek', \vea}{\syshat{\veK}^\ell} \rate{\vek', \vea}{\syshat{\veK}^\ell} a_i \indibrac{\vetoK^\ell=\vzero} \right] \approx \E\left[\sumL \sumrate{\vek', \vea}{\sysbar{\veK}^\ell} \rate{\vek', \vea}{\sysbar{\veK}^\ell} a_i \right] = L\cdot \reqrate_i \approx \arr_i r,
\]
where we have used the approximations that $\syshat{\veK}^\ell \overset{d}{\approx} \sysbar{\veK}^\ell$,   $\indibrac{\vetoK^\ell=\vzero} \approx 1$ and $L = \lceil\sysbar{\nact}\rceil \approx \sysbar{\nact}$.

As $\{\ttK_i(t)\}$ randomly moves up and down with approximately the same rate and reflects on the boundaries of $0$ and $\toklim$, it behaves as a reflected simple symmetric random walk. Intuitively speaking, the steady-state distribution of $\ttK_i$ is approximately a uniform distribution over $[0, \toklim]$. Recall that $\dtvjg_i$ and $\dtof_i$ can only be non-zero when $\ttK_i(t)$ is near the boundaries. Since the length of the interval $\toklim = \Thebrac{\sqrt{r}}$, 
we can expect that $\dtvjg_i$ and $\dtof_i$ diminish as $r\to\infty$. 

In the proof, we first establish the relationship between $\E[\tvJ_i]$, $\E[\toF_i]$ and $\dtvjg_i$, $\dtof_i$ using Little's Law. Then we derive bounds on $\dtvjg_i$ and $\dtof_i$ by analyzing the drift of several test functions of $\ttK_i$. This step is implicitly based on the intuition that $\ttK_i$ is approximately uniformly distributed over $[0, \toklim]$, with the tokens being generated and eliminated at similar rates. Finally, we invoke Lemma~\ref{lem:sat:w1-distance} to show that the tokens are indeed generated and eliminated at similar speeds, which leads to bounds on $\dtvjg_i$ and $\dtof_i$.

Finally, we make some additional remarks on the notations. First, $\dtvjg_i$ and $\dtof_i$ depend on the total number of type $i$ tokens $\ttk_i$ and the number of newly requested jobs $a_i$, although we omit the dependency expression for ease of exposition. Second, we abuse the notation $\dtvjg_i$ and $\dtof_i$ to denote the corresponding random variables. We also write $\sumka$ as a shorthand for $\sumrate{\vek', \vea}{\vek^\ell}$ when the context is clear.

\subsubsection{\textbf{Proofs of Lemma~\ref{lem:sat:virtual-jobs} and Lemma~\ref{lem:sat:overflow-jobs}.}}
We are now ready to prove Lemma~\ref{lem:sat:virtual-jobs} and Lemma~\ref{lem:sat:overflow-jobs}. 

\begin{proof}
\textbf{Step 1: Bounding Virtual Jobs and Jobs on Backup Servers using Little's Law.}
We first apply Little's Law to $\tvJ_i$ and $\toF_i$. For each type $i$, we let the expected time that a type $i$ job stays in the system be $t_i$. Let $\servmax = \max_{i\in\jobspace} t_i$. Because there are only finitely many types of jobs, and each job spends finite expected time in the system, $\servmax$ is a finite constant. By Little's law, 
\begin{align}
    \E[\tvJ_i] &\leq \servmax \E\left[\sumL \sumka \rate{\vek', \vea}{\syshat{\veK}^\ell} \dtvjg_i \indibrac{\vetoK^\ell=\vzero} \right], \label{eq:sat:tvj-from-dtvj}  \\
    \E[\toF_i] &\leq \servmax \E\left[\arr_i r \cdot \dtof_i \right]. \label{eq:sat:tof-from-dtof}
\end{align}

\textbf{Step 2: Drift Analysis.}
The above two equations \eqref{eq:sat:tvj-from-dtvj} and \eqref{eq:sat:tof-from-dtof} suggest that we can derive upper bounds on $\E[\tvJ_i]$ and $\E[\toF_i]$ by analyzing the following two terms:
\begin{itemize}
    \item $\E\left[\sumL \sumka \rate{\vek', \vea}{\syshat{\veK}^\ell} \dtvjg_i\indibrac{\vetoK^\ell=\vzero}\right]$, interpreted as the average rate that $\ttK_i(t)$ reflects on the boundary at $\toklim$;
    \item $\E\left[\arr_i r \cdot  \dtof_i\right]$, interpreted as the average rate that $\ttK_i(t)$ reflects on the boundary at $0$.
\end{itemize}
We establish the relationships of $\dtvjg_i$, $\dtof_i$ and $\ttK_i$ by analyzing the drift of two test functions $g$.

Letting $g(\ttk_i) = \ttk_i$ and taking steady-state expectation over its drift, by \eqref{eq:sat:gen-ttk} and the fact that the drift is zero in steady state, we get 
\begin{equation}\label{eq:sat:ttk-first-order}
    \E\left[\sumL \sumka \rate{\vek', \vea}{\syshat{\veK}^\ell} (a_i - \dtvjg_i)\indibrac{\vetoK^\ell=\vzero} + \arr_i r (-1 + \dtof_i) \right] =  0.
\end{equation}

Similarly, letting $g(\ttk_i) = \ttk_i^2$ and taking steady-state expectation over its drift, one can verify that 
\begin{align}
    &\mspace{23mu} \E\left[\sumL \sumka \rate{\vek', \vea}{\syshat{\veK}^\ell} \dtvjg_i \indibrac{\vetoK^\ell=\vzero}\right]  \label{eq:sat:ttk-second-order} \\
    &= \frac{1}{\toklim} \cdot \E\left[\left(\sumL \sumka \rate{\vek', \vea}{\syshat{\veK}^\ell} a_i  \indibrac{\vetoK^\ell=\vzero} - \arr_i r \right) \cdot \ttK_i\right] \label{eq:sat:ttk-second-order-term1} \\
    &\mspace{2mu}+ \frac{1}{2\toklim}  \cdot \E\left[\sumL \sumka \rate{\vek', \vea}{\syshat{\veK}^\ell} \left(a_i^2 - (\dtvjg_i)^2\right)\indibrac{\vetoK^\ell=\vzero} + \arr_i r \cdot (1-(\dtof_i)^2)\right]. \label{eq:sat:ttk-second-order-term2}
\end{align}
Readers may refer to the complete calculation at the end of this subsection.

\textbf{Step 3: Estimating the Terms Obtained from Drift Analysis.}
We will first focus on bounding $ \E\left[\sumL \sumka \rate{\vek', \vea}{\syshat{\veK}^\ell} \dtvjg_i\indibrac{\vetoK^\ell=\vzero}\right]$ analyzing the two terms in \eqref{eq:sat:ttk-second-order-term1} and \eqref{eq:sat:ttk-second-order-term2} separately. Then we invoke \eqref{eq:sat:ttk-first-order} to bound $ \E\left[\arr_i r \cdot  \dtof_i\right]$. 

The term in \eqref{eq:sat:ttk-second-order-term2} is easy to deal with. Observe that the number of jobs requested each time should be no more than the maximal number of jobs that a server can hold, i.e., $a_i \leq \Kmax$, so
\begin{align}
    \eqref{eq:sat:ttk-second-order-term2} &\leq  \frac{1}{2\toklim}  \cdot \E\left[\sumL \sumka \rate{\vek', \vea}{\syshat{\veK}^\ell} \Kmax^2 + \arr_i r \right] \nonumber\\
    &\leq  \frac{1}{2\toklim}\cdot \E\left[\sumL \ratetot{\max}  \Kmax^2 + \arr_i r\right] \nonumber\\
    &=\Obrac{\sqrt{r}},
\end{align}
where in the second inequality we have used the fact that the total rate is uniformly bounded by $\ratetot{\max}$, and the last step uses the facts that $L = \Obrac{r}$ and $\toklim = \Thebrac{\sqrt{r}}$. 

To bound the term in \eqref{eq:sat:ttk-second-order-term1}, first observe that $Z_i \leq \toklim$, which implies that 
\begin{equation*}
    \eqref{eq:sat:ttk-second-order-term1} \leq \E\left[\abs{\sumL \sumka \rate{\vek', \vea}{\syshat{\veK}^\ell} a_i \indibrac{\vetoK^\ell=\vzero}  - \arr_i r }\right].
\end{equation*}
The term on the RHS of the above equation is the expected absolute difference between the rates of generating and eliminating type $i$ tokens, which can be shown to be small relative to $r$. Specifically, we claim that
\begin{equation}\label{eq:sat:rate-abs-diff}
    \E\left[\abs{\sumL \sumka \rate{\vek', \vea}{\syshat{\veK}^\ell} a_i \indibrac{\vetoK^\ell=\vzero}  - \arr_i r }\right] = \Obrac{\sqrt{r}}.
\end{equation}

To show \eqref{eq:sat:rate-abs-diff}, first notice that we can remove the indicator $\indibrac{\vetoK^\ell=\vzero}$ without introducing much error: 
\begin{align*}
    &\mspace{23mu} \E\left[\abs{\sumL \sumka \rate{\vek', \vea}{\syshat{\veK}^\ell} a_i \indibrac{\vetoK^\ell=\vzero}  - \arr_i r }\right]\\
    &\leq \E\left[\abs{\sumL \sumka \rate{\vek', \vea}{\syshat{\veK}^\ell} a_i - \arr_i r }\right] + \E\left[\abs{\sumL \sumka \rate{\vek', \vea}{\syshat{\veK}^\ell} a_i \indibrac{\vetoK^\ell\neq\vzero}}\right]\\
    &\leq \E\left[\abs{\sumL \sumka \rate{\vek', \vea}{\syshat{\veK}^\ell} a_i - \arr_i r }\right] + \E\left[\abs{\sumL \ratetot{\max} \Kmax \indibrac{\vetoK^\ell\neq\vzero}}\right]\\
    &\leq \E\left[\abs{\sumL \sumka \rate{\vek', \vea}{\syshat{\veK}^\ell} a_i - \arr_i r }\right] + \ratetot{\max}\Kmax \abs{\jobspace} \toklim,
\end{align*}
where the first inequality is due to triangular inequality, the second inequality is due to the definition of $\ratetot{\max}$, and the last inequality is because $\sumL  \indibrac{\vetoK^\ell\neq\vzero} \leq \abs{\jobspace} \toklim$. It remains to bound the term $\E\left[\abs{\sumL \sumka \rate{\vek', \vea}{\syshat{\veK}^\ell} a_i - \arr_i r }\right]$, which can be seen as showing that the rate of generating type $i$ tokens concentrates around the type $i$ jobs' arrival rate $\arr_i r$. It is natural to think of using some Law of Large Numbers. Unfortunately, $\sumL \sumka \rate{\vek', \vea}{\syshat{\veK}^\ell} a_i$ is not a sum of i.i.d. random variables due to dependencies among $\syshat{\veK}^\ell$ for different $\ell$'s. As a result, we want to invoke the Wasserstein distance bound in Lemma~\ref{lem:sat:w1-distance} to replace $\syshat{\veK}^\ell$ in the above expression with $\sysbar{\veK}^\ell$. We define the function $f(\vek^\allL)$ as 
    \begin{equation}
        f(\vek^\allL) = \frac{1}{2\ratetot{\max}\Kmax} \abs{\sumL \sumka \rate{\vek', \vea}{\vek^\ell} a_i  - \arr_i r }. 
    \end{equation}
We claim that $f\in \text{Lip}(1)$. For any two $\vek^{\allL, 1}$, $\vek^{\allL, 2}$, 
\begin{align*}
    &\mspace{23mu} 2\ratetot{\max}\Kmax \cdot \left( f(\vek^{\allL, 1}) - f(\vek^{\allL, 2})\right)\\
    &= \abs{\sumL \sumka \rate{\vek', \vea}{\vek^{\ell, 1}} a_i  - \arr_i r } - \abs{\sumL \sumka \rate{\vek', \vea}{\vek^{\ell, 2}} a_i  - \arr_i r } \\
    &\leq \abs{\sumL \sumka \rate{\vek', \vea}{ \vek^{\ell, 1}} a_i  - \sumL \sumka \rate{\vek', \vea}{\vek^{\ell, 2}} a_i}\\
    &= \abs{\sumL \sumka \left(\rate{\vek', \vea}{ \vek^{\ell, 1}}  -  \rate{\vek', \vea}{\vek^{\ell, 2}} \right)a_i \indibrac{(\vek^{\ell, 1}) \neq (\vek^{\ell, 2}) }} \\
    &\leq \sumL \sumka \abs{\rate{\vek', \vea}{\vek^{\ell, 1}}  -  \rate{\vek', \vea}{\vek^{\ell, 2}}}\cdot \Kmax \cdot  \norm{\vek^{\ell, 1} - \vek^{\ell, 2}}\\
    &\leq \sumL 2\ratetot{\max} \Kmax \cdot \norm{\vek^{\ell, 1} - \vek^{\ell, 2}}\\
    &= 2\ratetot{\max}\Kmax \cdot \norm{ \vek^{\allL, 1} -  \vek^{\allL, 2}},
\end{align*}
where the first inequality is due to triangular inequality; the second inequality uses the fact that $a_i \leq \Kmax$ and $\indibrac{\vek^{1,\ell} \neq \vek^{2,\ell} } \leq \norm{\vek^{1,\ell} - \vek^{2,\ell}}$; the third inequality uses triangular inequality, the fact that the total rate at a configuration $\vek$ is bounded by $\ratetot{\max}$ and the property of the $L^1$ norm $\norm{\cdot}$. Therefore, $f\in \text{Lip}(1)$. The Lipschitz continuity of $f$ allows us to invoke Lemma~\ref{lem:sat:w1-distance} and get 
\[
    \E\left[\abs{\sumL \sumka \rate{\vek', \vea}{\syshat{\veK}^\ell} a_i  - \arr_i r }\right] - \E\left[\abs{\sumL \sumka \rate{\vek', \vea}{\sysbar{\veK}^\ell} a_i  - \arr_i r }\right] \leq 2\ratetot{\max}\Kmax \cdot \Obrac{\sqrt{r}}.
\]
Therefore,
\begin{equation}
    \E\left[\abs{\sumL \sumka \rate{\vek', \vea}{\syshat{\veK}^\ell} a_i  - \arr_i r }\right] \leq  \E\left[\abs{\sumL \sumka \rate{\vek', \vea}{\sysbar{\veK}^\ell} a_i  - \arr_i r }\right] + \Obrac{\sqrt{r}}.
\end{equation}
Observe that under a Markovian policy, the request rate of type $i$ jobs can be written as $\reqrate_i = \E[\sumka \rate{\vek', \vea}{\sysbar{\veK}^\ell} a_i]$, so we have
\begin{equation}
    \E\left[\sumL \sumka \rate{\vek', \vea}{\sysbar{\veK}^\ell} a_i  - \arr_i r \right] = \reqrate_i \cdot \lceil \sysbar{\nact}\rceil - \arr_i r = \Obrac{1}.
\end{equation}
Moreover, because $\sumka \rate{\vek', \vea}{\sysbar{\veK}^\ell} a_i$ are i.i.d. for $\ell=1,\dots, L$, we have
\begin{align*}
  \E\left[\abs{\sumL \sumka \rate{\vek', \vea}{\sysbar{\veK}^\ell} a_i  - \arr_i r }\right] &\leq   \sqrt{\E\left[\left(\sumL \sumka \rate{\vek', \vea}{\sysbar{\veK}^\ell} a_i  - \arr_i r \right)^2\right]} + \Obrac{1} \\
  &= \sqrt{\sumL \text{Var}\left(\sumka \rate{\vek', \vea}{\sysbar{\veK}^\ell} a_i\right)} + \Obrac{1}\\
  &= \Obrac{\sqrt{r}}.
\end{align*}
Therefore, by combining the arguments above, we get
\begin{align}
    \E\left[\abs{\sumL \sumka \rate{\vek', \vea}{\syshat{\veK}^\ell} a_i \indibrac{\vetoK^\ell=\vzero}  - \arr_i r }\right]
    &\leq \E\left[\abs{\sumL \sumka \rate{\vek', \vea}{\syshat{\veK}^\ell} a_i  - \arr_i r }\right] + \Obrac{\sqrt{r}} \nonumber \\
    &\leq \E\left[\abs{\sumL \sumka \rate{\vek', \vea}{\sysbar{\veK}^\ell} a_i  - \arr_i r }\right] + \Obrac{\sqrt{r}} \nonumber \\
    &\leq \Obrac{\sqrt{r}} ,  \label{eq:sat:step-2-pf:drift-term-bound-combine} 
\end{align}
which proves \eqref{eq:sat:rate-abs-diff}. This implies that the term in \eqref{eq:sat:ttk-second-order-term1} is also in $\Obrac{\sqrt{r}}$.

Combining the bounds on the terms in \eqref{eq:sat:ttk-second-order-term1} and \eqref{eq:sat:ttk-second-order-term2}, we get
\begin{equation*}
 \E\left[\sumL \sumka \rate{\vek', \vea}{\syshat{\veK}^\ell} \dtvjg_i \indibrac{\vetoK^\ell=\vzero} \right] = \Obrac{\sqrt{r}}.
\end{equation*}

Finally, we bound $\E\left[\arr_i r \cdot \dtof_i\right]$. We rearrange the terms in \eqref{eq:sat:ttk-first-order} and get 
\begin{align*}
\E\left[\arr_i r \cdot \dtof_i\right] &= \E\left[\sumL \sumka \rate{\vek', \vea}{\syshat{\veK}^\ell} (-a_i + \dtvjg_i)\indibrac{\vetoK^\ell=\vzero}  + \arr_i r \right] \\
& = \E\left[-\sumL \sumka \rate{\vek', \vea}{\syshat{\veK}^\ell} a_i \indibrac{\vetoK^\ell=\vzero} + \arr_i r \right] + \Obrac{\sqrt{r}}.
\end{align*}
By \eqref{eq:sat:rate-abs-diff}, we have $\E\left[-\sumL \sumka \rate{\vek', \vea}{\syshat{\veK}^\ell} a_i \indibrac{\vetoK^\ell=\vzero} + \arr_i r \right] = \Obrac{\sqrt{r}}$. Therefore,
\begin{equation*}
    \E\left[\arr_i r \cdot \dtof_i\right] = \Obrac{\sqrt{r}}.
\end{equation*}

We invoke the equations \eqref{eq:sat:tvj-from-dtvj} and \eqref{eq:sat:tof-from-dtof} that we get at the beginning of the proof, and conclude that
\begin{equation*}
    \E[\tvJ_i] \leq \servmax \E\left[\sumL \sumka \rate{\vek', \vea}{\syshat{\veK}^\ell} \dtvjg_i \indibrac{\vetoK^\ell=\vzero} \right] = \Obrac{\sqrt{r}}.
\end{equation*}
\begin{equation*}
    \E[\toF_i] \leq \servmax \E\left[\arr_i r \cdot \dtof_i \right] = \Obrac{\sqrt{r}}.
\end{equation*}
This finishes the proof.
\end{proof}

\subsubsection{\textbf{Deriving the equality in  \eqref{eq:sat:ttk-second-order}}}
We show the calculation detail of deriving the following equality.
\begin{align}
    &\mspace{23mu} \E\left[\sumL \sum_{\vek',\vea} \rate{\vek', \vea}{\syshat{\veK}^\ell} \dtvjg_i \indibrac{\vetoK^\ell=\vzero}\right] \tag{\ref{eq:sat:ttk-second-order}} \\
    &= \frac{1}{\toklim} \cdot \E\left[\left(\sumL \sum_{\vek',\vea} \rate{\vek', \vea}{\syshat{\veK}^\ell} a_i \indibrac{\vetoK^\ell=\vzero} - \arr_i r \right) \cdot \ttK_i\right] \tag{\ref{eq:sat:ttk-second-order-term1}} \\
    &\mspace{2mu}+ \frac{1}{2\toklim} \cdot \E\left[\sumL \sum_{\vek',\vea} \rate{\vek', \vea}{\syshat{\veK}^\ell} \left(a_i^2 - (\dtvjg_i)^2\right)\indibrac{\vetoK^\ell=\vzero} + \arr_i r \cdot (1-(\dtof_i)^2)\right]. \tag{\ref{eq:sat:ttk-second-order-term2}}
\end{align}
The equality is obtained by considering the drift of the function $g(\ttk_i) = \ttk_i^2$, which is zero in steady state. Recall that the drift of $g(\ttk_i)$ is given by
\begin{align}
    \syshat{\gen} g(\ttk_i, \vestt^\allL) &= \sumL \sum_{\vek',\vea} \rate{\vek', \vea}{\vek^\ell} \left(g(\ttk_i + a_i - \dtvjg_i) - g(\ttk_i)\right)\indibrac{\vetoK^\ell=\vzero} \nonumber \\
    &\mspace{15mu} + \arr_i r \left(g(\ttk_i - 1 + \dtof_i) - g(\ttk_i)\right). \tag{\ref{eq:sat:gen-ttk}}
\end{align}
We will first calculate $g(\ttk_i + a_i - \dtvjg_i) - g(\ttk_i)$, then $g(\ttk_i - 1 + \dtof_i) - g(\ttk_i)$, and finally plug them into the \eqref{eq:sat:gen-ttk}. 

The calculation of $g(\ttk_i + a_i - \dtvjg_i) - g(\ttk_i)$ utilizes the following property of $\dtvjg_i$:
\begin{equation}\label{eq:sat:dtvjg-reflection}
    (\ttk_i + a_i - \dtvjg_i) \cdot \dtvjg_i = \toklim \cdot \dtvjg_i.
\end{equation}
This property follows from the definition $\dtvjg_i = (\ttk_i + a_i - \toklim)^+$. Intuitively, this is because $\dtvjg_i$ is the ``force'' that pushes $\ttk_i$ back when it hits the boundary at $\toklim$. Using the property, we have 
\begin{align*}
    &\mspace{23mu} g(\ttk_i + a_i - \dtvjg_i) - g(\ttk_i)\\
    &= (\ttk_i + a_i - \dtvjg_i)^2 - \ttk_i^2\\
    &= (\ttk_i + a_i - \dtvjg_i)^2 - (\ttk_i + a_i - \dtvjg_i - a_i + \dtvjg_i)^2\\
    &= (\ttk_i + a_i - \dtvjg_i)^2 - \left((\ttk_i + a_i - \dtvjg_i)^2 + 2(-a_i + \dtvjg_i)\cdot  (\ttk_i + a_i - \dtvjg_i) + (-a_i + \dtvjg_i)^2\right)\\
    &= 2(a_i - \dtvjg_i)\cdot  (\ttk_i + a_i - \dtvjg_i) - (-a_i + \dtvjg_i)^2\\
    &= 2 a_i \cdot  (\ttk_i + a_i - \dtvjg_i) - (-a_i + \dtvjg_i)^2 -2 \dtvjg_i\cdot \toklim\\
    &= 2 a_i \cdot \ttk_i  + a_i^2 - (\dtvjg_i)^2 -2 \dtvjg_i\cdot \toklim.
\end{align*}
The second last equality is due to \eqref{eq:sat:dtvjg-reflection}, and the rest are all algebraic manipulations.

We carry out a similar calculation for $g(\ttk_i - 1 + \dtof_i) - g(\ttk_i)$:
\begin{align*}
    &\mspace{23mu} g(\ttk_i - 1 + \dtof_i) - g(\ttk_i) \\
    &= 2(-1+\dtof_i)\cdot \ttk_i + (-1+\dtof_i)^2\\
    &= -2\ttk_i + 2\ttk_i \cdot \dtof_i + 1 - 2\dtof_i + (\dtof_i)^2\\
    &= -2\ttk_i + 1 + 2(\ttk_i -1 + \dtof_i)\cdot \dtof_i - (\dtof_i)^2\\
    &= -2\ttk_i + 1 - (\dtof_i)^2.
\end{align*}
where the last equality is due to the property that
\begin{equation}\label{eq:sat:dtof-reflection}
    (\ttk_i -1 + \dtof_i) \cdot \dtof = 0,
\end{equation}
and the rest are all algebraic manipulations.

Putting together, 
\begin{align}
    \syshat{\gen} g(\ttk_i, \vestt^\allL) &= \sumL \sum_{\vek',\vea} \rate{\vek', \vea}{\vek^\ell} \left(g(\ttk_i + a_i - \dtvjg_i) - g(\ttk_i)\right)\indibrac{\vetok^\ell=\vzero} \nonumber \\
    &\mspace{15mu} + \arr_i r \left(g(\ttk_i - 1 + \dtof_i) - g(\ttk_i)\right)\nonumber \\
    &= \sumL \sum_{\vek',\vea} \rate{\vek', \vea}{\vek^\ell} \left(2 a_i \cdot \ttk_i  + a_i^2 - (\dtvjg_i)^2  -2 \dtvjg_i\cdot \toklim \right) \indibrac{\vetok^\ell=\vzero} \nonumber \\
    &\mspace{15mu} + \arr_i r \cdot\left( -2\ttk_i + 1 - (\dtof_i)^2\right) \nonumber.
\end{align}
After recombining the terms, we get
\begin{align}
    &\mspace{23mu} \toklim \cdot \E\left[\sumL \sum_{\vek',\vea} \rate{\vek', \vea}{\syshat{\veK}^\ell} \dtvjg_i \indibrac{\vetoK^\ell=\vzero}\right] \nonumber \\
    &= \E\left[\left(\sumL \sum_{\vek',\vea}\rate{\vek', \vea}{\syshat{\veK}^\ell} a_i\indibrac{\vetoK^\ell=\vzero}  - \arr_i r \right) \cdot \ttK_i\right] \nonumber \\
    &\mspace{23mu}+ \frac{1}{2} \E\left[\sumL \sum_{\vek',\vea} \rate{\vek', \vea}{\syshat{\veK}^\ell} \left(a_i^2 - (\dtvjg_i)^2\right)\indibrac{\vetoK^\ell=\vzero} + \arr_i r \cdot (1-(\dtof_i)^2)\right] \nonumber 
\end{align}

This finishes the calculation.

\subsection{Proof of Conversion Theorem without Assuming Irreducibility} \label{sec:proof-conv-general}

In this subsection, we prove Theorem \ref{theo:sat:conversion-general} without assuming $\vek^0$-irreducibility of the subroutine $\sysbar{\policy}$.
Specifically, suppose that we have a Markovian single-server policy $\sysbar{\policy}$ and an initial distribution $p_j$ over its recurrent classes $S_j$ for $j=1,2\dots J$, 
we will construct an infinite-server policy $\policy$ such that \eqref{eq:sat:active-servers} and \eqref{eq:sat:cost} still hold. 
The basic idea is to decompose a general Markovian single-server policy $\sysbar{\policy}$ into multiple $\vek^0$-irreducible Markovian policies, each induces one recurrent class and preserves stationary distribution and the throughput of $\sysbar{\policy}$ on that recurrent class, as stated in Lemma~\ref{lem:sat:break-reducible} below.

\begin{lemma}[Decomposing The Reducible Policy]\label{lem:sat:break-reducible}
    Let $\sysbar{\policy}$ be a general single-server Markovian policy with recurrent classes $S_j$ for $j=1,2,\dots, J$. 
    Then for each $j$ exists a Markovian policy $\sysbar{\policy}^j$ such that 
    \begin{itemize}
        \item The induced Markov chain is $\vek^0$-irreducible with the unique recurrent class being $S_j$;
        \item The stationary distribution is the same as the stationary distribution under $\sysbar{\policy}$ starting from a configuration in $S_j$.
    \end{itemize}
\end{lemma}

\begin{proof}
For each $j=1,2,\dots, J$, we define the policy $\sysbar{\policy}^j$ as follows: when the system has configuration $\vek \in S_j$, the policy $\sysbar{\policy}^j$ makes the same decisions as $\sysbar{\policy}$; when $\vek\notin S_j$ and $\vek=\vzero$, the policy starts a timer whose duration follows an exponential distribution with rate $1$ and immediately adds $\vek^0$ many jobs of each type when the timer ticks for some arbitrary $\vek^0\in S_j$; when $\vek\notin S_j$ and $\vek \neq \vzero$, the policy does not request any jobs.

We show that under the new policy $\sysbar{\policy}^j$, $S_j$ is also a recurrent class of the induced Markov chain. This is because if the system starts from a configuration $\vek\in S_j$, then it will stay in $S_j$ since it makes the same decisions and has the same transitions as under the policy $\sysbar{\policy}$. Because $S_j$ is a recurrent class under $\sysbar{\policy}$, it is still a recurrent class under $\sysbar{\policy}^j$.

To show that the Markov chain induced by $\sysbar{\policy}^j$ is $\vek^0$-irreducible, observe that starting from any $\vek \notin S_j$, the system state will return to $S_j$. Specifically,
\begin{itemize}
    \item If the system starts from a configuration $\vek$ such that $\vek \notin S_j$ and $\vek\neq\vzero$, then no new jobs will be requested until either $\vek\in S_j$ or $\vek=\vzero$. In the latter case, by the construction of the policy, the system jumps to a configuration in $S_j$ after the next transition.
    \item If the system starts from a configuration $\vek$ such that $\vek \notin S_j$ and $\vek=\vzero$, the system jumps to a configuration in $S_j$ after the next transition.
\end{itemize}

The claim that the stationary distribution under $\sysbar{\policy}^j$ is the same as the stationary distribution under $\sysbar{\policy}$ starting from a configuration in $S_j$ is trivial to show, because when the system initializes from any configuration in $S_j$, it stays in $S_j$ and the transitions are exactly the same under the two policies. 
\end{proof}

The JRS policy with a general Markovian single-server policy as its subroutine $\sysbar{\policy}$ is constructed using the $\vek^0$-irreducible policies $\sysbar{\policy}^j$'s obtained from the decomposition of $\sysbar{\policy}$ and the probabilities $p^j$'s. 
\begin{enumerate}[leftmargin=1.5em]
    \item We divide all the servers into $J$ server pools, each with infinitely many servers. Let the $j$-th server pool run the \metapolicy\ policy with subroutine $\sysbar{\policy}^j$ (defined in \Cref{sec:sat:JRRS-irreducible}) for each $j=1,\dots J$. 
    \item Whenever we see an arrival of type $i$, we route the job to the $j$-th infinite-server system with probability $\frac{p^j \reqrate_i^j }{\sum_{j} p^j \reqrate_{i'}^j}$ for each $j=1,\dots J$. 
\end{enumerate}

To analyze the policy $\policy$, let $\pi^j$ and $\reqrate_i^j$'s be the stationary distribution and throughput of the policy $\sysbar{\policy}^j$ for $j=1,\dots, J$. By Lemma~\ref{lem:sat:break-reducible}, we have the following relationships:
\begin{align}
    \pi(\vek) &= \textstyle\sum_{j=1}^J p^j\pi^j(\vek), \quad \forall \vek\in\Kinner, \label{eq:reducible-break:sta} \\
    \reqrate_i &= \textstyle\sum_{j=1}^J p^j\reqrate_i^j, \quad \forall i\in\jobspace. \label{eq:reducible-break:throughput}
\end{align}
Based on the above relationships, we can prove the general version of the Conversion Theorem that does not require $\vek^0$-irreducibility.

\begin{proof}[\textbf{Proof of Theorem~\ref{theo:sat:conversion-general}.}]
    For each $j=1,\dots, J$, Lemma~\ref{lem:sat:break-reducible} implies that the policy $\sysbar{\policy}^j$ and stationary distribution $\pi^j$ form a feasible solution to the single-server problem $\ssp\big(\big(p^j \reqrate_i^j \sysbar{\nact}\big)_{i\in\jobspace}, \budget\big)$, and the corresponding objective value is $p^j \sysbar{\nact}$. Consider the infinite-server system with arrival rates $\big(p^j \reqrate_i^j \sysbar{\nact}\big)_{i\in\jobspace}$ and budget $\budget$. As we have proved Theorem \ref{theo:sat:conversion-general} for the \metapolicy\ policy with $\vek^0$-irreducible subroutine $\sysbar{\policy}^j$, it follows that
    \begin{align}
        \left\lvert \textstyle\sum_{\vek\neq\vzero} \E[X_{\vek}^j] - \ceil{p^j \sysbar{\nact}} \cdot (1- \pi^j(\vzero)) \right\rvert &= \Obrac{\sqrt{r}}, \label{eq:sat:active-servers:j-th} \\
        \left\lvert \textstyle\sum_{\vek\neq\vzero} h(\vek) \E[X_{\vek}^j] - \ceil{p^j\sysbar{N}} \cdot \textstyle\sum_{\vek\neq\vzero} h(\vek) \pi^j(\vek)\right\rvert  &= \Obrac{\sqrt{r}}\label{eq:sat:cost:j-th},
    \end{align}
    where $X_{\vek}^j$ is the random variable representing the steady-state number of servers in configuration $\vek$ in the infinite-server system under the \metapolicy\ policy with subroutine $\sysbar{\policy}^j.$
    
    By the construction of $\policy$, the arrival rate to the $j$-th server pool is equal to 
    $ \frac{p^j \reqrate_i^j }{\sum_{j} p^j \reqrate_{i'}^j} \cdot \arr_i r = \frac{p^j \reqrate_i^j}{\reqrate_i} \cdot \arr_i r = p^j \reqrate_i^j \sysbar{\nact}$,  
    where the first equality is due to \eqref{eq:reducible-break:throughput}, and the second equality is due to the condition that $\arr_i r = \reqrate_i^j \cdot L$. Therefore, \eqref{eq:sat:active-servers:j-th} and \eqref{eq:sat:cost:j-th} hold, and we have
    \begin{align*}
       \abs{ \sum_{\vek\neq\vzero} \E[X_{\vek}] - \sysbar{\nact} \cdot (1- \pi(\vzero))}
       &= \abs{ \sum_{j=1}^J  \sum_{\vek\neq\vzero} \E[X_{\vek}^j] - \sum_{j=1}^J p^j\sysbar{\nact} \cdot (1-\pi^j(\vzero)) }\\
        &\leq \sum_{j=1}^J \abs{ \sum_{\vek\neq\vzero} \E[X_{\vek}^j] - \ceil{p^j\sysbar{\nact}} \cdot (1-\pi^j(\vzero)) } + O(1)\\
        &= O(\sqrt{r}).
    \end{align*}
    Here we use \eqref{eq:sat:active-servers:j-th} and the relationship between $\pi(\vek)$ and $\pi^j(\vek)$. Similarly, 
    \begin{align*}
       \abs{ \sum_{\vek\neq\vzero} h(\vek) \E[X_{\vek}] - \sysbar{\nact}\cdot \sum_{j=1}^J \sum_{\vek\neq\vzero} h(\vek) \pi(\vek)} 
       &= \abs{ \sum_{j=1}^J  \sum_{\vek\neq\vzero} \E[X_{\vek}^j] - \sum_{j=1}^J  p^j \sysbar{\nact}\cdot \sum_{\vek\neq\vzero} h(\vek) \pi^j(\vek) }\\
        &= \sum_{j=1}^J  \abs{  \sum_{\vek\neq\vzero} \E[X_{\vek}^j] - 
        \ceil{p^j \sysbar{\nact}}\cdot \sum_{\vek\neq\vzero} h(\vek) \pi^j(\vek) }+ O(1) \\
        &= O(\sqrt{r}).
    \end{align*}
    After re-indexing the servers, we get \eqref{eq:sat:active-servers} and  \eqref{eq:sat:cost}. The bounds on $\nact(\policy)$ and $\costp(\policy)$, \eqref{eq:alpha-optimal} and \eqref{eq:beta-optimal}, follow from \eqref{eq:sat:active-servers} and  \eqref{eq:sat:cost}. They can be verified in the same way as that in the proof for the irreducible case, so we omit the argument here. 
\end{proof}

\section{Solving the Single-Server Problem}\label{sec:solve-single-lp}
In this section, we show the equivalence of the single-server problem in \eqref{eq:sat:ssss-opt-problem} and a linear program (LP) as stated in \Cref{theo:single-opt}. The equivalence needs to be proved in two directions. In \Cref{sec:sat:lp-relaxation}, we first derive the linear program \eqref{eq:sat:ssss-lp} as a relaxation of the single-server problem \eqref{eq:sat:ssss-opt-problem} so that the optimal value of the LP is a lower bound to the optimal value of \eqref{eq:sat:ssss-opt-problem}. Then in \Cref{sec:sat:policy-class-lp-based}, we will construct a Markovian single-server policy that achieves the optimal value of the LP, which implies the optimality of the policy.

\subsection{Lower Bound via LP Relaxation}\label{sec:sat:lp-relaxation}

In this subsection we derive an LP relaxation of the optimization problem \eqref{eq:sat:ssss-opt-problem}, restated below: 
\begin{equation}\tag{\ref{eq:sat:ssss-opt-problem} revisited}
    \begin{aligned}
        & \underset{\sysbar{\nact}, \mspace{10mu}  \sysbar{\policy}, \mspace{10mu} \pi }{\text{minimize}}
        & & \sysbar{\nact} \\
        & \text{subject to}
        &  &\E\left[h\big(\sysbar{\veK}(\infty)\big)\middle| \sysbar{\veK}(\infty) \neq \vzero\right] \leq \budget,\\
        & &  & \sysbar{\nact} \cdot \reqrate_i = \arr_i r, \quad \forall i\in\jobspace.
    \end{aligned}
\end{equation}
Here we allow $\sysbar{\policy}$ to be non-Markovian, i.e., it can make decisions based on the history, but we still require its performance metrics used in the objective and the constraints of \eqref{eq:sat:ssss-opt-problem} to have well-defined steady-state distributions. 

Observe that both $\sysbar{\veK}(\infty)$ and $\reqrate_i$ depend on the stationary distribution $\pi$, but the constraints in terms of $\pi$ are implicit. To derive an LP relaxation, we give an explicit characterization of the constraints that must be satisfied by the stationary distribution $\pi$ induced by any feasible policy $\sysbar{\policy}$.

To do this, we derive a version of the stationary equation in terms of a quantity called \textit{transition frequency}. The transition frequency of type $i$ jobs is a function $\trfr_i \colon \Kinner \to \R$ describing the steady-state frequency of requesting a type $i$ job when the system has configuration $\ve{k}$. 
To rigorously define transition frequency, we first introduce a concept called \textit{nominal transition}.

\begin{definition}[Nominal Transition]
    Consider a single-server system under any policy. When the configuration $\sysbar{\veK}(t)$ transitions from $\vek$ to $\vek'+ \vea$ for some $\vek, \vek'\in\Kinner$ with $\vea = (a_i)_{i\in\jobspace}$ new jobs added into service, we decompose the transition by adding intermediate configurations as illustrated below, where $\vek$ first goes to $\vek'$ if $\vek'\neq\vek$, then add jobs of each type one by one. 
    \[
         \vek \to \vek' \to (\vek' + \vei_{i_1}) \to \dots \to (\vek' +a_{i_1} \vei_{i_1}) \to \dots \to (\vek'+a_{i_1} \vei_{i_1} + \dots + a_{i_{\abs{\jobspace}}} \vei_{i_{\abs{\jobspace}}}),
    \]
    where $(i_1, i_2, \dots, i_{\abs{\jobspace}})$ is a fixed ordering of the set of phases $\jobspace$. 
    We call each short transition in the diagram a \textit{nominal transition}.
\end{definition}

For $\vek^1, \vek^2\in\Kinner$, we denote $F(\vek^1, \vek^2, t)$ as the cumulative number of nominal transitions from $\vek^1$ to $\vek^2$ during the time interval $[0, t]$, which is a random variable with a distribution depending on the single-server policy and initial distribution of configurations.

Note that for any $i\in\jobspace$ and $\vek\in \Kinner$ s.t. $k_i \geq 1$, $F(\vek, \vek-\vei_i, t)$ counts the number of times that a type $i$ job departs when being in configuration $\vek$ . As a result,
\[
    F(\vek, \vek-\vei_i, t) = \mathcal{N}\left(\int_0^t k_i\mu_{i\perp} \indibrac{K(s)=\vek} ds\right),
\]
where $\mathcal{N}(t)$ denotes a unit rate Poisson process. If we take expectation, divide both sides by $t$, and let $t\to\infty$, we have
\begin{equation}\label{eq:departure-frequency}
     \lim_{t\to\infty} \frac{1}{t} \E[F(\vek, \vek-\vei_i, t)] = k_i \ser_{i\perp} \cdot \lim_{t\to\infty} \frac{1}{t} \int_0^t \Prob(\veK(s) = \vek) ds =  k_i \ser_{i\perp} \pi(\vek).
\end{equation}
Similarly, $F(\vek, \vek-\vei_i + \vei_{i'}, t)$ counts the number of times a job in phase $i$ transitions to phase $i'$ when being in configuration $ \vek$ for any $i,i'\in\jobspace$, $\vek\in\Kinner$ s.t. $i'\neq i$ and $k_i \geq 1$, so
\begin{equation}\label{eq:internal-trans-frequency}
    \lim_{t\to\infty} \frac{1}{t} \E[F(\vek, \vek-\vei_i+\vei_{i'}, t)] = k_i \ser_{ii'} \pi(\vek).
\end{equation}

We define transition frequency as follows.
 
\begin{definition}[Transition Frequency]\label{def:trans-freq}
    Transition frequency of type $i$ jobs at state $\vek$ is the long-run average number of nominal transitions from configuration $\vek$ to $\vek+\vei_i$ per unit time,
    \begin{equation}
        \trfr_i(\vek) \triangleq \lim_{t\to\infty} \frac{1}{t} \E[F(\vek, \vek+\vei_i, t)].
    \end{equation}
\end{definition}

The transition frequencies allow us to derive the following version of the stationary equation.

\begin{lemma}[Stationary Equation]\label{lem:sta-eq}
    Under any policy, the stationary distribution $\pi$ and the transition frequency $\trfr_i$ satisfy the following equation:
    \begin{equation}\label{eq:sta-eq}
    \begin{aligned}
        &\mspace{23mu} \sum_i \trfr_i(\vek - \vei_i)\indibrac{k_i\geq 1} + \sum_i (k_i+1)\ser_{i\perp} \pi(\vek + \vei_i) \\
        &\mspace{23mu}+ \sum_{i,i'\colon i\neq i'}(k_i+1) \ser_{ii'} \pi(\vek + \vei_i - \vei_{i'}) \indibrac{k_{i'} \geq 1} \\ 
        &= \sum_{i} \trfr_i(\vek) + \Big(\sum_i k_i\ser_{i\perp}+\sum_{i,i'\colon i\neq i'}k_i \ser_{ii'}\Big) \pi(\vek)
    \end{aligned}
    \end{equation}
    for any state $\vek \in \Kinner$, and $\sum_i, \sum_{i,i'}$ are shorthand for $\sum_{i\in\jobspace}, \sum_{i,i'\in\jobspace}$.
\end{lemma}

\begin{proof}
    For each configuration $\vek\in\Kinner$, if we look at the difference between the number of nominal transitions into configuration $\vek$ and that out of configuration $\vek$ by time $t$, we have the following equation, 
    \begin{equation}\label{eq:nominal-transitions-master}
        \begin{aligned}
        &\mspace{23mu} \indibrac{\sysbar{\veK}(t)=\vek} - \indibrac{\sysbar{\veK}(0)=\vek}\\
        &= \sum_i F(\vek-\vei_i, \vek, t)\indibrac{k_i \geq 1} +  \sum_i F(\vek+\vei_i, \vek, t) + \sum_{i,i'\colon i\neq i'} F(\vek+\vei_i-\vei_{i'}, \vek, t)\indibrac{k_{i'} \geq 1}\\
        &\mspace{23mu}-  \sum_i F(\vek, \vek+\vei_i, t) -  \sum_i F(\vek, \vek-\vei_i, t)\indibrac{k_i \geq 1} - \sum_{i,i'\colon i\neq i'} F(\vek, \vek-\vei_i+\vei_{i'}, t)\indibrac{k_i \geq 1}
        \end{aligned}
    \end{equation}
    By Definition~\ref{def:trans-freq} and \eqref{eq:departure-frequency}-\eqref{eq:internal-trans-frequency}, if we divide both sides of \eqref{eq:nominal-transitions-master} by $t$ and let $t \to \infty$, we get the stationary equation in \eqref{eq:sta-eq}.
\end{proof}

Since the stationary equation in \eqref{eq:sta-eq} is linear in $u_i(\vek)$ and $\pi(\vek)$, we can write it in matrix form:
\begin{equation}\label{eq:sta-eq-matrix-no-nonneg}
    A \ve{\pi} + \textstyle\sum_{i\in\jobspace} B_i \ve{u}_i = 0,
\end{equation}
where $\ve{\pi}$ and  $\ve{\trfr}_i$ are column vectors representing $\pi(\cdot)$, $\trfr_i(\cdot)$, and $A, B_i$ are matrices that make \eqref{eq:sta-eq-matrix-no-nonneg} equivalent to \eqref{eq:sta-eq}. Therefore, the following three conditions are necessary for any tuple $(\ve{\pi}, (\ve{\trfr}_i)_{i\in \jobspace})$ to be a possible pair of stationary distribution and transition frequencies for a Markovian policy.
\begin{equation}\label{eq:sta-eq-matrix}
    \begin{aligned}
      A \ve{\pi} + \textstyle\sum_{i\in\jobspace} B_i \ve{u}_i &= 0\\
    \textstyle\sum_{\vek} \pi(\vek) &= 1\\
   \ve{\pi}, \ve{\trfr}_i &\geq 0, \quad \forall i\in\jobspace
    \end{aligned}
\end{equation}
Based on the characterization of stationary $\ve{\pi}$ and $\ve{\trfr}_i$'s in \eqref{eq:sta-eq-matrix}, we can now formulate a linear program $\sslp((\arr_i)_{i\in\jobspace}, \budget)$. The linear program has decision variables $\Phi\in\R$, $\ve{\pi} \in \R^{\Kinner}$, and $\ve{\trfr}_i \in \R^{\Kinner}$ for $i\in\jobspace$, where $\Phi$ is a factor that scales the throughput of each type of jobs in the direction of $(\arr_i)_{i\in\jobspace}$. 
\begin{equation}\label{eq:sat:ssss-lp}
    \begin{aligned}
        & \underset{\substack{\Phi, \mspace{10mu}\ve{\pi},\mspace{8mu} (\ve{\trfr}_i)_{i\in\jobspace}} }{\text{maximize}}         & & \Phi \\
        & \text{subject to}
        &  & \ve{h}^T \ve{\pi} \leq \budget (1-\pi(\vzero))\\
        & &  &  \ve{1}_{o}^T \ve{\trfr}_i = \Phi \cdot \arr_i \quad \forall i\\
        & &  & A \ve{\pi} + \sum\nolimits_{i\in\jobspace} B_i \ve{u}_i = \vzero\\
        & &  &  \ve{1}^T \ve{\pi} = 1\\
        & &  &  \ve{\pi} \geq 0, \ve{\trfr}_i \geq 0 \quad \forall i\in\jobspace
    \end{aligned}
\end{equation}  
where $\ve{h}$ is the vector form of the cost rate function $h$; $\ve{1}_{o}^T$ is a $\abs{\Kinner}$-dimensional vector with one in all entries except those with $\sum_{i\in\jobspace} k_i = \Kmax$. In addition to the last three constraints on stationarity, the first constraint of $\sslp((\arr_i)_{i\in\jobspace}, \budget)$ is the resource contention constraint; the second constraint is because the transition frequency from $\vek$ to $\vek+\vei_i$ is equal to the throughput of type $i$ jobs. 

By the construction of $\sslp((\arr_i)_{i\in\jobspace}, \budget)$, any feasible solution of $\ssp((\arr_i r)_{i\in\jobspace}, \budget)$ can be converted to a feasible solution of $\sslp((\arr_i)_{i\in\jobspace}, \budget)$. Let $\sysbar{\nact}^*$ be the optimal value of \eqref{eq:sat:ssss-opt-problem} and $\Phi^*$ be the optimal value of \eqref{eq:sat:ssss-lp}. Then we have
\begin{equation}\label{eq:lp-relax-lb}
    \sysbar{\nact}^* \geq \frac{r}{\Phi^*}.
\end{equation}

\subsection{Policy Construction}\label{sec:sat:policy-class-lp-based}
In this subsection, we describe a procedure that allows us to construct a policy that achieves the lower bound given by the LP relaxation in \eqref{eq:sat:ssss-lp}.  Specifically, given a feasible solution $(\ve{\pi}, (\ve{\trfr}_i)_{i\in\jobspace})$ to \eqref{eq:sat:ssss-lp}, we define an \textit{LP-based policy} that requests jobs as follows:
   \begin{itemize}
    \item Case 1: When the system enters a configuration $\vek$ with $\pi(\vek) \neq 0$, for each $i\in\jobspace$, the policy starts a timer whose duration follows an exponential distribution with rate $\frac{\trfr_{i}(\vek)}{\pi(\vek)}$. The policy requests a type $i$ job when the $i$-th timer ticks. When the configuration changes, all timers are canceled. 
    \item Case 2: When the system enters a configuration $\vek$ with $\pi(\vek) = 0$ and $\sum_{i'} \trfr_{i'}(\vek) \neq 0$, the policy immediately requests a type $i$ job with probability $\frac{\trfr_{i}(\vek)}{\sum_{i'} \trfr_{i'}(\vek)}$. 
    \item Case 3: When the system enters a configuration $\vek$ with $\pi(\vek) = 0$ and $\sum_{i'} \trfr_{i'}(\vek) = 0$, the policy does not request any jobs. 
\end{itemize}

We denote the LP-based policy based on the solution $(\ve{\pi}, (\ve{\trfr}_i)_{i\in\jobspace})$ as $\sysbar{\policy}(\ve{\pi}, (\ve{\trfr}_i)_{i\in\jobspace})$. 

\begin{remark}\label{remark:two-views-of-single-server-policy}
Note that the definition of the LP-based policy here is stated from a view different from the view in \Cref{sec:sat:JRRS-irreducible}: 
here each request only adds one job to the server, and one request can happen immediately after another; while in \Cref{sec:sat:JRRS-irreducible} each request can add multiple jobs to the server, and there is only one request happening at a time. 
We refer to the view here as the \emph{impulsive view}, because multiple requests happening at the same time can be thought of as having an infinite request rate. In contrast, we call the view in \Cref{sec:sat:JRRS-irreducible} the \emph{non-impulsive view}. 

The LP-based policy can be alternatively described using the non-impulsive view if we see multiple requests happening at the same time as one request that adds multiple jobs to the server. 
More specifically, each reactive request of the LP-based policy is initiated by an internal transition or departure and consists of one or multiple requests of Case~2; 
each proactive request of the LP-based policy consists of one request in Case~1 and possibly several requests in Case~2. 
\end{remark}

The following lemma characterizes the steady-state behavior of a single-server system under the LP-based policy.

\begin{lemma}[Properties of LP-based Policies]\label{lem:sat:policy-recovers-lpsol}
    Consider a single-server system under the LP-based policy $\sysbar{\policy}(\ve{\pi}, (\ve{\trfr}_i)_{i\in\jobspace})$, where $(\ve{\pi}, (\ve{\trfr}_i)_{i\in\jobspace})$ is a feasible solution to \eqref{eq:sat:ssss-lp}. We have that $\ve{\pi}$ is a stationary distribution under policy $\sysbar{\policy}$, and $(\ve{\trfr}_i)_{i\in\jobspace}$ are the transition frequencies corresponding to $\sysbar{\policy}$.
\end{lemma}

The proof of Lemma~\ref{lem:sat:policy-recovers-lpsol} is based on \eqref{eq:nominal-transitions-master}, following the same argument as the proof of Lemma~\ref{lem:sta-eq}, as well as an induction argument.

\begin{proof}
    Let $(\ve{\pi}, (\ve{\trfr}_i)_{i\in\jobspace})$ be a feasible solution to the LP in \eqref{eq:sat:ssss-lp}. To show that $\ve{\pi}$ and $(\ve{\trfr}_i)_{i\in\jobspace}$ are also the actual stationary distribution and transition frequencies, it suffices to show that if the initial distribution follows $\ve{\pi}$, i.e.
    \begin{equation*}
        \Prob\big(\sysbar{\veK}(0) = \vek\big) = \pi(\vek) 
    \end{equation*}
    then under the policy $\sysbar{\policy}(\ve{\pi}, (\ve{\trfr}_i)_{i\in\jobspace})$, we have
    \begin{align}
        &\lim_{T\to\infty} \frac{1}{T} \int_0^T \Prob\big(\sysbar{\veK}(t) = \vek\big)dt = \pi(\vek), \\
        &\lim_{T\to\infty} \frac{1}{T} \E[F(\vek, \vek+\vei_i, T)] = \trfr_{i}(\vek),
    \end{align}
    where $F$ is the cumulative number of nominal transitions under the policy $\sysbar{\policy}(\ve{\pi}, (\ve{\trfr}_i)_{i\in\jobspace})$ and the initial distribution $\ve{\pi}$. 
    
    Our proof is based on the following equation:
    \begin{equation}\label{eq:opt-recover:transition-rule}
        \begin{aligned}
        &\mspace{23mu}\frac{d}{dt} \Prob(\sysbar{\veK}(t) = \vek) \Big|_{t=0}\\
        &= \sum_i  \frac{d}{dt}\E\left[F(\vek-\vei_i, \vek, t)\right] \Big|_{t=0} -  \sum_i \frac{d}{dt} \E\left[F(\vek, \vek+\vei_i, t)\right]\Big|_{t=0} \\        &\mspace{23mu} +  \sum_i (k_i+1)\ser_{i\perp} \pi(\vek + \vei_i)
         + \sum_{i,i'}(k_i+1) \ser_{ii'} \pi(\vek + \vei_i - \vei_{i'}) \indibrac{k_{i'} \geq 1}\\
        &\mspace{23mu} - \sum_i k_i\ser_{i\perp} \pi(\vek) - \sum_{i,i'}k_i \ser_{ii'} \pi(\vek),
        \end{aligned}
    \end{equation}
   The equation is a straightforward consequence of \eqref{eq:nominal-transitions-master}, following the same argument as the proof of Lemma~\ref{lem:sta-eq}.

    We prove the following two equations by induction on $\sum_{i \in \jobspace} k_i$. 
    \begin{align}
        \lim_{t\to 0} \frac{1}{t} \E[F(\vek, \vek+\vei_i, t)] &= \trfr_i(\vek), \label{eq:opt-recover:rate-invariant} \\
        \frac{d}{dt} \Prob(\sysbar{\veK}(t) = \vek) |_{t=0} &= 0. \label{eq:opt-recover:distr-invariant}
    \end{align}
    
    We first consider the base case when $\sum_{i \in \jobspace} k_i=0$. In this case, $\vek = \vzero$ and we have
    \begin{equation}\label{eq:opt-recover:base-case-condition}
    \frac{d}{dt}\E\left[F(\vek-\vei_i, \vek, t)\right] \Big|_{t=0} = \trfr_i(\vek-\vei_i) \indibrac{k_i\geq 1}=0,
    \end{equation}
    for all $i$. This reduces \eqref{eq:opt-recover:transition-rule} to 
    \begin{equation}\label{eq:opt-recover:transition-rule:with-induction}
        \begin{aligned}
        &\mspace{23mu}\frac{d}{dt} \Prob(\sysbar{\veK}(t) = \vek) \Big|_{t=0}\\ 
        &= \sum_i \trfr(\vek-\vei_i)\indibrac{k_i\geq 1} -  \sum_i \frac{d}{dt} \E\left[F(\vek, \vek+\vei_i, t)\right]\Big|_{t=0} \\        &\mspace{23mu} +  \sum_i (k_i+1)\ser_{i\perp} \pi(\vek + \vei_i)
         + \sum_{i,i'}(k_i+1) \ser_{i,i'} \pi(\vek + \vei_i - \vei_{i'}) \indibrac{k_{i'} \geq 1}\\
        &\mspace{23mu} - \sum_i k_i\ser_{i\perp} \pi(\vek) - \sum_{i,i'}k_i \ser_{i,i'} \pi(\vek),
        \end{aligned}
    \end{equation}
    Now we discuss based on whether $\pi(\vek) = 0$.
    If $\pi(\vek) \neq 0$, by the definition of our policy, for all $i$,
    \begin{equation}
        \frac{d}{dt} \E\left[F(\vek, \vek+\vei_i, t)\right]\Big|_{t=0} = \frac{\trfr_i(\vek)}{\pi(\vek)}\cdot \Prob(\sysbar{\veK}(0) = \vek) = \trfr_i(\vek),
    \end{equation}
    which is \eqref{eq:opt-recover:rate-invariant}. Combining the above equation and the stationary equation \eqref{eq:sta-eq} satisfied by $(\ve{\pi}, (\ve{\trfr}_i)_{i\in\jobspace})$, we conclude that the RHS of \eqref{eq:opt-recover:transition-rule:with-induction} is zero, i.e.
    \[
        \frac{d}{dt} \Prob(\sysbar{\veK}(t) = \vek) |_{t=0} = 0,
    \]
    which is \eqref{eq:opt-recover:distr-invariant}. For the case when $\pi(\vek) = 0$ and $\sum_i \trfr_i(\vek) \neq 0$, because the system immediately leave the configuration $\vek$ after reaching it through a nominal transition,
    \[
        \frac{d}{dt} \Prob(\sysbar{\veK}(t) = \vek) |_{t=0} = 0,
    \]
    i.e., the LHS of  \eqref{eq:opt-recover:transition-rule:with-induction} is $0$. Again we compare \eqref{eq:opt-recover:transition-rule:with-induction} against the stationary equation \eqref{eq:sta-eq} and get
    \[
        \sum_{i} \frac{d}{dt} \E\left[F(\vek, \vek+\vei_{i}, t)\right]\Big|_{t=0} = \sum_{i} \trfr_i(\vek).
    \]
    By the definition of our policy, we have
    \begin{equation}
    \begin{aligned}
        \frac{d}{dt} \E\left[F(\vek, \vek+\vei_i, t)\right]\Big|_{t=0} &= \frac{\trfr_i(\vek)}{\sum_{i'} \trfr_{i'}(\vek)}\cdot \sum_{i'} \frac{d}{dt} \E\left[F(\vek, \vek+\vei_{i'}, t)\right]\Big|_{t=0} = \trfr_i(\vek).
    \end{aligned}
    \end{equation}
    which is \eqref{eq:opt-recover:rate-invariant}. For the case when $\pi(\vek) = 0$ and $\sum_i \trfr_i(\vek) = 0$, \eqref{eq:sta-eq} implies that $\trfr_i(\vek-\vei_i) = 0$, $\pi(\vek+\vei_i) = 0$, $\pi(\vek+\vei_i - \vei_{i'}) = 0$ for any $i$. Then \eqref{eq:opt-recover:transition-rule:with-induction} is further reduced to 
    \[
        \frac{d}{dt} \Prob(\sysbar{\veK}(t) = \vek) \Big|_{t=0} =  -  \sum_i \frac{d}{dt} \E\left[F(\vek, \vek+\vei_i, t)\right]\Big|_{t=0}.
    \]
    Because $\Prob(\sysbar{\veK}(t) = \vek) \geq 0$ and $\Prob(\sysbar{\veK}(0) = \vek) = 0$, the LHS of the above expression is non-negative. However, the RHS of the above expression is non-positive.  Therefore, both sides are equal to zero, thus we have $\frac{d}{dt} \Prob(\sysbar{\veK}(t) = \vek) \Big|_{t=0} = 0$ and $\frac{d}{dt} \E\left[F(\vek, \vek+\vei_i, t)\right]\Big|_{t=0} = \trfr_i(\vek)$.
    
    Having proved the base case, we do the induction step. Suppose we have proved \eqref{eq:opt-recover:rate-invariant} and  \eqref{eq:opt-recover:distr-invariant} for all $\vek$ such that $\sum_{i \in \jobspace} k_i \leq m-1$ for some integer $m \geq 1$. We consider $\vek$ with $\sum_{i \in \jobspace} k_i = m$. By the induction hypothesis, 
    \begin{equation}
        \frac{d}{dt}\E\left[F(\vek-\vei_i, \vek, t)\right] \Big|_{t=0} = \trfr_i(\vek-\vei_i)\indibrac{k_i\geq 1}.
    \end{equation}
    Then we repeat the arguments after \eqref{eq:opt-recover:base-case-condition} of the base case verbatim. By induction, we have proved  \eqref{eq:opt-recover:rate-invariant} and \eqref{eq:opt-recover:distr-invariant}.
    
    Therefore, given the policy and initial distribution, the distribution of $\sysbar{\veK}(t)$ is stationary, i.e., we always have $\Prob(\sysbar{\veK}(t) = \vek) = \pi(\vek)$ for all $\vek \in \Kinner$. As a result, an analogue of \eqref{eq:opt-recover:rate-invariant} holds for all $t\geq 0$: $\E[F(\vek, \vek+\vei_i, t)]$ is differentiable with respect to $t$ and
    \begin{align*}
        \frac{d}{dt} \E[F(\vek, \vek+\vei_i, t)] &= \trfr_i(\vek),
    \end{align*}
    for all $\vek\in \Kinner$ and all $i\in\jobspace$. Therefore,
    \begin{align*}
        \lim_{T\to\infty} \frac{1}{T} \int_0^T \Prob(\sysbar{\veK}(t) = \vek)dt &= \lim_{T\to\infty} \frac{1}{T} \cdot \pi(\vek) \cdot T = \pi(\vek),\\
        \lim_{T\to\infty} \frac{1}{T} \E[F(\vek, \vek+\vei_i, T)] &= \lim_{T\to\infty} \frac{1}{T} \int_0^T \frac{d}{dt} \E[F(\vek, \vek+\vei_i, t)]dt \\
        &= \lim_{T\to\infty} \frac{1}{T} \cdot \trfr_i(\vek) \cdot T = \trfr_i(\vek).
    \end{align*}
    This completes the proof.
\end{proof}

\subsection{Proof of Theorem~\ref{theo:single-opt}} \label{appx:proof-single-opt}
\begin{customthm}{\ref{theo:single-opt}}[Optimality of Single-OPT]
    Given an optimal solution $(\Phi^*, \ve{\pi}^*, (\ve{\trfr}_i)_{i\in\jobspace})$ to the linear program $\sslp((\arr_i)_{i\in\jobspace}, \budget)$, we can solve the single-server problem $\ssp((\arr_i r)_{i\in\jobspace}, \budget)$ in \eqref{eq:sat:ssss-opt-problem} to achieve an optimal value $r/\Phi^*$, using the optimal policy $\sysbar{\policy}(\ve{\pi}^*, (\ve{\trfr}_i)_{i\in\jobspace}^*)$ and the optimal stationary distribution $\ve{\pi}^*$. Moreover, the optimal policy $\sysbar{\policy}(\ve{\pi}^*, (\ve{\trfr}_i)_{i\in\jobspace}^*)$ is a Markovian policy.
\end{customthm}

\begin{proof}
    By Lemma~\ref{lem:sat:policy-recovers-lpsol}, under the policy $\sysbar{\policy}(\ve{\pi}^*, (\ve{\trfr}^*_i)_{i\in\jobspace})$, $\ve{\pi}^*$ is a stationary distribution, and $(\ve{\trfr}^*_i)_{i\in\jobspace})$ are the corresponding transition frequencies. Recall the single-server $\ssp((\arr_i r  )_{i\in\jobspace}, \budget)$ problem
    \begin{equation}\tag{\ref{eq:sat:ssss-opt-problem} revisited}
    \begin{aligned}
        & \underset{\sysbar{\nact}, \mspace{10mu}  \sysbar{\policy}, \mspace{10mu} \pi }{\text{minimize}}
        & & \sysbar{\nact} \\
        & \text{subject to}
        &  &\E\left[h\big(\sysbar{\veK}(\infty)\big)\middle| \sysbar{\veK}(\infty) \neq \vzero\right] \leq \budget,\\
        & &  & \sysbar{\nact} \cdot \reqrate_i = \arr_i r, \quad \forall i\in\jobspace.
    \end{aligned}
    \end{equation}
    Observe that under the policy $\sysbar{\policy}(\ve{\pi}^*, (\ve{\trfr}^*_i)_{i\in\jobspace})$, the cost rate of resource contention is $\ve{h}^T \ve{\pi} \leq \budget (1-\pi(\vzero))$, the request rate of type $i$ jobs is $\reqrate_i = \ve{1}_{o}^T \ve{\trfr}^*_i = \Phi^* \cdot \arr_i$, so 
    \begin{align*}
        &\E\left[h\big(\sysbar{\veK}(\infty)\big)\middle| \sysbar{\veK}(\infty) \neq \vzero\right] = \frac{\ve{h}^T \ve{\pi}}{1-\pi(\vzero)} \leq \budget,\\
        &\reqrate_i = \Phi^* \cdot \arr_i, \quad \forall i\in\jobspace.
    \end{align*}
    where we have used the fact that $h(\vzero) = 0$ in the first equality. Therefore, $(\Phi^* / r, \sysbar{\policy}(\ve{\pi}^*, (\ve{\trfr}^*_i)_{i\in\jobspace}), \ve{\pi})$ is a feasible solution to the single-server problem  $\ssp((\arr_i r)_{i\in\jobspace}, \budget)$, achieving the objective value of $r/\Phi^*$, which is the optimal value because $r/\Phi^* \leq \sysbar{\nact}^*$. 
\end{proof}

\section{Performance guarantee of \metapolicyfull\ with an estimated model}\label{sec:proof-prop-estimation}

\subsection{Assumptions and result}
In this section, we consider the performance of \metapolicy\ when it is based on an estimated model. We will state the performance guarantee in terms of the estimation error, and give a proof sketch by pointing out which part of \Cref{theo:sat:conversion-general}'s proof needs to be changed accordingly. 

Consider the setting where the maximal jobs on a server $\Kmax$, the set of job phases $\jobspace$, the cost rate function $h(\cdot)$, and the budget $\budget$ are all known. 
However, we only have estimations of the jobs' arrival rates, internal transition rates, and departure rates.

Specifically, for any $i,i'\in\jobspace$, let $\mutrue_{ii'}$ be the \emph{true} rate of internal transition from phase $i$ to phase $i'$; let  $\mutrue_{i\perp}$ be the \emph{true} departure rate of phase $i$; let $\lamtrue_i r$ be the \emph{true} arrival rate of type $i$ jobs. We let $\muest_{ii'}$'s, $\muest_{i\perp}$'s, and $\lamest_i r$'s be the \emph{estimated} internal transition rates, departure rates, and arrival rates, respectively.
We assume that there exists a small positive constant $\delta$ that is independent of the scaling factor $r$ such that the following assumptions hold. 

\begin{assumption}[$\delta$-accurate estimation]\label{assump:delta-accurate}
	For any $i,i'\in\jobspace$, 
	\begin{align}
		\abs{\muest_{ii'} - \mutrue_{ii'}} &\leq \delta, \\
		\abs{\muest_{i\perp} - \mutrue_{i\perp}} &\leq \delta, \\
		\big\lvert\lamest_i - \lamtrue_i\big\rvert &\leq \delta. 
	\end{align}
\end{assumption}

\begin{assumption}[Scaling of the single-server objective value]\label{assump:Nbar-scaling}
	Consider the single-server problem in \eqref{eq:sat:ssss-opt-problem}. Let $(\sysbar{\nact},\sysbar{\sigma},\pi)$ be a solution feasible to the single-server problem with the estimated parameters that are $\delta$-accurate, where $\delta \in [0,\delta_{\max}),$ for some constant $\delta_{\max}>0$. We assume that there exist constants $0 < m_1 < m_2$ independent of $\delta$ and $r$ such that
    \[
        m_1 r \leq \sysbar{\nact}\leq m_2 r.
    \]
\end{assumption}

\begin{assumption}[$\delta$-insensitivity of the optimal value]\label{assump:delta-insensitive}
	Consider the single-server problem in \eqref{eq:sat:ssss-opt-problem}. 
    Let the optimal value of the single-server problem with the estimated parameters be $\sysbar{\nact}^*$, where the estimated parameters are $\delta$-accurate; let the optimal value of the single-server problem with true parameters be $\sysbar{\nact}^*_{\text{true}}$. We assume that  
	\[
		\sysbar{\nact}^* \leq \sysbar{\nact}^*_{\text{true}} + \delta r.
	\]
\end{assumption}

We also assume that \metapolicy\ can accurately simulate the virtual jobs. 

\begin{assumption}[Accurate simulation]\label{assump:accurate-simulation}
	The virtual jobs simulated in \metapolicy\ follow the true transition dynamics. 
\end{assumption}

Given the above assumptions, we have the following proposition that states the optimality gap of JRS policy with estimated model parameters, which has a similar form as \Cref{theo:sat:conversion-general} for JRS under true model parameters. 

\begin{proposition}[Optimality gap with model estimation]\label{theo:misspec}
Consider a stochastic bin-packing problem in service systems with time-varying job resource requirements. 
Let the infinite-server policy $\policy$ be \metapolicy\ with subroutine $\sysbar{\policy}$. Suppose $\policy$ is specified based on an estimated model satisfying Assumption~\ref{assump:delta-accurate}, \ref{assump:Nbar-scaling}, \ref{assump:delta-insensitive}, and \ref{assump:accurate-simulation}, for $\delta$ s.t. $\delta \in [0, \delta_{\max} )$, where $\delta_{\max}$ is some positive constant independent of $r$. Let $\sysbar{\nact}$ be the objective value achieved by $\sysbar{\policy}$ in the single-server problem $\ssp((\arr_i r)_{i\in\jobspace}, \budget)$  with estimated parameters. Under $\policy$, for any initial state, we have
\begin{align}
    \abs{\sum_{\vek \neq \vzero} \E\left[X_{\vek}\right] - \ceil{\sysbar{\nact}} \cdot\Prob\left(\sysbar{\veK}\neq \vzero\right)}  &= O\big(\sqrt{r}\big)+ \delta \cdot O\big(r\big), \label{eq:misspec:active-servers}\\
    \abs{\sum_{\vek\neq\vzero} h(\vek) \E\left[X_{\vek}\right] - \ceil{\sysbar{\nact}} \cdot \E\left[h(\sysbar{\veK})\right]} &= O\big(\sqrt{r}\big)+ \delta \cdot O\big(r\big),\label{eq:misspec:cost}
\end{align}
where $\sysbar{\veK}$ denotes the steady-state configurations of the single-server system under $\sysbar{\policy}$ with the {estimated parameters}. 
If we let $\sysbar{\policy}$ be the optimal policy of the single-server problem with estimated parameters, for any initial state, we have
\begin{align}
    \nact(\policy) &\leq \left(1 + B\delta + \Obrac{r^{-0.5}}\right) \cdot \sysbar{\nact}^*_{\text{true}}, \label{eq:misspec:alpha-optimal}\\
    \costp(\policy) &\leq \left(1 + B\delta + \Obrac{r^{-0.5}}\right)\cdot \budget \label{eq:misspec:beta-optimal},
\end{align}
where $B$ is a positive constant independent of $r$. In other words, $\policy$ is $\left(1+ B \delta + O\left(r^{-0.5}\right), 1+ B \delta+O\left(r^{-0.5}\right)\right)$-optimal. 
\end{proposition}

\begin{remark}
    Note even if the single-server system with estimated parameters $\vek^0$-irreducible under $\sysbar{\policy}$, it is hard to guarantee that the original system is $\vek^0$-irreducible due to the estimation errors. 
    Consequently, the steady-state performance metrics $N(\policy)$ and $C(\policy)$ could depend on the system's initial state. 
    Fortunately, our proof for \Cref{theo:sat:conversion-general} does not rely on the uniqueness of the stationary distribution, so we can adapt the proof to show the inequalities in \Cref{theo:misspec} for any initial state. 
\end{remark}

\begin{remark}
Note that \Cref{assump:accurate-simulation} ensures that the real jobs and virtual jobs on a server are indistinguishable, so that $(\syshat{\veK}^\allL(t), \vetoK^\allL(t))$ is still a Markov chain. Recall that $\syshat{\veK}^{\ell}$ and and $\vetoK^{\ell}$ denote the observed configuration and tokens for each normal server $\ell$, respectively  (see \Cref{sec:conv-thm:preliminary}).
We suspect that \Cref{assump:accurate-simulation} can be removed since our proof does not rely too much on $(\syshat{\veK}^\allL(t), \vetoK^\allL(t))$ being a Markov chain. However, removing the assumption requires a more careful and notationally heavy analysis. 
We argue that in practice, this assumption is not restrictive, as one can record the traces of the jobs that arrived in the past and resample virtual jobs from those traces.  
\end{remark}

\subsection{Lemmas}
In this section, we give a proof sketch for \Cref{theo:misspec} when the single-server system under $\sysbar{\policy}$ \emph{with estimated parameters} is $\vek^0$-irreducible. The argument of extending to the general case is essentially the same as \Cref{sec:proof-conv-general}, so we omit it here.

On a high level,the proof of \Cref{theo:misspec} is similar to that of \Cref{theo:asymp-opt}. In particular, the key steps of the proof are to show that $d(\veK^\allL, \sysbar{\veK}^\allL) = O(\sqrt{r})$ and $\sum_{\ell=L+1}^\infty \sum_{i\in\jobspace} K_i^\ell=O(\sqrt{r})$ as $r\to\infty$, where $L = \lceil\sysbar{\nact}\rceil$, and $\sysbar{\veK}^\allL$ are i.i.d. copies of steady-state configurations of the single-server system under $\sysbar{\policy}$ with the \emph{estimated parameters}, and $\veK$ is the steady-state configurations of the infinite server system under ${\policy}$. 

\Cref{theo:misspec} is based on the three lemmas stated below. The proof of \Cref{theo:misspec} using the three lemmas is essentially the same as the argument in \Cref{sec:conv-thm:proof-thm}. 

\begin{lemma}\label{lem:misspec:w1-distance} Under the conditions of \Cref{theo:misspec} and the single-server system with estimated parameters under $\sysbar{\policy}$ being $\vek^0$-irreducible, for any initial state, we have     
\begin{equation*}
        d\left(\syshat{\veK}^\allL, \sysbar{\veK}^\allL \right) = O\big(\sqrt{r}\big)+ \delta \cdot O\big(r\big)
    \end{equation*}
\end{lemma}

\begin{lemma}\label{lem:misspec:virtual-jobs} Under the conditions of \Cref{theo:misspec} and the single-server system with estimated parameters under $\sysbar{\policy}$ being $\vek^0$-irreducible, for any initial state and $i\in\jobspace$, the steady-state expected number of virtual jobs of type $i$ s.t.
    \begin{equation*}
        \E\left[\textstyle \sumL \vj_i^\ell\right] = O\big(\sqrt{r}\big)+ \delta \cdot O\big(r\big).
    \end{equation*}
\end{lemma}

\begin{lemma}\label{lem:misspec:overflow-jobs} Under the conditions of \Cref{theo:misspec} and the single-server system with estimated parameters under $\sysbar{\policy}$ being $\vek^0$-irreducible, for any initial state and $i\in\jobspace$, the steady-state expected number of type $i$ jobs on backup servers s.t.
    \begin{equation*}
     \E\left[\textstyle \sum_{\ell=L+1}^\infty K_i^\ell\right]= O\big(\sqrt{r}\big)+ \delta \cdot O\big(r\big).
    \end{equation*}
\end{lemma}

These three lemmas are analogous to \Cref{lem:sat:w1-distance}, \Cref{lem:sat:virtual-jobs}, and \Cref{lem:sat:overflow-jobs}, respectively.
In the rest of the section, we sketch the proofs for the three lemmas above, highlighting the difference from the proofs of their analogues.

\subsection{Proof sketch for \Cref{lem:misspec:w1-distance}}
Recall from \Cref{appx:proof-sat-w1-distance} that the proof of \Cref{lem:sat:w1-distance} is based on Stein's method, which compares the generator of the i.i.d. copies of the single-server system with the generator of the infinite-server system. 
To write down the generators with the estimated model, recall that in the single-server system, each transition can be represented by the diagram
\[
\vek \overset{}{\rightarrow} \vek' \overset{}{\rightarrow} \vek' + \vea, 
\]
where the arrow $\vek\to\vek'$ denotes an internal transition or a departure if $\vek\neq \vek'$; the arrow $\vek'\to\vek'+\vea$ denotes a job request that is made right after reaching $\vek'$. For any $\vek$, let $\Kka(\vek)$ be the set of $(\vek',\vea)\in\Kouter^2$ such that $\vek'+\vea \in \Kouter$. We define two sets of transition rates as below: for any $\vek \in \Kouter$ and $(\vek',\vea)\in\Kka(\vek)$, 
\begin{itemize}
    \item Under the estimated parameters and the policy $\sysbar{\policy}$, we let $\rate{\vek', \vea}{\vek}$ be the rate of the transition $\vek\to\vek'\to\vek'+\vea$, and let $\ratetot{\vek} \triangleq \sum_{\vek',\vea} \rate{\vek',\vea}{\vek}$ be the total transition rate at configuration $\vek$. 
    \item Under the \emph{true parameters and the policy $\sysbar{\policy}$}, we let $\ratetjm{\vek',\vea}{\vek}$ be the rate of the transition $\vek\to\vek'\to\vek'+\vea$, and let $\ratetottjm{\vek}\triangleq \sum_{\vek',\vea} \ratetjm{\vek',\vea}{\vek}$ be the total transition rate at configuration $\vek$. 
\end{itemize}

Let $\sysbar{\gen}$ be the generator of $L = \lceil\sysbar{\nact}\rceil$ i.i.d. copies of single-server systems under the estimated parameters. For any $g\colon \Kouter^L \to \R$, we have
\begin{equation}
    \sysbar{\gen}g(\vek^\allL) = \sumL \sum_{\vek', \vea} \rate{\vek', \vea}{\vek^\ell} \left(g(\cdot, \vek'+\vea, \cdot) - g(\cdot, \vek^\ell, \cdot) \right)\label{eq:misspec:gen-bar},
\end{equation}
where $\sumka$ is a shorthand for $\sumkae$. 
Let $\syshat{\gen}$ be the generator of $(\syshat{\veK}^\allL(t), \vetoK^\allL(t))$ for the infinite-server system. For any function $g\colon \Kinner^L \to \R$ and $\psi(\vek^\allL, \vetok^\allL) = g(\vek^\allL+\vetok^\allL)$, we have
\begin{align}
    \syshat{\gen} \psi(\vek^\allL, \vetok^\allL)
    &= \sumL \sumka \ratetjm{\vek', \vea}{\vek^\ell} \left(g(\cdot, \vek'+\vea, \cdot) - g(\cdot, \vek^\ell, \cdot) \right) \indibrac{\vetok^\ell = \vzero} \nonumber\\
    &\mspace{23mu}+\sumL \sumka \ratetjm{\vek', \vea}{\vek^\ell} \left(g(\cdot, \vek'+\vetok^\ell, \cdot) - g(\cdot, \vek^\ell+\vetok^\ell, \cdot) \right)\indibrac{\vetok^\ell\neq \vzero}
    \label{eq:misspec:gen-hat}.
\end{align}

To prove \Cref{lem:misspec:w1-distance}, we need to show that for any $f\in \text{Lip}(1)$, $\vek^\allL \in \Kouter^L$, and $\vetok^\allL \in \Kouter^L$,
\begin{equation}\label{eq:misspec:w1-distance-proof:goal}
    \sysbar{\gen}g_f(\vek^\allL+\vetok^\allL) - \syshat{\gen} \psi_f(\vek^\allL, \vetok^\allL) = O\big(\sqrt{r}\big) + \delta \cdot O\big(r\big),
\end{equation}
where $g_f$ is the solution to
\begin{equation}\label{eq:misspec:lem1-poisson-eq}
    \E\left[f\big(\sysbar{\veK}^\allL\big)\right] - f\big(\vek^\allL\big) = \sysbar{\gen} g_f\big(\vek^\allL\big),
\end{equation}
and $\psi_f(\vek^\allL, \vetok^\allL) = g_f(\vek^\allL+\vetok^\allL)$. 
Same as the proof of \Cref{lem:sat:w1-distance}, we prove \eqref{eq:misspec:w1-distance-proof:goal} in two steps: the generator comparison step, and the Stein factor bound step. 

In the generator comparison step, we observe that the formula of $\sysbar{\gen}g(\vek^\allL)$ and $\syshat{\gen} \psi(\vek^\allL, \vetok^\allL)$ in \eqref{eq:misspec:gen-bar} and \eqref{eq:misspec:gen-hat} look almost the same as \eqref{eq:sat:gen-bar} and \eqref{eq:sat:gen-hat}, except that the rates in \eqref{eq:misspec:gen-hat} are $\ratetjm{\vek', \vea}{\vek^\ell}$ instead of $\rate{\vek', \vea}{\vek^\ell}$. As a result, after carrying out similar calculations as in the poof of \Cref{lem:sat:w1-distance}, we get an extra error term involving $\rate{\vek', \vea}{\vek^\ell}-\ratetjm{\vek', \vea}{\vek^\ell}$, which can be bounded using the lemma below. This error term results in the $\delta \cdot O(r)$ in \eqref{eq:misspec:w1-distance-proof:goal}.

\begin{lemma}\label{lem:misspec:gamma-rates-approx}
    Under \Cref{assump:delta-accurate}, 
    for any $\vek \in \Kouter$ and $(\vek',\vea) \in \Kka(\vek)$, we have 
    \begin{equation}
	\abs{\ratetjm{\vek',\vea}{\vek} - \rate{\vek',\vea}{\vek}} \leq \Kmax \delta.
    \end{equation}
\end{lemma}

\begin{proof}
	When $\vek = \vek'$, $\rate{\vek', \vea}{\vek}$  and $\ratetjm{\vek', \vea}{\vek}$ are both the rate of adding $\vea$ jobs via a proactive request under the policy $\sysbar{\policy}$, so they are identical. 
	
	When $\vek \neq \vek'$ and $\vea = \vzero$, $\rate{\vek', \vea}{\vek}$ is equal to the rate of going from $\vek$ to $\vek'$ via an internal transition or departure under the estimated job model, while $\ratetjm{\vek', \vea}{\vek}$ is under the true job model. Because there are at most $\Kmax$ jobs, and by our assumption the estimation error of each job's transition rates are bounded by $\delta$, thus $\rate{\vek', \vea}{\vek}$ and $\ratetjm{\vek', \vea}{\vek}$ differ by at most $\Kmax \delta$.
	
	When $\vek \neq \vek'$ and $\vea \neq \vzero$, $\rate{\vek', \vea}{\vek}$ is equal to the rate of going from $\vek$ to $\vek'$ via an internal transition or departure, multiplied by the probability of adding $\vea$ jobs via a reactive request, under the estimated job model and the policy $\sysbar{\policy}$; $\ratetjm{\vek', \vea}{\vek}$ is under the true job model instead of the estimated job model, but uses the same policy. Because the rate of going from $\vek$ to $\vek'$ differs by at most $\Kmax \delta$ under two different job models, and the probability of adding $\vea$ jobs after going from $\vek$ to $\vek'$ is the same under the same policy, thus $\rate{\vek', \vea}{\vek}$ and $\ratetjm{\vek', \vea}{\vek}$ differ by at most $\Kmax \delta$.
\end{proof}

In the Stein factor bound step, we need to show that 
\begin{equation}\label{eq:misspec:grad-moment-bound}
   \sup_{\vek, \vek'\in\Kouter} \abs{g_f(\cdot, \vek', \cdot) - g_f(\cdot, \vek, \cdot )} = \Obrac{1}.
\end{equation}
This involves analyzing the i.i.d. copies of the single-server system, and the fact that the single-server system with estimated parameters is $\vek^0$-irreducible under $\sysbar{\policy}$. This part of the proof is identical to the corresponding part of the proof of \Cref{lem:sat:w1-distance}.

\subsection{Proof sketch for \Cref{lem:misspec:virtual-jobs} and \Cref{lem:misspec:overflow-jobs}}

The proof of \Cref{lem:misspec:virtual-jobs} and \Cref{lem:misspec:overflow-jobs} has a similar structure as the proof of \Cref{lem:sat:virtual-jobs} and \Cref{lem:sat:overflow-jobs}. In the first step, we use Little's law to bound expectations of the number of type $i$ virtual jobs,  $\tvJ_i$, and the number of type $i$ jobs on backup servers, $\toF_i$. We have two equations almost the same as \eqref{eq:sat:tvj-from-dtvj} and \eqref{eq:sat:tof-from-dtof} except that the rates $\rate{\vek', \vea}{\syshat{\veK}^\ell}$ and $\lamest_i r$ are replaced by $\ratetjm{\vek', \vea}{\syshat{\veK}^\ell}$ and $\lamtrue_i r$:
\begin{align}
    \E[\tvJ_i] &\leq \servmax \E\left[\sumL \sumka \ratetjm{\vek', \vea}{\syshat{\veK}^\ell} \dtvjg_i \indibrac{\vetoK^\ell=\vzero} \right], \label{eq:misspec:tvj-from-dtvj}  \\
    \E[\toF_i] &\leq \servmax \E\left[\lamtrue_i r \cdot \dtof_i \right], \label{eq:misspec:tof-from-dtof}
\end{align}
where $\servmax$ is the maximal expected service time of any type of virtual job or real job. Because $\servmax = O(1)$, it suffices to bound $\E\left[\sumL \sumka \ratetjm{\vek', \vea}{\syshat{\veK}^\ell} \dtvjg_i \indibrac{\vetoK^\ell=\vzero} \right]$ and  $\E\left[\lamtrue_i r \cdot \dtof_i \right]$. 

In the second step, we utilize the fact that the two Lyapunov functions $g(\ttk_i) = \ttk_i$ and $g(\ttk_i) = \ttk_i^2$ have zero drift in steady-state, where $\ttk_i$ is the total number of type $i$ tokens. We get two equalities similar to \eqref{eq:sat:ttk-first-order} and \eqref{eq:sat:ttk-second-order-term2}: 
\begin{equation}\label{eq:misspec:ttk-first-order}
    \E\left[\sumL \sumka \ratetjm{\vek', \vea}{\syshat{\veK}^\ell} (a_i - \dtvjg_i)\indibrac{\vetoK^\ell=\vzero} + \lamtrue_i r (-1 + \dtof_i) \right] =  0.
\end{equation}
\begin{align}
    &\mspace{23mu} \E\left[\sumL \sumka \ratetjm{\vek', \vea}{\syshat{\veK}^\ell} \dtvjg_i \indibrac{\vetoK^\ell=\vzero}\right]  \label{eq:misspec:ttk-second-order} \\
    &= \frac{1}{\toklim} \cdot \E\left[\left(\sumL \sumka \ratetjm{\vek', \vea}{\syshat{\veK}^\ell} a_i  \indibrac{\vetoK^\ell=\vzero} - \lamtrue_i r \right) \cdot \ttK_i\right] \label{eq:misspec:ttk-second-order-term1} \\
    &\mspace{2mu}+ \frac{1}{2\toklim}  \cdot \E\left[\sumL \sumka \ratetjm{\vek', \vea}{\syshat{\veK}^\ell} \left(a_i^2 - (\dtvjg_i)^2\right)\indibrac{\vetoK^\ell=\vzero} + \lamtrue_i r \cdot (1-(\dtof_i)^2)\right]. \label{eq:misspec:ttk-second-order-term2}
\end{align}

In the third step, we use the above two equalities to bound $\E\left[\sumL \sumka \ratetjm{\vek', \vea}{\syshat{\veK}^\ell} \dtvjg_i \indibrac{\vetoK^\ell=\vzero} \right]$ and $\E\left[\lamtrue_i r \cdot \dtof_i \right]$. We first use the equality in \eqref{eq:misspec:ttk-second-order} to \eqref{eq:misspec:ttk-second-order-term2} to bound $\E\left[\sumL \sumka \ratetjm{\vek', \vea}{\syshat{\veK}^\ell} \dtvjg_i \indibrac{\vetoK^\ell=\vzero} \right]$. Following the same argument as the proof of \Cref{lem:sat:virtual-jobs} and \Cref{lem:sat:overflow-jobs} until \eqref{eq:sat:step-2-pf:drift-term-bound-combine}, we can show that
\[
    \eqref{eq:misspec:ttk-second-order-term2} \leq O\big(\sqrt{r}\big),
\]
\begin{align}
    \eqref{eq:misspec:ttk-second-order-term1} 
    &\leq 
    \E\left[\abs{\sumL \sumka \ratetjm{\vek', \vea}{\syshat{\veK}^\ell} a_i \indibrac{\vetoK^\ell=\vzero}  - \lamtrue_i r }\right] \nonumber \\
    &\leq \E\left[\abs{\sumL \sumka \ratetjm{\vek', \vea}{\syshat{\veK}^\ell} a_i  - \lamtrue_i r }\right] + O\big(\sqrt{r}\big)\nonumber \\
    &\leq \E\left[\abs{\sumL \sumka \ratetjm{\vek', \vea}{\sysbar{\veK}^\ell} a_i  - \lamtrue_i r }\right] + O\big(\sqrt{r}\big)+ \delta \cdot O\big(r\big), \label{eq:misspec:step-2-pf:drift-term-bound-apply-lemma8}
\end{align}
where to get \eqref{eq:misspec:step-2-pf:drift-term-bound-apply-lemma8}, we apply \Cref{lem:misspec:w1-distance} to replace $\syshat{\veK}^\ell$ with $\sysbar{\veK}^\ell$, which causes an $O(\sqrt{r}) + \delta \cdot O(r)$ error. 
Next, we show that
\begin{equation}\label{eq:misspec:step-2-pf:drift-term-intermediate-2}
    \E\left[\abs{\sumL \sumka \ratetjm{\vek', \vea}{\sysbar{\veK}^\ell} a_i  - \lamtrue_i r }\right] = O\big(\sqrt{r}\big)+ \delta \cdot O\big(r\big).
\end{equation}
By \Cref{assump:delta-accurate} and \Cref{lem:misspec:gamma-rates-approx}, we have
\begin{align}
    \label{eq:misspec:step-2-pf:arr-rate-approx}
    \abs{\lamtrue_i r - \lamest_i r} &\leq \delta r,
    \\
    \label{eq:misspec:step-2-pf:req-rate-approx}
    \Big\lvert \sumL \sumka \ratetjm{\vek', \vea}{\sysbar{\veK}^\ell} a_i 
    -  \sumL \sumka \rate{\vek', \vea}{\sysbar{\veK}^\ell} a_i \Big\rvert &\leq \delta \cdot O\big(r\big).
\end{align}
These two bounds allow us to replace $\ratetjm{\vek', \vea}{\sysbar{\veK}^\ell}$ and $\lamtrue_i r$ on the LHS of \eqref{eq:misspec:step-2-pf:drift-term-intermediate-2} with $\rate{\vek', \vea}{\sysbar{\veK}^\ell}$ and $\lamest_i r$ at the cost of introducing $\delta \cdot O(r)$ error. 
Moreover, because $\{\sysbar{\veK}^\ell\}_{\ell=1,\dots,L}$ are i.i.d., $\sumL \sumka \rate{\vek', \vea}{\sysbar{\veK}^\ell} a_i$ concentrates around its mean with $O(\sqrt{r})$ error, where the mean can be shown to be
\begin{equation*}
    \E\left[\sumL \sumka \rate{\vek', \vea}{\sysbar{\veK}^\ell} a_i\right] = \reqrate_i L = \lamest_i r + O(1).
\end{equation*}
Note that the first equality in \eqref{eq:misspec:step-2-pf:mean-of-sth-is-reqrate} holds because for each $\ell$, $\E[\sumka \rate{\vek', \vea}{\sysbar{\veK}^\ell} a_i]$ is equal to $\reqrate_i$, i.e., the long-run average rate of requesting type $i$ jobs on a single-server system with estimated parameters under $\sysbar{\policy}$; the second equality in \eqref{eq:misspec:step-2-pf:mean-of-sth-is-reqrate}  is because $L = \lceil \sysbar{\nact} \rceil$, and $\sysbar{\nact} \reqrate_i = \lamest_i r$. As a result,
\begin{equation}\label{eq:misspec:step-2-pf:mean-of-sth-is-reqrate}
     \E\left[\abs{\sumL \sumka \rate{\vek', \vea}{\sysbar{\veK}^\ell} a_i - \lamest_i r} \right] \leq O\big(\sqrt{r}\big).
\end{equation}
Combining \eqref{eq:misspec:step-2-pf:arr-rate-approx} to \eqref{eq:misspec:step-2-pf:mean-of-sth-is-reqrate}, we get \eqref{eq:misspec:step-2-pf:drift-term-intermediate-2}. Therefore, 
\begin{align*}
    \E[\tvJ_i] &\leq \servmax \E\left[\sumL \sumka \ratetjm{\vek', \vea}{\syshat{\veK}^\ell} \dtvjg_i \indibrac{\vetoK^\ell=\vzero} \right] \\
    &\leq O(1)\cdot \big(\eqref{eq:misspec:ttk-second-order-term1} 
    + \eqref{eq:misspec:ttk-second-order-term2}\big)\\
    &\leq O\big(\sqrt{r}\big)+ \delta \cdot O\big(r\big),
\end{align*}
which completes the proof of \Cref{lem:misspec:virtual-jobs}.

Finally, it is straightforward to show that $\E\left[\lamtrue_i r \cdot \dtof_i \right] = O(\sqrt{r}) + \delta \cdot O(r)$ using \eqref{eq:misspec:ttk-first-order} and the bound on $\E\left[\sumL \sumka \ratetjm{\vek', \vea}{\syshat{\veK}^\ell} \dtvjg_i \indibrac{\vetoK^\ell=\vzero} \right]$. Therefore, $\E[\toF_i] \leq \servmax \E\left[\lamtrue_i r \cdot \dtof_i \right]  = O(\sqrt{r}) + \delta \cdot O(r)$. This proves \Cref{lem:misspec:overflow-jobs}.

\end{document}